%% file: main.tex
\DocumentMetadata{}
\documentclass[sigconf,authorversion,nonacm]{acmart}
\usepackage{subcaption}
\usepackage{bm}
\usepackage{mathtools}

\usepackage{algorithm}
\usepackage[noend]{algpseudocode}
\algnewcommand\algorithmicinput{\textbf{Input:}}
\algnewcommand\algorithmicoutput{\textbf{Output:}}
\algnewcommand\Input{\item[\algorithmicinput]}%
\algnewcommand\Output{\item[\algorithmicoutput]}

\input{0_custom_declares}

\begin{document}

\title{Differentially Private Synthetic High-dimensional Tabular Stream}

%%%%%%%%%%%%%%%% Authors' Info %%%%%%%%%%%%%%%%%

\author{Girish Kumar}
\email{gkum@ucdavis.edu}
\affiliation{%
  \institution{Department of Mathematics, University of California}
  \city{Davis}
  \state{CA}
  \country{USA}
}

\author{Thomas Strohmer}
\email{strohmer@math.ucdavis.edu}
\affiliation{%
  \institution{Department of Mathematics, University of California}
  \city{Davis}
  \state{CA}
  \country{USA}
}

\author{Roman Vershynin}
\email{rvershyn@uci.edu}
\affiliation{%
  \institution{Department of Mathematics, University of California}
  \city{Irvine}
  \state{CA}
  \country{USA}
}

\renewcommand{\shortauthors}{Kumar et al.}

\begin{abstract}
  While differentially private synthetic data generation has been explored extensively in the literature, how to update this data in the future if the underlying private data changes is much less understood.  We propose an algorithmic framework for streaming data that generates multiple synthetic datasets over time, tracking changes in the underlying private data. Our algorithm satisfies differential privacy for the entire input stream (continual differential privacy) and can be used for high-dimensional tabular data. Furthermore, we show the utility of our method via experiments on real-world datasets. The proposed algorithm builds upon a popular {\em select, measure, fit, and iterate} paradigm (used by offline synthetic data generation algorithms) and private {\em counters} for streams.
\end{abstract}


\maketitle

\input{1_intro}

\input{2_setup}

\input{3_preliminaries}
\input{3_offline_synth_data}
\input{4_baseline}
\input{4_method}
\input{5_experiments}

\input{6_unbounded_block_counter}

\section{Conclusion}
In this work, we discuss the task of streaming differentially private high-dimensional synthetic data that accurately represents the true data over a set of marginal queries. Our focus is on developing an algorithm that can be used in practical applications and is better than the naive algorithm of running independent instances of offline algorithms on differential data stream at any time. We build upon existing research in offline synthetic data generation and counters for streaming algorithms to create a framework for the task. We also show with experiments over real-world datasets that our method outperforms the baseline. 

%============================ End of sections ==========

\begin{acks}
G.K.\ and T.S.\ acknowledge support from NSF DMS-2027248, NSF DMS-2208356, and NIH R01HL16351.  R.V.\ acknowledges support from NSF DMS-1954233, NSF DMS-2027299, and NSF+Simons Research Collaborations on the Mathematical and Scientific Foundations of Deep Learning.
\end{acks}

% \clearpage

\bibliographystyle{ACM-Reference-Format}
\bibliography{intro_bib,main_bib}

%============================ Start appendices ==========

\appendix
\onecolumn

\input{appendices/app_acc_baseline}
\input{appendices/app_acc_proposed_algo}

\end{document}

%% file: 0_custom_declares.tex
\usepackage{dsfont, xcolor, commath}
\usepackage{multirow}
\usepackage{longtable}

\DeclareMathOperator{\Lap}{Lap}

\DeclareMathOperator{\avg}{avg}
\DeclareMathOperator*{\argmax}{arg\,max}
\DeclareMathOperator*{\argmin}{arg\,min}
\DeclareMathOperator{\maxerr}{maxerr}
\DeclareMathOperator{\adderr}{adderr}

\newcommand{\inparanth}[1]{\left( #1 \right)}
\newcommand{\bp}[1]{\left( #1 \right)}
\newcommand{\prob}[1]{\mathbb{P} \left\{ #1 \right\}}
\newcommand{\bigo}[1]{\mathcal{O}\left( #1 \right)}

\def \npx {\mathbb{N}_0^{\mathcal{X}}}

\def \N {\mathbb{N}}
\def \P {\mathbb{P}}
\def \R {\mathbb{R}}

\def \AA {\mathcal{A}}
\def \CC {\mathcal{C}}

\def \RR {\mathcal{R}}

\def \XX {\mathcal{X}}

\def \a {\alpha}
\def \b {\beta}

\def \e {\varepsilon}
\def \d {\delta}

\def \ind {{\mathds 1}}

\newcommand\restr[2]{{\left.\kern-\nulldelimiterspace{#1}\right|_{#2}}}
  

%% file: 1_intro.tex
\section{Introduction}
Data availability has become crucial to technological advancement in today's world. Publicly available datasets have the additional advantage that people across various domains, affiliations, and geographic locations can contribute to the research on such datasets. However, in many critical domains such as healthcare, public policy, and market research, it is challenging to release curated datasets publicly without compromising user privacy. A potential solution explored in many previous works is the release of a synthetic dataset that mimics some relevant statistical properties of the private data. Differential privacy has emerged as a standard notion to provide a mathematical guarantee that the synthetic data preserves the privacy of individuals contributing to the private data. Achieving differential privacy is non-trivial and typically requires adding carefully calibrated noise to measurements of true data. Many existing works in the differential privacy literature have explored the task of generating synthetic data such as \cite{tabular_pgm_mckenna19a,aydore_rap,pate_gan,tao2021benchmarking,dp_cgan,xu2019ganobfuscator,yue2022synthetic,rapporGoogle}.

Most of the research in differential privacy has focused on generating the synthetic dataset once, based on all private data available at the time. However, in many real-world scenarios, the private data may change over time resulting in a requirement to update the synthetic data with time as well. For example, consider the electronic health records of patients admitted to the hospital in the middle of a pandemic such as COVID-19. The availability of public data which can be updated over time will allow research developments in real time.

Motivated by this scenario, in this work, we are interested in developing an algorithm to generate a privacy-preserving synthetic high-dimensional {\em tabular stream}. A high-dimensional tabular data is where each record in the dataset consists of (say) $p$ fixed attributes, where $p$ is sufficiently large. A {\em stream} is a collection of such datasets over time. Moreover, our algorithm is streaming (online) in the sense that we generate the updated synthetic dataset at each time, only using the private stream until that time. To preserve privacy, we will use the notion of {\em continual differential privacy}, which extends the concept of (offline) differential privacy to streaming algorithms. We discuss these terms and the setup rigorously in Section~\ref{s:intro}.

\section{Overview and Related Work}

The notion of differential privacy for streaming algorithms, such that the privacy guarantee spans the entire time horizon, was introduced in the seminal work of  \cite{Chan2010ContinualPrivateStats} and \cite{Dwork2010ContinualDP}. They also introduced the concept of {\em counters}, which are differentially private streaming algorithms that can efficiently count the occurrence of an event over an input stream. Counters have been used as building blocks of many streaming algorithms such as  \cite{cdp_wang2021continuousStream,cdp_chen2017pegasus,cdp_jain2023countingDistinct}. Our proposed method also uses counters as sub-routines.

However, differential privacy has not been explored as much for other complicated streaming tasks such as synthetic data generation. To the best of our knowledge \cite{cdp_bun2023continualSynthData}, \cite{He2024OnlineDP}, and \cite{kumar_algorithm_2024} are the only other works that explore differentially private synthetic data generation with streaming algorithms. \cite{cdp_bun2023continualSynthData} approach a different problem of generating synthetic streams for a fixed universe of users such that each user contributes at all times. Moreover, they assume that each record in data is a boolean. Both \cite{He2024OnlineDP} and \cite{kumar_algorithm_2024} provide an algorithm that works with a hierarchical decomposition of the data space. \cite{He2024OnlineDP} provides superior theoretical guarantees of their proposed algorithm but does not provide any experimental evaluation. \cite{kumar_algorithm_2024} limits the experiments to low-dimensional domains such as spatial streams. However, algorithms based on hierarchical decomposition typically do not scale well for high-dimensional data as the number of nodes in the tree grows exponentially with dimension leading to large time complexity and poor utility. To that end, this work is the first to provide a differentially private streaming algorithm for synthetic data generation that is tractable for real-world datasets and provides better utility than the trivial baseline. In the discussion that follows, we give an overview of our method and discuss how research in offline differential privacy has motivated it.

There is a plethora of research in differential privacy on offline algorithms generating synthetic tabular datasets such as \cite{tabular_pgm_mckenna19a,aydore_rap,pate_gan,hardt2012mwem,zhang2017privbayes,tao2021benchmarking,mckenna2021winning,mckenna2021hdmm,Liu2021IterativeMF,dual_query}. While they cannot be directly applied for our use case, they are the motivation behind our proposed method. First, similar to most of these works, we measure the quality of the synthetic stream using marginal queries. A marginal query counts the number of instances in a dataset where a particular combination of the values of some attributes occurs. For the previously mentioned data space of electronic health records, a marginal query may be - ``how many patients are more than 50 years old, have been previously diagnosed with Asthma, and have tested positive for COVID-19''? In this work, we are interested in a pre-defined set of marginal queries and we target that at any time the synthetic stream has almost the same value for any query as the true stream. 

Second, similar to most of the offline algorithms mentioned earlier, we use the {\em select, measure, fit, and iterate} paradigm. In this approach, the algorithm iteratively selects a query (typically the worst-performing query), measures it, and fits the data-generating model according to this noisy measurement. An early work that used this paradigm is MWEM \cite{hardt2012mwem} which combined the Multiplicative Weights algorithm with the Exponential Mechanism (a differentially private selection algorithm) to generate synthetic data. The MW algorithm directly maintains an estimated distribution on the entire data space and adjusts the distribution to comply with the noisy measurements of the selected marginal queries. The algorithm thus quickly becomes intractable for reasonably high dimensions of data. MWEM is typically the best-performing algorithm for low data dimensionality \cite{Liu2021IterativeMF}. In this work, we build upon MWEM+PGM \cite{mckenna2021winning}, a scalable version of the MWEM approach that replaces modeling the data distribution from MW to a Probabilistic Graphical Model (PGM), and has been shown to perform well for a variety of real-world datasets \cite{tao2021benchmarking}.

A related task to our problem is generating answers to histogram queries for a stream. This problem has been explored in works such as \cite{henzinger_differentially_2023} and \cite{jain_price_2023}. However, they limit the data space to the set $\set{0,1}^d$ (for some dimension $d$) and the histogram queries to the column sum. Moreover, while the accuracy scales logarithmically in time, it scales linearly with $d$. In contrast, our method outperforms a baseline streaming algorithm, which exhibits accuracy scaling poly-logarithmically with the number of queries. Thus, it is more practical for answering $k$-way marginal queries for high-dimensional datasets.

There exists a simple baseline differentially private streaming algorithm for our problem - run an independent instance of an offline algorithm (such as MWEM+PGM) on the differential data at any time to create the stream \cite{kumar_algorithm_2024}. We show in our experiments that our proposed method outperforms such a baseline. In principle, our method has two key differences: (1) we use counters as a sub-routine to measure any query over time, and (2) at any time, we recycle some information about the queries from the synthetic stream generated so far. The contributions of this work are summarized below:
\begin{enumerate}
    \item we present a novel framework that extends the {\em select, measure, fit, and iterate} paradigm from offline to streaming algorithms for synthetic data generation;
    \item we give a theoretical accuracy guarantee for the baseline algorithm;
    \item furthermore, we demonstrate using experiments that our proposed method outperforms the baseline on real-world high-dimensional streams.
\end{enumerate}

%% file: 2_setup.tex
\section{Problem setup and useful tools}\label{s:intro}

\subsection{Notation}
Let $\N_0=\N\cup\set{0}$. Let $[n]$ denote the set $\{1, 2, \ldots, n\}$ for any $n\in\N$. Let $\XX = \XX_1 \times \XX_2 \times \ldots \times \XX_p$ denote a space of $p$-dimensional points such that $\XX_i$ be a finite set of cardinality $|\XX_i|$ for any $i \in [p]$. For example, if each data point results from a survey of $p$ boolean questions, then $\XX=\{0,1\}^p$. We refer to any mapping of the form $h:\XX\to\N_0$ as a {\em dataset} such that $h(x)$ represents the number of times a point $x\in\XX$ appears. We denote by $|h|$ the total number of points in the dataset, that is $|h| = \sum_{x\in\XX}h(x)$.

In this work, we are interested in a dataset that changes with time, resulting in a data {\em stream}. Let $f: \XX\times\N \to \N_0$ be an input stream where $f(x, t)$ denotes the number of instances of point $x$ at time $t$. For example, consider that $\XX$ is the set of all possible demographics in a country's voting population. Then we can represent the eligible voters in the country as a stream $f:\XX\times\N\to\N_0$ where $f(x,t)$ denotes the number of individuals in the population with demographics $x\in\XX$ that are eligible to vote at time $t$. We provide a streaming algorithm that generates a privacy-preserving synthetic data stream $g: \XX\times\N \to \N_0$ such that $g$ accurately represents the input stream $f$. We define the terms ``streaming algorithm", ``privacy-preserving", and ``accurately" rigorously in the subsequent subsections.

For any $N \subseteq \N$, we will use the notation $f_N$ to denote the restriction of the stream $f$ to the time indices in the set $N$, that is $f_N : \XX \times N \to \N_0$ such that $f_N(x,t)=f(x,t)$ for all $t \in N$ and $x \in \XX$. Similarly, for any time $t \in \N$, $f_t: \XX \to \N_0$ denotes a restriction of $f$ to time $t$ such that $f_t(x) = f(x,t)$ for all $x \in \XX$. 

We have assumed that $\XX_i$ is a finite set for all $i\in[p]$. While this assumption may not hold in practice, we can {\em discretize} a continuous space by spending some of the privacy budget to create a differentially private histogram and mapping each point to the histogram bin \cite{tao2021benchmarking}.

\subsection{Streaming algorithm} 
As a motivating example, consider an algorithm $\AA$ that converts a given input data stream $f:\XX\times\N\to\N_0$ into a synthetic data stream $g:\XX\times\N\to\N_0$. The key idea of a streaming algorithm (Definition~\ref{def:streaming_algo}) is that at each time $t$, it can only see the part of the input stream $f(x,t')$ for all $x \in \XX$ and $t' \le t$, and at that time the algorithm outputs $g(x,t)$ for all $x \in \XX$.

\begin{definition}[Streaming algorithm]\label{def:streaming_algo}
    Given an input stream $f:\XX\times\N\to\N_0$, an algorithm $\AA$ with output stream $g:\XX\times\N\to\N_0$ is said to be streaming if at any time $t \in \N$ it maps $f_{[t]}$ to $g_t$, that is, $g(t, x)\coloneqq \AA(f_{[t]})(t, x)$ for all $x \in \XX$.
\end{definition}

\subsection{Differential Privacy}	\label{s:dp_streaming_algorithm}
To rigorously define privacy, we use the notion of Differential Privacy \cite{dwork2006differential} but for a streaming algorithm. Intuitively, differential privacy ensures that the output of an algorithm does not depend extensively on any particular user's data by ensuring robustness in the output probability distribution against change in a single data point. To this end, we need a concept of streams that differ in a single data point. Borrowing from \cite{kumar_algorithm_2024}, we first provide definitions of differential and neighboring streams and finally the extension of differential privacy to streaming algorithms.

\begin{definition}[Differential stream]
    A {\em differential stream} for $f$ is the stream $\nabla f$ defined as,
    \begin{equation}	\label{eq: differential stream}
    	\nabla f(x,t) \coloneqq f(x,t)-f(x,t-1), \quad t \in \N,
    \end{equation}
    where we set $f(x,0)=0$. The total change of $f$ over all times and points is the quantity
    \begin{equation}
        \norm{f}_\nabla \coloneqq \sum_{x \in \XX} \sum_{t \in \N} \abs{\nabla f(x,t)},    
    \end{equation}
    which defines a seminorm on the space of data streams.
\end{definition}

\begin{definition}[Neighboring streams]\label{def:neighboring_streams}
    Two streams $f$ and $\tilde f$ are said to be neighbors if 
    \begin{equation}	\label{eq:neighboring_stream}
        \norm[0]{f-\tilde{f}}_\nabla = 1,
    \end{equation}    
\end{definition}
that is if $\tilde{f}$ can be obtained from $f$ by changing the count of a single data point at some particular time. 

\begin{definition}[Differential privacy (for streams)] \label{def:dp_stream}
  A randomized streaming algorithm $\AA$ that takes data streams as input is $\e$-differentially private if for any two neighboring streams $f$ and $\tilde f$ that satisfy $\norm[0]{f-\tilde{f}}=1$, the inequality
  \begin{equation} \label{eq:dp_stream}
  \P\{\AA(\tilde{f}) \in S\} \le e^\e \cdot \P\{\AA(f) \in S\}
  \end{equation}
  holds for any measurable set of outputs $S$.
\end{definition}

\subsection{Accuracy over marginal queries}
In this work, we measure the accuracy of our output stream using marginal queries. A marginal query is a low-dimensional counting query. For example, suppose we have a census-like dataset where some of the demographics are age, marital status, and income. An example marginal query on these attributes is -  how many individuals in the dataset are age $30$, never married, and have income more than $\$100,000$ per annum. We define marginal query more formally in Definition~\ref{def:marginal_query}.

\begin{definition}[k-way marginal query] \label{def:marginal_query}
  A $k$-way marginal query $q:\XX \to \{0,1\}$ is a mapping defined by a tuple $(c_1, c_2, \ldots, c_k)$ of $k$ attribute indices and their corresponding values $(v_1, v_2, \ldots, v_k)$ such that $v_i \in \XX_{c_i}$ for all $i \in [k]$ and the mapping is defined as,
  \begin{equation}\label{eq:marginal_query_for_point}
      q(x) = \prod_{i=1}^k \left(\ind_{\{{x_{c_i}=v_i\}}} \right),
  \end{equation}
  for any $x \in \XX$. With a slight abuse of notation, we extend the definition of the marginal query $q$ from points to datasets as,
  \begin{equation}\label{eq:marginal_query_distribution}
      q(h) = \sum_{x \in \XX} h(x) q(x),
  \end{equation}
  for any dataset $h: \XX \to \N_0$.
\end{definition}

\subsubsection{Accuracy}
We first discuss the accuracy of offline algorithms and then extend it to streaming algorithms. The quality of an algorithm generating a synthetic dataset is typically measured by the aggregate performance of the resulting synthetic data over a predefined set of marginal queries (say) $Q$. We define the accuracy formally in Definition~\ref{def:accuracy_dataset}.

\begin{definition}[Accuracy of an algorithm generating synthetic dataset] \label{def:accuracy_dataset}
  Let $\AA$ be a randomized algorithm that maps an input dataset $f:\XX\to\N_0$ to a synthetic dataset $g: \XX\to\N_0$. Then, for any $\b>0$, $\AA$ is said have an accuracy of $(\a, \b)$, with respect to a set of marginal queries $Q$, if 
  \begin{equation}\label{acc_dataset}
      \prob{\max_{q \in Q} \abs{q(g)-q(f)} \geq \a} \leq \b,
  \end{equation}
  with probability taken over the randomness of algorithm $\AA$.
\end{definition}

A natural extension of Definition~\ref{def:accuracy_dataset}, to streams can be created by restricting the stream to any time $t\in\N$ and looking at the accuracy of the dataset present at that time. 

\begin{definition}[Accuracy of a streaming algorithm] \label{def:accuracy_stream}
  Let $\AA$ be a randomized streaming algorithm that maps an input stream $f:\XX\times\N\to\N_0$ to a synthetic stream $g: \XX\times\N\to \N_0$. Then, at any time $t\in\N$ and for any $\b>0$, $\AA$ is said have an accuracy of $(\a, \b)$, with respect to a set of marginal queries $Q$, if 
  \begin{equation}\label{acc_dataset}
      \prob{\max_{q \in Q} \abs{q(g_t)-q(f_t)} \geq \a} \leq \b,
  \end{equation}
  where the probability is taken over the randomness of the algorithm $\AA$. Here, $\a$ may be a function of $\b$ and $t$.
\end{definition}

%% file: 3_preliminaries.tex
\subsection{Exponential Mechanism}\label{s:dp_algos_selection}
Let $\RR$ be a finite set. Let $u: \npx \times \RR \to \R$ be a function such that $u(h, r)$ denotes the utility of an element $r \in \RR$ for a dataset $h \in \npx$. Our task is to find an element in $\RR$ with maximum utility while preserving differential privacy. Note that the term utility is very general and its exact definition is governed by the problem. As an example, suppose we want to find a mode of a given dataset $h\in\npx$. The mode is any point $x_* \in \XX$ such that $x_*=\argmax_{x\in\XX}h(x)$. We can use the exponential mechanism in this case with the utility being the absolute difference between the frequency of any point from the maximum possible frequency in $h$. Thus in this case $\RR = \XX$, and $u(h, x;x_*) = \abs{h(x)-h(x_*)}$ hor all $x\in\XX$.

\begin{definition}[Exponential mechanism]\label{def:exp_mech}
    Let $\Delta_u$ denote the sensitivity of the utility function defined as,
    \begin{equation}\label{eq:exp_mech_sensitivity}
        \Delta_u \coloneqq \max_{ \substack{ h, \tilde h \in \npx \\ \norm[0]{h-\tilde h} =1 } } \max_{r \in \RR} \abs{ u(h, r) - u(\tilde h, r)}.
    \end{equation}
    Then, the exponential mechanism is defined as the algorithm $\AA: \npx \to \RR$ such that, for all $r \in \RR$,
    $$
        \prob{\AA(h)=r} \propto \exp{ \bp{ -\frac{\e}{2 \Delta_u} u(h,r) } }.
    $$
\end{definition}

\begin{theorem}[Privacy of exponential mechanism]\label{thm:exp_mech_privacy}
    The exponential mechanism, as defined in Definition~\ref{def:exp_mech}, satisfies $\e$-differential privacy.
\end{theorem}

\begin{theorem}[Accuracy of exponential mechanism]\label{thm:exp_mech_accuracy}
    For any $\b>0$, the exponential mechanism, as defined in Definition~\ref{def:exp_mech}, satisfies
    \begin{equation*}
        \prob{ u(h,\AA(h)) \leq u_{OPT} - \frac{2\Delta_u}{\e} \ln\bp{{\frac{|\RR|}{\b}} } }
        \leq \b,
    \end{equation*}
    where $u_{OPT} = \max_{r\in\RR} u(h,r)$ denotes the maximum possible utility.
\end{theorem}

% ========================================
% ============= Counters
% ========================================

\subsection{Counters}
We borrow Definition~\ref{def:counter} from \cite{kumar_algorithm_2024}. Intuitively, it is an algorithm that estimates the sum of a stream with a certain accuracy.

\begin{definition}[Counter]\label{def:counter}
    An $(\a, \d)$-accurate counter $\CC$ is a randomized streaming algorithm that estimates the sum of an input stream of values $f:\N \to\R$ and maps it to an output stream of values $g: \N\to\R$ such that at any time $t\in\N$, 
    $$\mathbb{P} \biggl\{ \bigg\lvert g(t)- \sum_{t'\leq t} f(t') \bigg\rvert \leq \a(t, \d) \biggr\} \geq 1-\d,$$
    where the probability is over the randomness of $\CC$ and $\d$ is a small constant.
\end{definition}

\cite{Chan2010ContinualPrivateStats} and \cite{Dwork2010ContinualDP} first introduced differentially private counters that efficiently find the sum of a bit stream. The algorithms proposed in these works satisfy Definition~\ref{def:counter}. In this work we will be using the Simple II, Two-Level, and Hybrid Mechanism from \cite{Chan2010ContinualPrivateStats}, hereafter referred to as Simple, Block, and Binary Tree Counters respectively. As explained in \cite{Chan2010ContinualPrivateStats}, the key principle behind the design of these algorithms is dividing the time horizon into intervals and adding together the {\em noisy partial sums} from these intervals. For a fixed failure probability, the simple, block and binary tree counters are $\bigo{\sqrt{t}}$, $\bigo{t^{1/4}}$, and $\bigo{\inparanth{\ln{t}}^{3/2}}$ accurate respectively.

%% file: 3_offline_synth_data.tex
\section{Offline tabular synthetic dataset generation}
Our streaming algorithm uses many ideas from offline algorithms in the literature for dataset generation. Let us first discuss these ideas. Consider the task of generating synthetic tabular data when the dataset is available at once (offline). Let $f:\XX \to \N_0$ be a dataset. Assume we are interested in generating a synthetic dataset $g:\XX\to\N_0$ that is accurate for marginal queries $Q$. A straightforward approach to generating the synthetic dataset would be: (1) generate a differentially private measurement for all queries using the Laplace Mechanism as
$$
m = \bp{ q(f) + \Lap \bp{ \frac{\Delta q}{\e / |Q|} } }_{q\in Q};
$$ 
(2) find a dataset that minimizes the maximum error over the query set by solving the following optimization problem,
\begin{equation}\label{eq:minimization_problem}
    \argmin_{g \in \npx} \max_{q \in Q}  \abs{ m_q - q(g) }.
\end{equation}

There are however two key problems with this approach: (1) the size of the query set is typically polynomial in the dimension $p$ which leads to a very small budget for answering an individual query, that is a large amount of noise is added in Laplace Mechanism, and (2) the optimization problem in equation~\eqref{eq:minimization_problem} is a high-dimensional discrete optimization problem which is NP-Hard and cannot be solved in time polynomial in dimension $p$. Many existing algorithms thus circumvent the above two problems by: (1) measuring only a subset of the queries in $Q$ which have the largest error, and (2) approximating the optimization problem in Equation~\ref{eq:minimization_problem}. Let us first assume that we have a way to model the data distribution given the noisy measurements. Let $\AA_{Dataset}$ be one such subroutine that takes as input noisy measurements of the queries in $Q$ and an initialization dataset (say) $h\in\npx$, to provide a dataset $g\in\npx$ that complies with the measurements. In Section~\ref{s:select_measure_iterate}, we use $\AA_{Dataset}$ as a black-box subroutine and discuss how to iteratively select and measure a subset of the queries. In Section~\ref{s:model_distribution}, we then look at some of these methods for creating the dataset given a value of the queries.

\subsection{The {\em Select, Measure, Fit, and Iterate} paradigm} \label{s:select_measure_iterate}

\begin{algorithm}[t]
    \begin{algorithmic}[1]
        \State {\bfseries Input:} Given dataset $f$, an ordered set of queries $Q$, privacy budget $\e$, a differentially private selection mechanism $\AA_{Select}$, a subroutine $\AA_{Dataset}$ to find a dataset given noisy query measurements, and the number of iterations $k$.
        \State {\bfseries Output:} A dataset $g \in \npx$.
        \State Create a dataset $h_0 \in \npx$ with  $h_0(x)=1$ for all $x \in \XX$.
        \State Set $M \leftarrow \emptyset$ as a set of selected queries and their measurements.
        \For{$i=1, 2, \ldots, k$}
            \State Set $e_i \leftarrow \bp {\abs{q(h_{i-1})-q(f)}}_{q \in Q} $ as the error in queries.
            \State {\bf Select:} $l_i \leftarrow \AA_{Select}(e_i, 2/\e)$, an index of query.
            \State {\bf Measure:} $m_i \leftarrow q_{l_i}(f) + \Lap\bp{2\Delta_{q_{l_i}}/\e}$, value of query.
            \State Set $M \leftarrow M \cup \{ (q_{l_i}, m_i) \}$; add selected query and its value.
            \State \label{lst:meta_offline_algo_optimize} {\bf Optimize: } Dataset $h_i \leftarrow \AA_{Dataset}(M, h_{i-1})$.
        \EndFor
    \end{algorithmic}
    \caption{Meta algorithm: generating differentially private synthetic tabular dataset}
    \label{alg:meta_dp_synthetic_dataset}
\end{algorithm}

Algorithm~\ref{alg:meta_dp_synthetic_dataset} is a meta-algorithm describing the {\em select, measure, fit, and iterate} paradigm. This paradigm is used in several existing methods and has been shown to achieve good empirical accuracy \cite{Liu2021IterativeMF}. The algorithm has a fixed number of iterations $k$. It receives two subroutine algorithms $\AA_{Select}$ and $\AA_{Dataset}$ which can be treated as a black box for now. $\AA_{Select}$ is a differentially private algorithm used for selection, whereas $\AA_{Dataset}$ does not guarantee differential privacy and is used to create a dataset based on the values of the queries. Algorithm~\ref{alg:meta_dp_synthetic_dataset} iteratively produces a series of synthetic datasets $h_1, h_2, \ldots, h_k$ that are, hopefully, more and more closer to the true data $f$ as per the queries $Q$. In each iteration $i$, the following happens: (1) using the subroutine $\AA_{Select}$, while upholding differential privacy, we select a query $q_{l_i}$ that has the most error on the current synthetic dataset $h_{i-1}$; (2) an approximation of the value of this query is generated as $m_i$ using the Laplace Mechanism; and finally (3) the dataset is updated from $h_{i-1}$ to $h_i$ by using the sub-routine $\AA_{Dataset}$. 

\subsection{Dataset complying with queries}\label{s:model_distribution}

In this subsection, we discuss some algorithms that can be used for $\AA_{Dataset}$ in Algorithm~\ref{alg:meta_dp_synthetic_dataset}. 

\subsubsection{Multiplicative Weights (MW)}\label{s:offline_mwem}

\cite{hardt2012mwem} first introduced the idea of Algorithm~\ref{alg:meta_dp_synthetic_dataset} and used the {\em Multiplicative Weights} (MW) algorithm as $\AA_{Dataset}$ together with the Exponentia Mechanism for $\AA_{Select}$. With the MW algorithm, Step~\ref{lst:meta_offline_algo_optimize} of Algorithm~\ref{alg:meta_dp_synthetic_dataset} results in a dataset $h_{i}$ that is $|f|$ times the distribution that satisfies
$$
    h_i(x) \propto h_{i-1}(x) \cdot \exp{ \bp{ q_{l_i}(x) \cdot \frac{m_i - q_{l_i}(h_{i-1})}{2|f|} } }.
$$

MWEM solves a convex approximation of Equation~\eqref{eq:minimization_problem} over the probability simplex in $\XX$, we refer the reader to \cite{Liu2021IterativeMF} for more details. The algorithm comes with a theoretical guarantee and works quite well for low-dimensional datasets. However, since it requires maintaining a probability distribution over $\XX$, it becomes computationally intractable for many real-world datasets.

\subsubsection{Probabilistic Graphical Model (PGM)} \label{s:offline_pgm}

An alternative to the MW algorithm as $\AA_{Dataset}$ in Algorithm~\ref{alg:meta_dp_synthetic_dataset} is the Probabilistic Graphical Model (PGM) algorithm \cite{tabular_pgm_mckenna19a}. PGM further approximates the optimization problem by restricting the solution space from all possible distributions on $\XX$ to distributions that can be represented as a graphical model of the form 
$$
    P_{\theta}(x) = \frac{1}{Z} \exp{\bp{ \sum_{C\in\CC} \theta_C(x_C)}},
$$
for all $x\in\XX$. Here, $\CC\subseteq2^{[d]}$ is a collection of subsets of $[d]$, $\theta_C$ is a function for each $C\in\CC$, $x_C$ is the restriction of $x\in\XX$ on the column indices in $C$, and $Z$ is a normalization constant. Thus the model $P_\theta$ is defined by low-dimensional functions $\theta_C$, one for each $C\in\CC$. PGM uses a proximal algorithm to solve the resulting convex optimization problem. PGM has been shown to perform very well in practice and we will be using it in our experiments.

%% file: 4_baseline.tex
\section{Baseline: Streaming MWEM}
% ========================
% ========================
Let us get to our problem of producing synthetic stream $g$ for private stream $f$. In this section, we propose a baseline algorithm which can convert any offline algorithm $\AA$ to a streaming algorithm (say) $\AA^+$. The idea is very simple: given an input stream $f$, at any time $t$, $\AA^+$ runs an independent instance of algorithm $\AA$ on the differential dataset at time $t$, that is $\nabla f_t$, and produces the differential synthetic dataset $\nabla g_t$. It can be shown that if $\AA$ satisfies (offline) $\e$-differential privacy, then $\AA^+$ satisfies $\e$-differential privacy as per Definition~\ref{def:dp_stream}. $\AA$ can be any differentially private algorithm that generates a synthetic dataset. For our current discussion, we create our baseline streaming algorithm by fixing the offline mechanism $\AA$ as the MWEM algorithm \cite{hardt2012mwem}. Let us refer to the streaming version of MWEM as {\em StreamingMWEM} and we present the complete algorithm in  Algorithm~\ref{alg:streaming_mwem}. 

\begin{theorem}[Accuacy of StreamingMWEM]\label{thm:streaming_mwem_acc}
    At any time $t\in\N$, StreamingMWEM (Algorithm~\ref{alg:streaming_mwem}) is 
    $$
    \bp{ \bigo{ |f_t|^{2/3} \bp{ \frac{ \bp{t\log{t}} \ln{|\XX|}\ln{|Q|} }{\e\b} }^{1/3} }, \b}
    $$ accurate with respect to the set of marginal queries $Q$.
\end{theorem}

We provide the proof of Theorem~\ref{thm:streaming_mwem_acc} in Appendix~\ref{s:appendix_baseline_accuracy}.

\begin{algorithm}[ht]
    \begin{algorithmic}[1]
        \State {\bf Input:} An input data stream $f$, an ordered set of marginal queries $Q$, number of marginals to select at any time $k$, the privacy budget $\e$.
        \State {\bf Output:} A synthetic stream $g$.
        \State Initialize $g(0,x) \leftarrow 1$ for all $x \in \XX$.
        \For{$t=1, 2, \ldots$}
            \State Set $I_{t,0} \leftarrow \emptyset$.
            \For{$l=1,2,\dots,k$}
                \State Set $J_{t, l} \leftarrow [|Q|] \setminus I_{t, l-1}$; as query indices not selected.
                \State Set $e_{t, l} \leftarrow \bp{ \vert q_i(\nabla f_t) - q_i(h_{t, l-1})\vert }_{i \in J_{t, l}} $.
                \Statex
                \State // Exponential Mechanism \label{lst:alg_smwem_selection}
                \State Sample a query index $\eta_{t,l}$ such that for any $i \in J_{t,l}$,
                $$
                    \prob{\eta_{t,l} = i} \propto \exp{ \bp{\frac{\e}{2k} (e_{t,l})_i } }.
                $$
                % \Statex
                % \State // Exponential mechanism with Gumbel noise \label{lst:alg_smwem_selection}
                % \State Initialize $\g_{t,l} \in \R^{\abs{J_{t,l}}}$ with i.i.d. $(\g_{t,l})_i \sim \Gumbel\bp{\frac{4k}{\e}}$.
                % \State Set, $\eta_{t,l} \leftarrow \argmax_{i \in J_{t,l}} \bp{ e_{t,l} + \g_{t,l}}_i $.
                \State Using $j$ as a shorthand for $\eta_{t,l}$.
                \State Set $I_{t, l} \leftarrow I_{t, l-1} \cup \{ j \}$.
                \Statex
                \State // Laplace Mechanism \label{lst:alg_smwem_measure}
                \State Sample $\d_{t,l} \sim \Lap\bp{\frac{2k}{\e}}$.
                \State Set $m(t,j) \leftarrow q_j(\nabla f_t) + \d_{t,l}$.
                \Statex
                \State // Multiplicative weights \label{lst:alg_smwem_mw}
                \State Set $h_{t, l}$ as $|\nabla f_t|$ times the distribution that satisfies
                $$
                    h_{t,l}(x) \propto h_{t, l-1}(x) \cdot \exp{ \bp{ q_j(x) \cdot \frac{m_j - q_j(h_{t,l-1})}{2|\nabla f_t|} } }
                $$
            \EndFor
            \State Set $g_t \leftarrow g_{t-1} + \avg_{l \in [k]} h_{t,l}$.
        \EndFor
    \end{algorithmic}
    \caption{Baseline algorithm: StreamingMWEM}
    \label{alg:streaming_mwem}
\end{algorithm}

%% file: 4_method.tex
\section{Proposed algorithm}

In this section, we present out proposed algorithm.

\subsection{Outline}
In a nutshell, our algorithm also follows the {\em select, measure, fit, and iterate} paradigm described in Algorithm~\ref{alg:meta_dp_synthetic_dataset}. At any time $t\in\N$, the goal is to ensure that $f_t$ and $g_t$ are close to each other as evaluated using the queries in the enumerated set $Q$. We start with a dataset $h_{t,0}=g_{t-1}$ and update it over $k$ iterations from $h_{t,0}, h_{t,1}, \dots$, to $h_{t,k}$. At any iteration $l \in [k]$, we select a query index $\eta_{t,l} \in [|Q|]$ for which our dataset $h_{t,l-1}$ has approximately the highest error when compared to $f_t$. We will discuss how exactly this selection is done soon, but for now, let us accept it as a black-box. At the end of the $k$ iterations, $g_t$ is set to some aggregate of the datasets $h_{t,1}, h_{t,2}, \ldots$, and $h_{t,k}$. We present our proposed method as a meta-algorithm in Algorithm~\ref{alg:main}.

\begin{algorithm}[ht]
    \begin{algorithmic}[1]
        \State {\bf Input:} An input data stream $f$, an ordered set of marginal queries $Q$, number of marginals to select at any time $k$, the privacy budget $\e$, a counter algorithm $\AA_{Counter}$, and a subroutine $\AA_{Dataset}$ to find a dataset given noisy query measurements.
        \State {\bf Output:} A synthetic stream $g$.
        \State Initialize $C_1, C_2, \ldots, C_{|Q|}$ as independent instances of the counter algorithm $\AA_{Counter}$, one for each query in the set $Q$, with privacy budget $\e/2k$.
        \State Initialize $g(0,x) \leftarrow 1$ for all $x \in \XX$.
        \State Initialize $m(0, i) \leftarrow 0$ for all $i \in [|Q|]$; query measurements of selected queries
        \State Initialize $r(0, i) \leftarrow 0$ for all $i \in [|Q|]$; remainder of query value for times when the query is not selected.
        \For{$t=1, 2, \ldots$}
            \State Set $I_{t,0} \leftarrow \emptyset$.
            \For{$l=1,2,\dots,k$}
                \State Set $J_{t, l} \leftarrow [|Q|] \setminus I_{t, l-1}$; as query indices not selected.
                \State Set $e_{t, l} \leftarrow \bp{ \vert q_i(\nabla f_t + g_{t-1}) - q_i(h_{t, l-1})\vert }_{i \in J_{t, l}} $.
                \State $\eta_{t,l} \leftarrow ExponentialMechanism \bp{ e_{t,l}, \e/2k }$.
                \State Using shorthand $j$ for $\eta_{t,l}$.
                \State Set $I_{t, l} \leftarrow I_{t, l-1} \cup \{ j \}$.
                \State Invoke counter subroutine $C_j$ with input $\nabla f_t$.
                \State Set $r(t,j) \leftarrow r(t-1, j)$.
                \State Set $m(t,j) \leftarrow C_j + r(t, j)$.
                \State Set $h_{t, l} \leftarrow \AA_{Dataset} \bp{ \set{ \bp{ q_i, m(t, i) } }_{i \in I_{t,l}}, h_{t, l-1} }$.
            \EndFor
            \State Set $g_t \leftarrow \avg_{l \in [k]} h_{t,l}$.
            \State Set $C_i(t) \leftarrow  C_i(t-1)$ for all $i \in [|Q|]\setminus I_{t,k}$.
            \State Set $r(t, i) \leftarrow  q_i(g_t) - C_i(t)$ for all $i \in [|Q|]\setminus I_{t,k}$.
        \EndFor
    \end{algorithmic}
    \caption{Main algorithm: streaming differentially private synthetic tabular stream}
    \label{alg:main}
\end{algorithm}

\subsection{Measure}
Let $m:\N\times[|Q|] \to \R$ be a map such that $m(t,i)$ denotes our differentially private approximation of $q_i(f_t)$, that is the value of query $q_i \in Q$ at time $t$. Since a single query may be selected at multiple time instances, we use a counter algorithm to measure the value of the query efficiently over time. We associate each query in $Q$ with an instance of some counter Algorithm, say~$\AA_{Counter}$. Consider a query $q_i\in Q$ and let $C_i$ be its corresponding counter. We use the notation $C_i(t)$ to conveniently refer to the value of the counter $C_i$ at time $t$. Let $N_i(t) \subseteq [t]$ be the time instances until time $t$ when the query $q_i$ was selected to be measured using the true data. Also, let $\bar N_i(t) \coloneqq [t]\setminus N_i(t)$ be the time instances until time $t$ at which query $q_i$ was not selected. Then the output $C_i(t)$ of the counter algorithm is based solely on the stream $\nabla f_{N_i(t)}$.

However, to generate the dataset $g_t$ we need an approximate measurement of the value $q_i(f_t)$. In other words, we are missing the measurement of the query on times $\bar N_i(t)$ when the index $i$ was not selected. At any such time $\tau \in \bar N_i(t)$, since $q_i$ was not selected, we assume that the query value $q_i(g_{\tau})$ on the synthetic dataset $g_\tau$ is close to the true value $q_i(f_{\tau})$. We create a map $r:\N\times[|Q|] \to \R$ such that $r(t,i)$ denotes our differentially private approximation of the value of query $q_i$ over times in $\bar N_i(t)$. Assuming $r(0,i)=0$, we define $r(t,i)$ for any $t \in \N$ as,
$$
    r(t,i) = \begin{cases}
        q_i(g_t) - C_i(t), &t \in N_i(t),\\
        r(t-1, i), &otherwise.
    \end{cases}
$$
Finally, our differentially private approximation $m(t,i)$ of the query $q_i$ at time $t$ becomes
$$
    m(t,i) = C_i(t) + r(t,i).
$$

\subsection{Fit}
At any time $t$ and iteration $l$, Algorithm~\ref{alg:main} uses the Algorithm~$\AA_{Dataset}$ as a subroutine to generate the synthetic dataset $h_{t,l}$ using the query indices selected so far at time $t$, that is $\{\eta_{t,1}, \ldots, \eta_{t,l}\}$, and their corresponding differentially private values $\{m(t, \eta_{t,1}),\ldots, m(t, \eta_{t,l})\}$. $\AA_{Dataset}$ can be any algorithm and is not required to satisfy differential privacy.

\subsection{Select}
We are finally ready to talk about query selection. During iteration $l$ of time $t$, we want to select the query with maximum error over the synthetic dataset $h_{t,l-1}$ as compared to the true dataset $f_t$. However, accessing $q(f_t)$ results in high sensitivity. Indeed a simple change at some time $\tau \in \N$ can affect the selection at all times $t>\tau$.

To control the sensitivity, we follow a trick inspired by \cite{kumar_algorithm_2024} and approximate $f_t$ as $g_{t-1} + \nabla f_t$ for selection. For any query $q_i \in Q$, $l \in [k]$, and $t\in\N$, we define
$$
    e_{t, l} \coloneqq \bp{ \vert q_i(\nabla f_t + g_{t-1}) - q_i(h_{t, l-1})\vert }_{i \in [|Q|]}.
$$
Finally, we use the Exponential Mecnahism as defined in Definition~\ref{def:exp_mech} for selecting a query index $\eta_{t,l}$ given the vector of query utilities $e_{t,l}$. Note that Algorithm~\ref{alg:main} does not find the error for all queries but instead only for queries that have not been chosen so far at iteration $l$ of time $t$ (whose indices are in the set $J_{t,l}$).

\begin{theorem}[Privacy of Algorithm~\ref{alg:main}]
    If the Algorithms $\AA_{Counter}$ and $\AA_{Dataset}$ satisfy differential privacy, then Algorithm~\ref{alg:main} is $\e$-differentially private.
\end{theorem}
\begin{proof}
    Note that Algorithm~\ref{alg:main} is an instance of the selective counting algorithm from \cite{kumar_algorithm_2024}. Moreover, since we split the budget as $\e/2$ for the selection (Exponential Mechanism) subroutine and the remaining $\e/2$ for the counters at any time $t$, Algorithm~\ref{alg:main} satisfies $\e$-differential privacy.
\end{proof}

%% file: 5_experiments.tex
\section{Experiments and results}\label{s:results}

 \begin{figure*}[ht]
     \centering
     \begin{subfigure}[t]{0.23\linewidth}
         \centering
         \includegraphics[height=0.15\textheight]{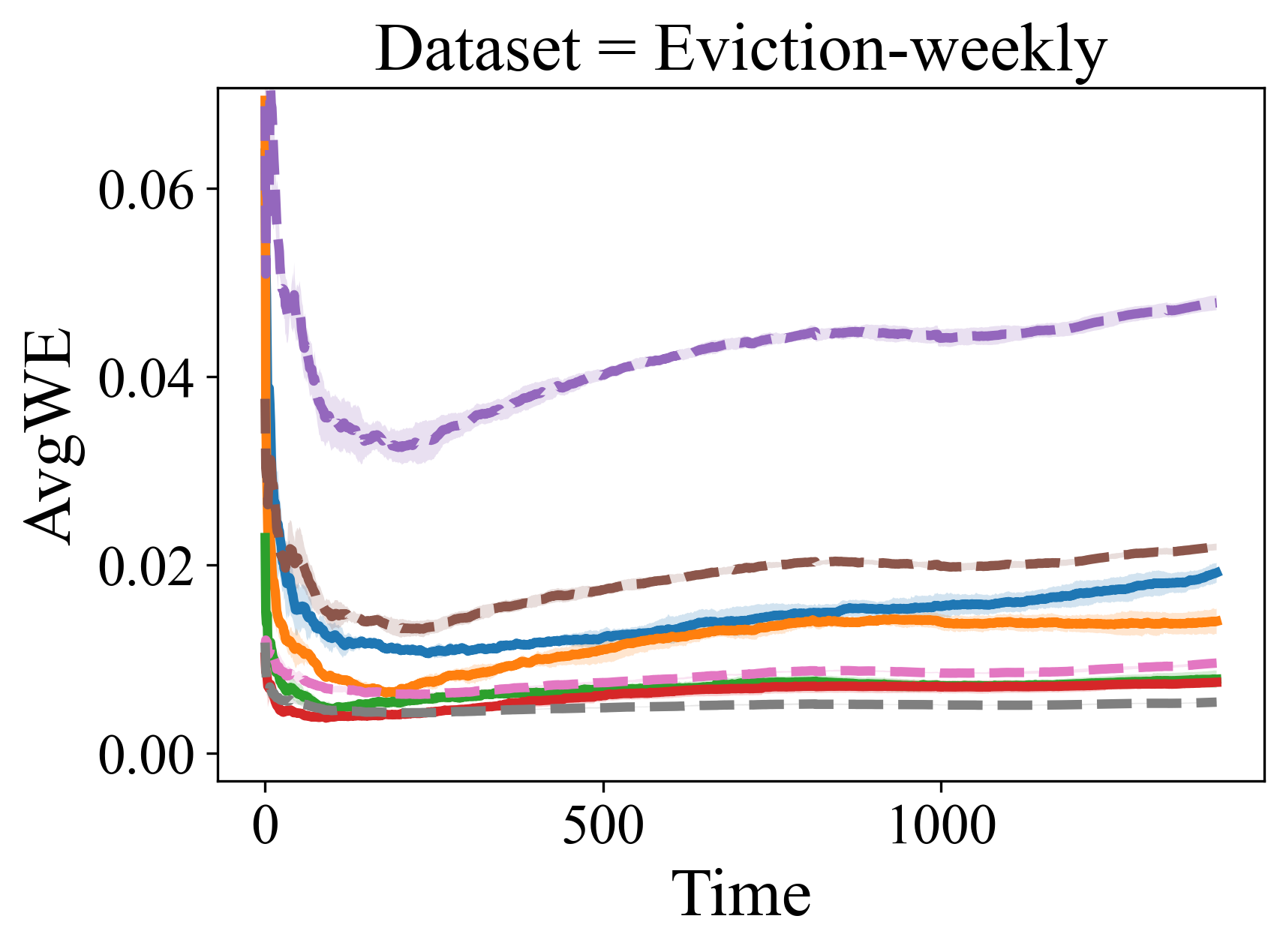}
        \caption{Average of workload errors}
     \end{subfigure}
     \hspace*{\fill}
     \begin{subfigure}[t]{0.23\linewidth}
         \centering
         \includegraphics[height=0.15\textheight]{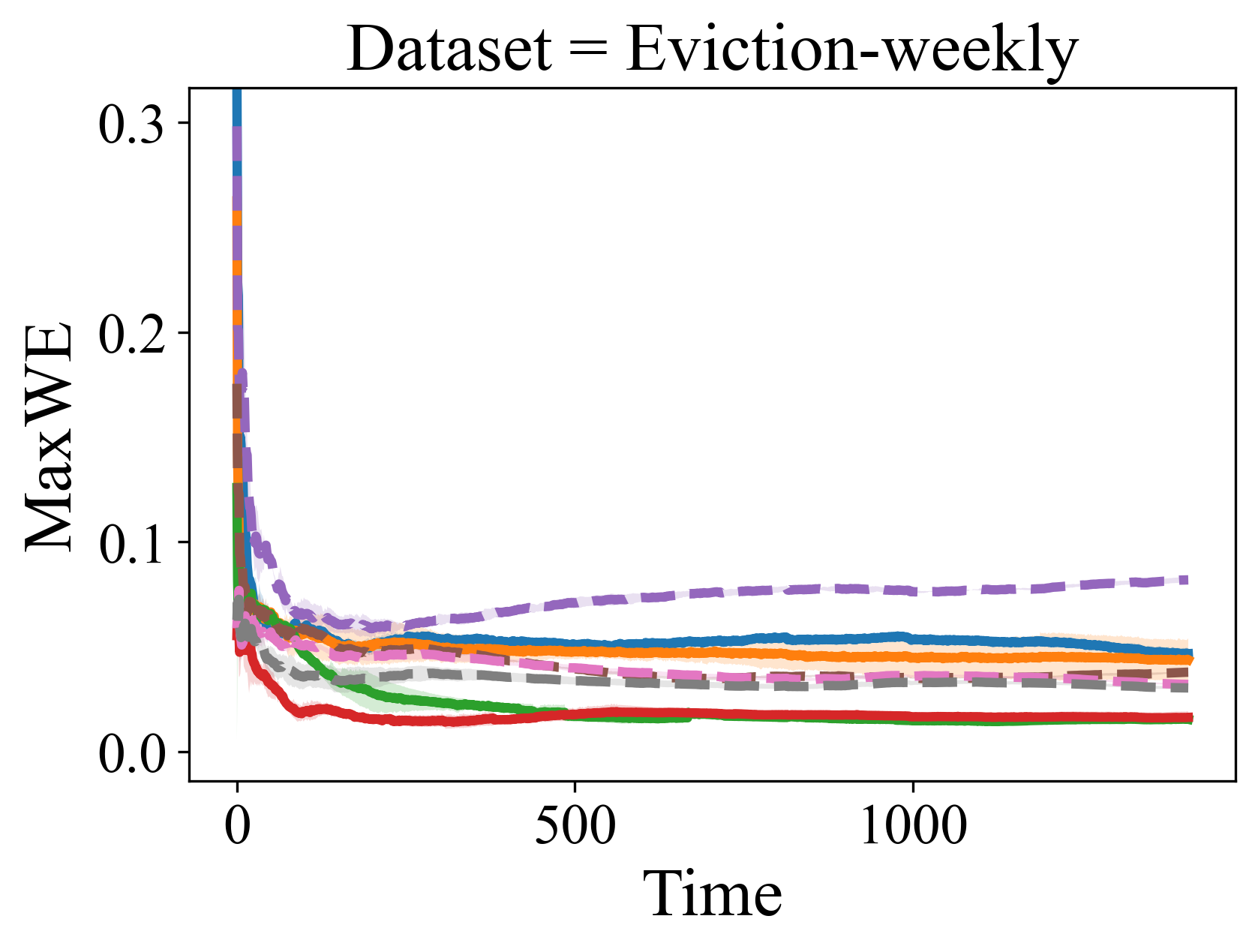}
        \caption{Maximum of workload errors}
     \end{subfigure}
     \hspace*{\fill}
     \begin{subfigure}[t]{0.23\linewidth}
         \centering
         \includegraphics[height=0.15\textheight]{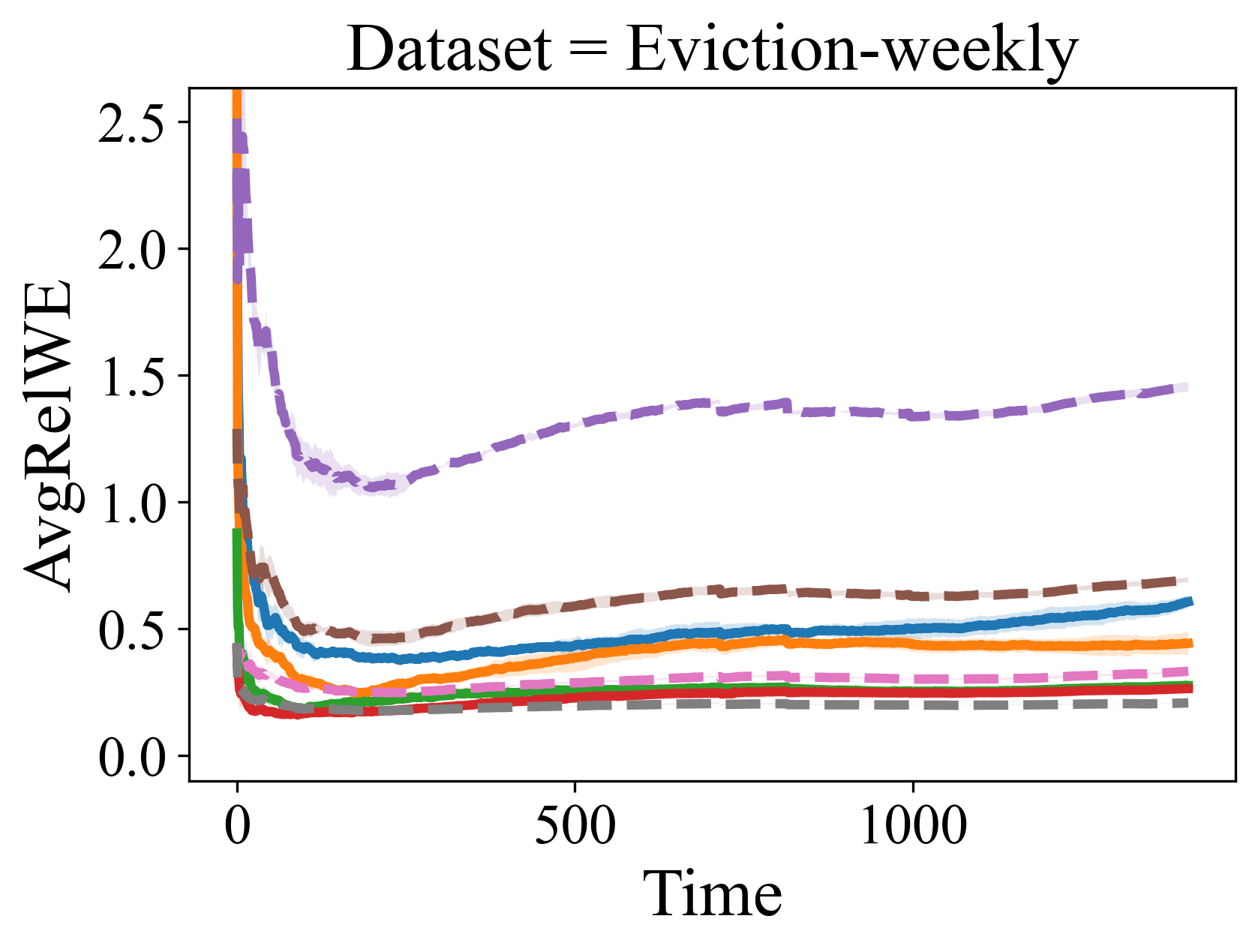}
        \caption{Average of relative workload errors}
     \end{subfigure}
     \hspace*{\fill}
     \begin{subfigure}[t]{0.23\linewidth}
         \centering
         \includegraphics[height=0.15\textheight]{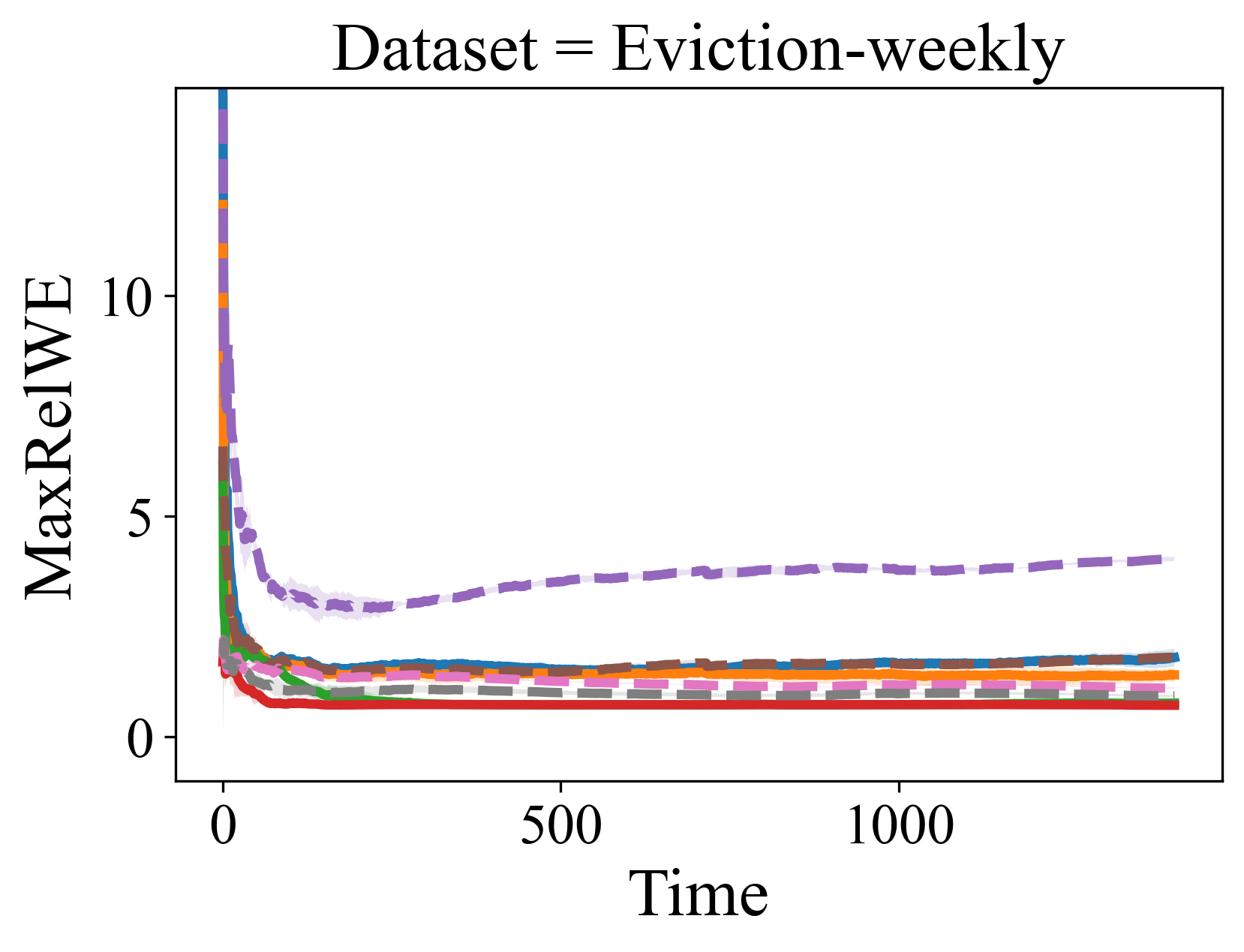}
        \caption{Maximum of relative workload errors}
     \end{subfigure}
    \hfill
     \begin{subfigure}[t]{\linewidth}
         \centering
         \includegraphics[height=0.03\textheight]{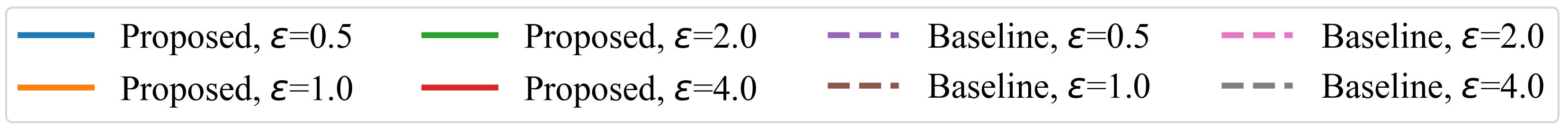}
     \end{subfigure}
    \caption{Metrics over time to compare the baseline and proposed method for the Eviction-weekly dataset.}
    \label{fig:Eviction-weekly}
\end{figure*}

 \begin{figure*}[ht]
     \centering
     \begin{subfigure}[t]{0.23\linewidth}
         \centering
         \includegraphics[height=0.15\textheight]{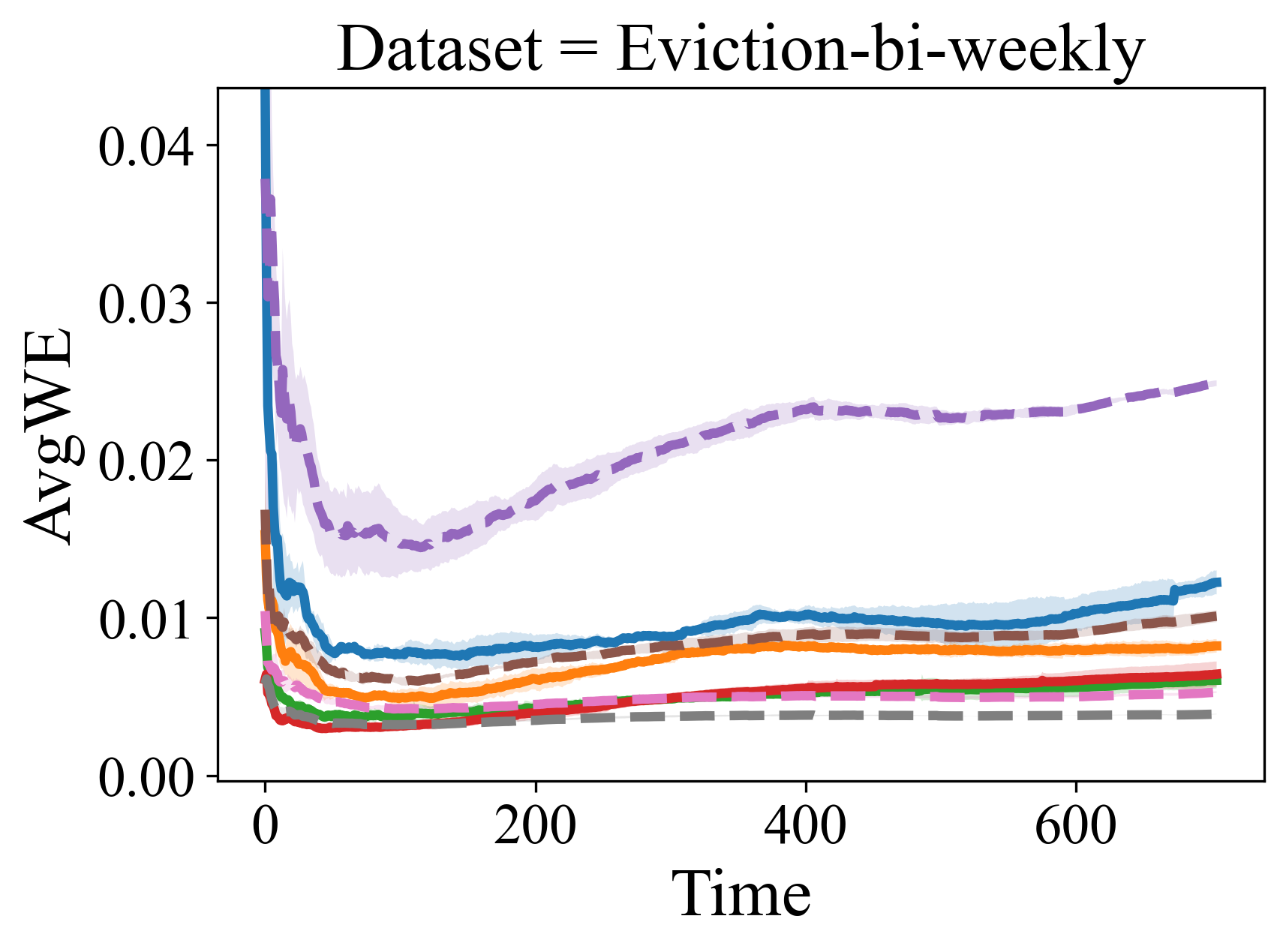}
        \caption{Average of workload errors}
     \end{subfigure}\hspace*{\fill}
     \begin{subfigure}[t]{0.23\linewidth}
         \centering
         \includegraphics[height=0.15\textheight]{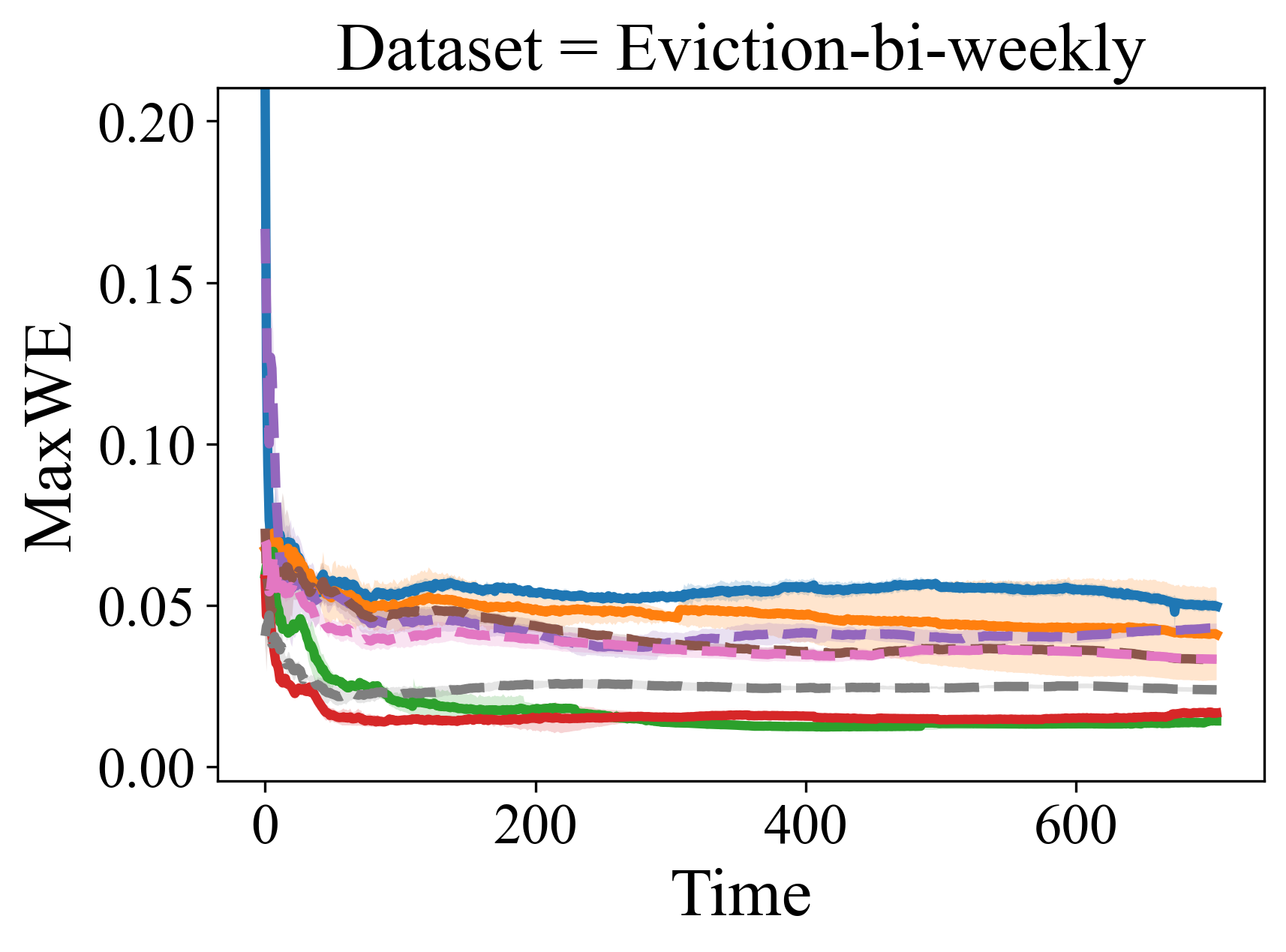}
        \caption{Maximum of workload errors}
     \end{subfigure}\hspace*{\fill}
     \begin{subfigure}[t]{0.23\linewidth}
         \centering
         \includegraphics[height=0.15\textheight]{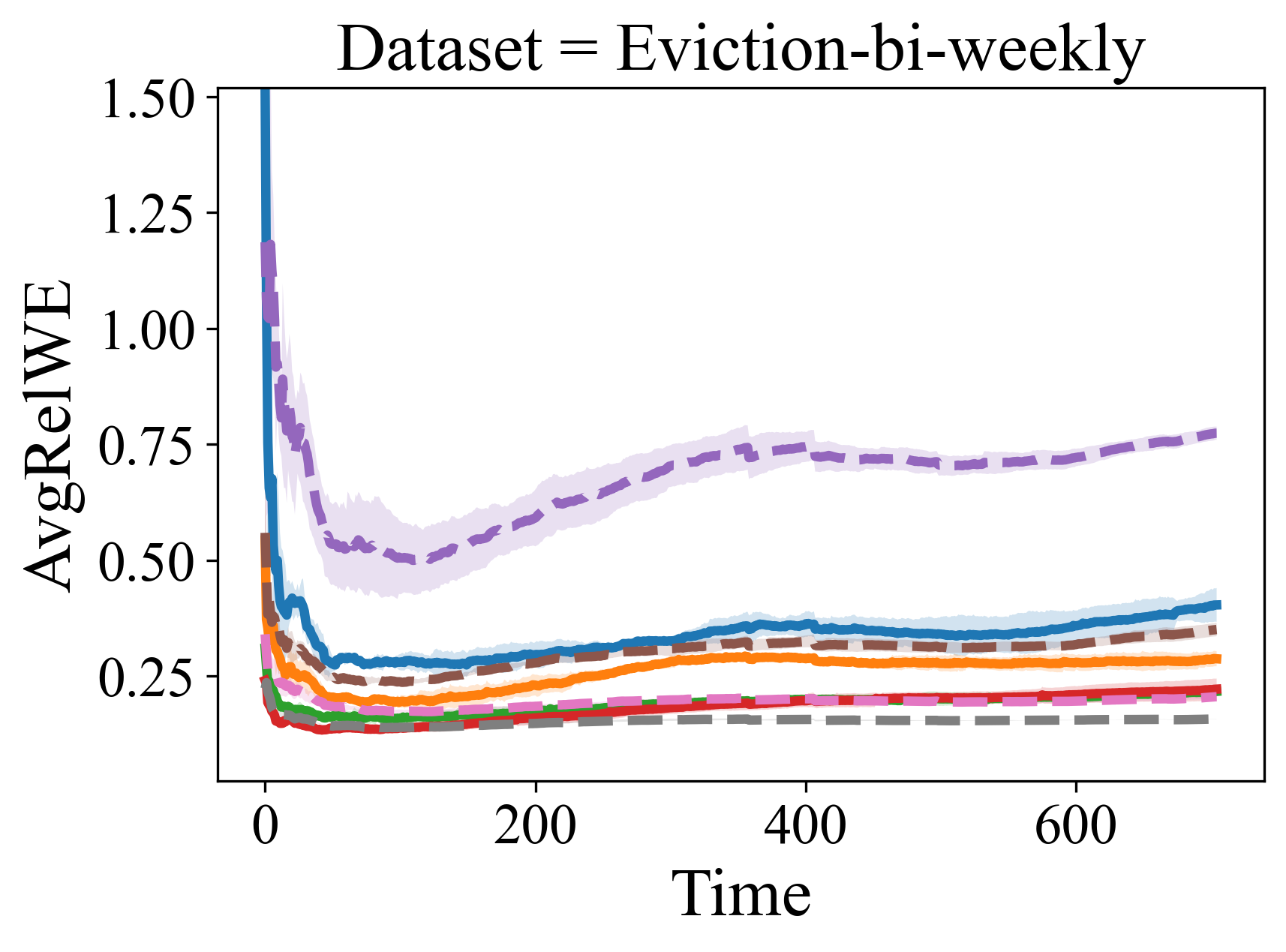}
        \caption{Average of relative workload errors}
     \end{subfigure}\hspace*{\fill}
     \begin{subfigure}[t]{0.23\linewidth}
         \centering
         \includegraphics[height=0.15\textheight]{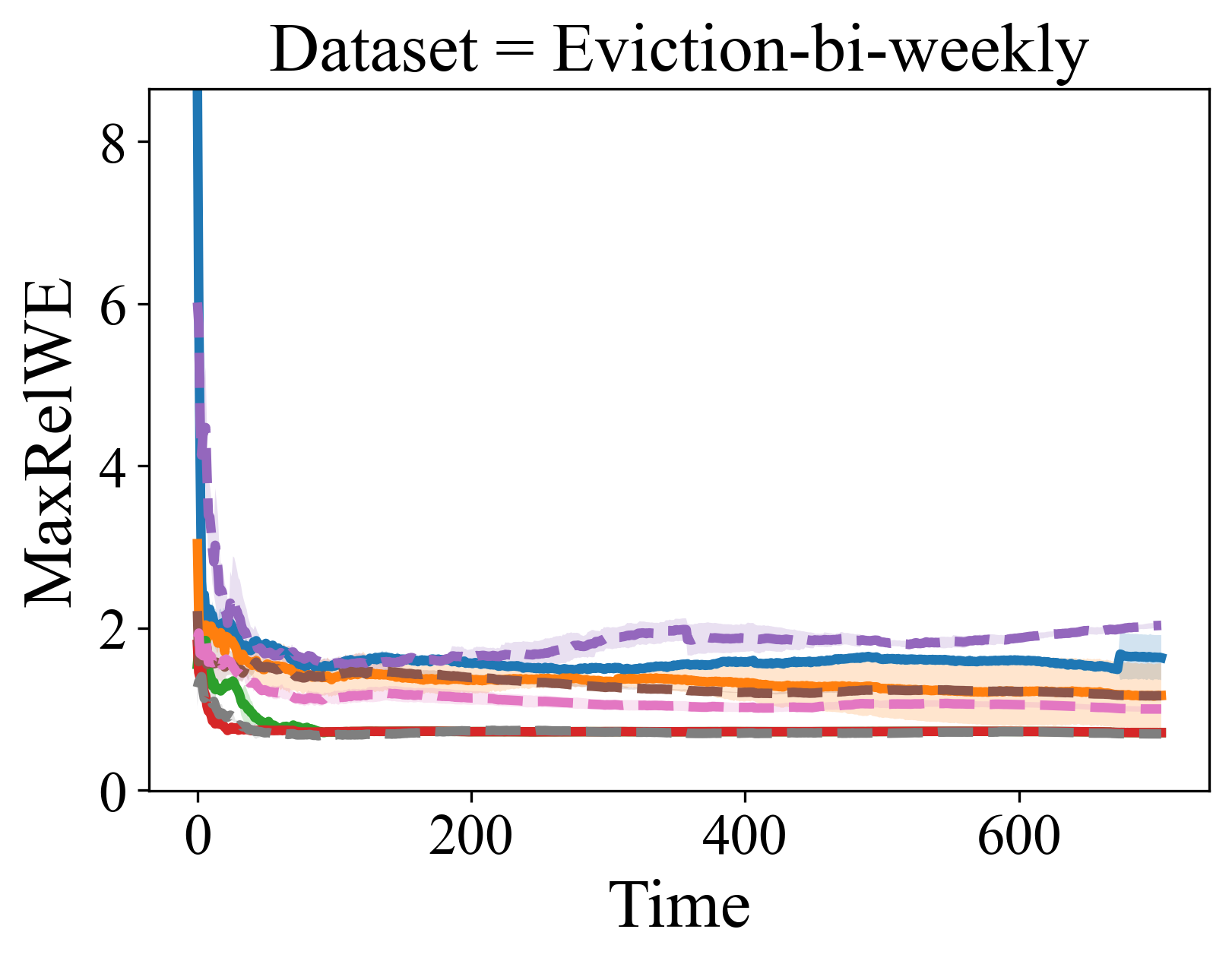}
        \caption{Maximum of relative workload errors}
     \end{subfigure}
    \hfill
     \begin{subfigure}[t]{\linewidth}
         \centering
         \includegraphics[height=0.03\textheight]{figures/legend.png}
     \end{subfigure}
    \caption{Metrics over time to compare the baseline and proposed method for the Eviction-bi-weekly dataset.}
    \label{fig:Eviction-bi-weekly}
\end{figure*}

 \begin{figure*}[ht]
     \centering
     \begin{subfigure}[t]{0.23\linewidth}
         \centering
         \includegraphics[height=0.15\textheight]{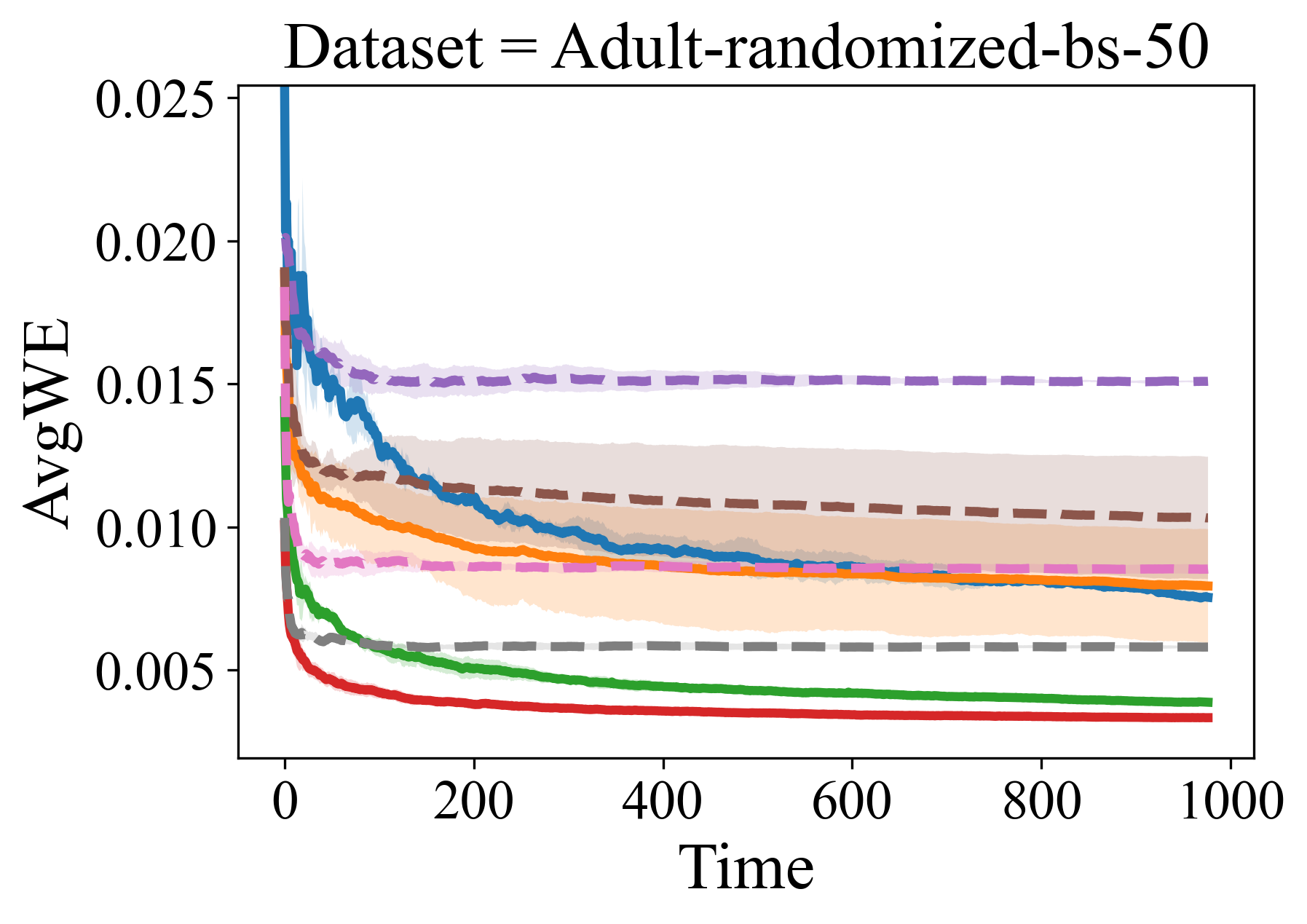}
        \caption{Average of workload errors}
     \end{subfigure}\hspace*{\fill}
     \begin{subfigure}[t]{0.23\linewidth}
         \centering
         \includegraphics[height=0.15\textheight]{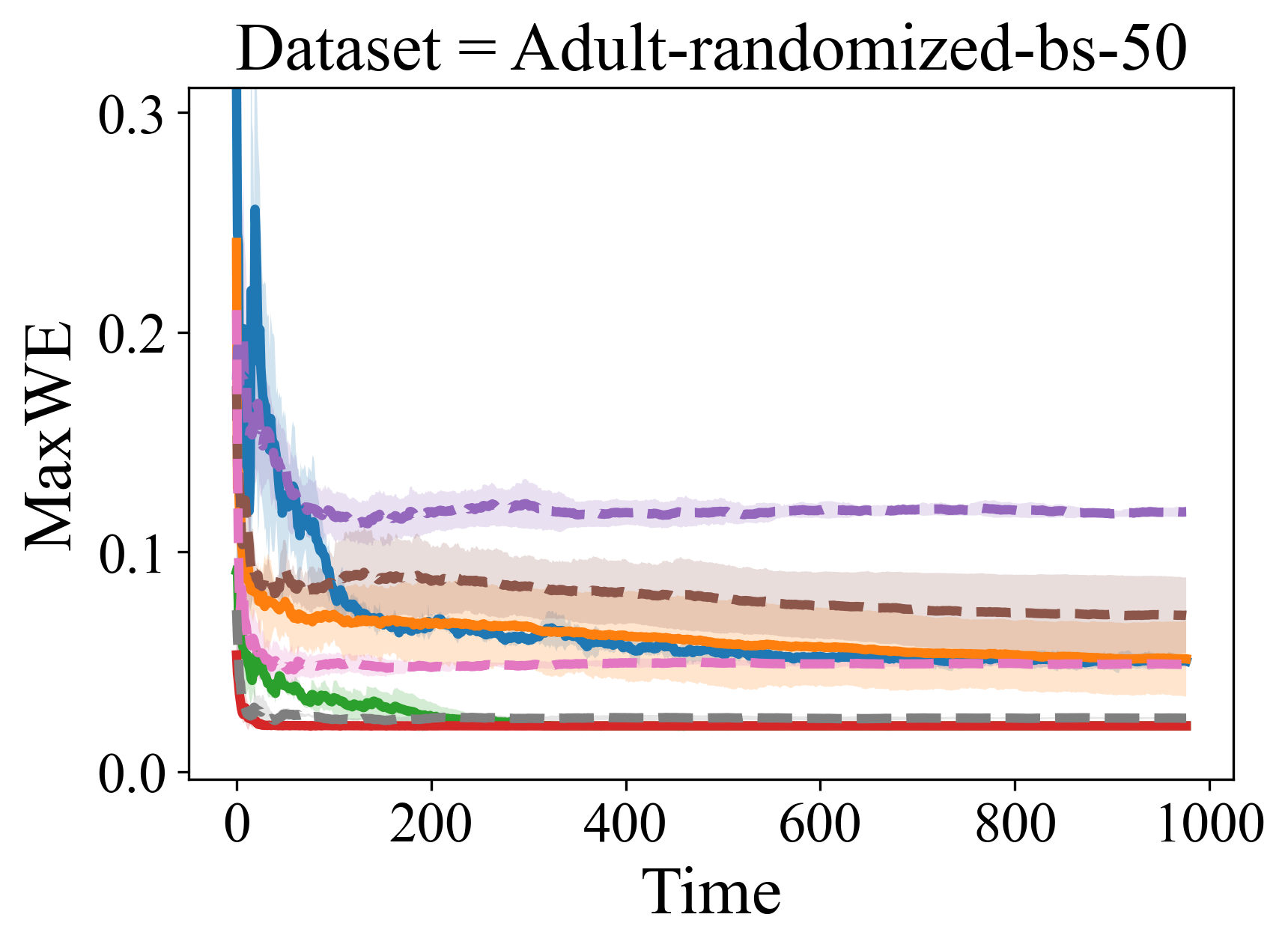}
        \caption{Maximum of workload errors}
     \end{subfigure}\hspace*{\fill}
     \begin{subfigure}[t]{0.23\linewidth}
         \centering
         \includegraphics[height=0.15\textheight]{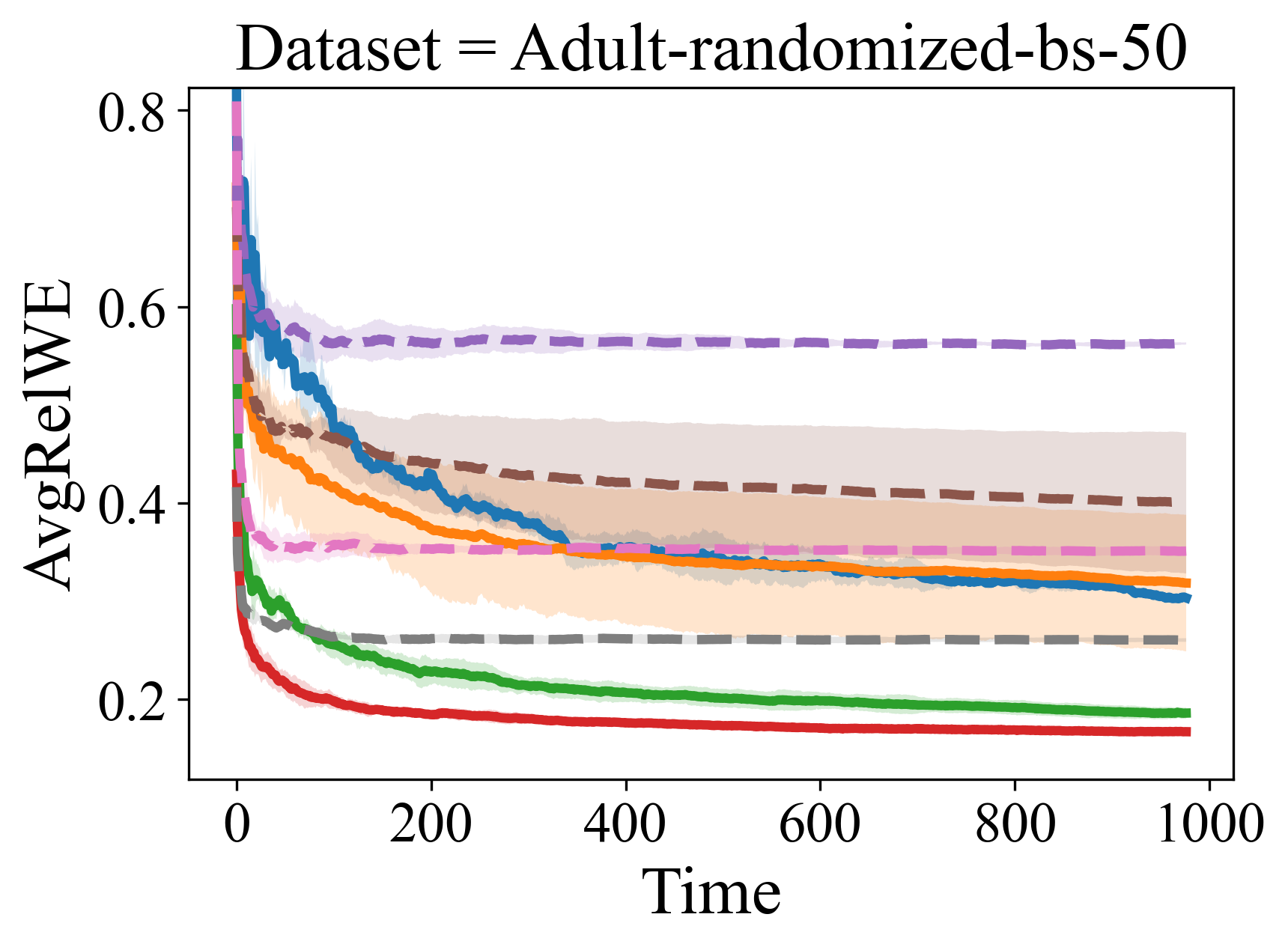}
        \caption{Average of relative workload errors}
     \end{subfigure}\hspace*{\fill}
     \begin{subfigure}[t]{0.23\linewidth}
         \centering
         \includegraphics[height=0.15\textheight]{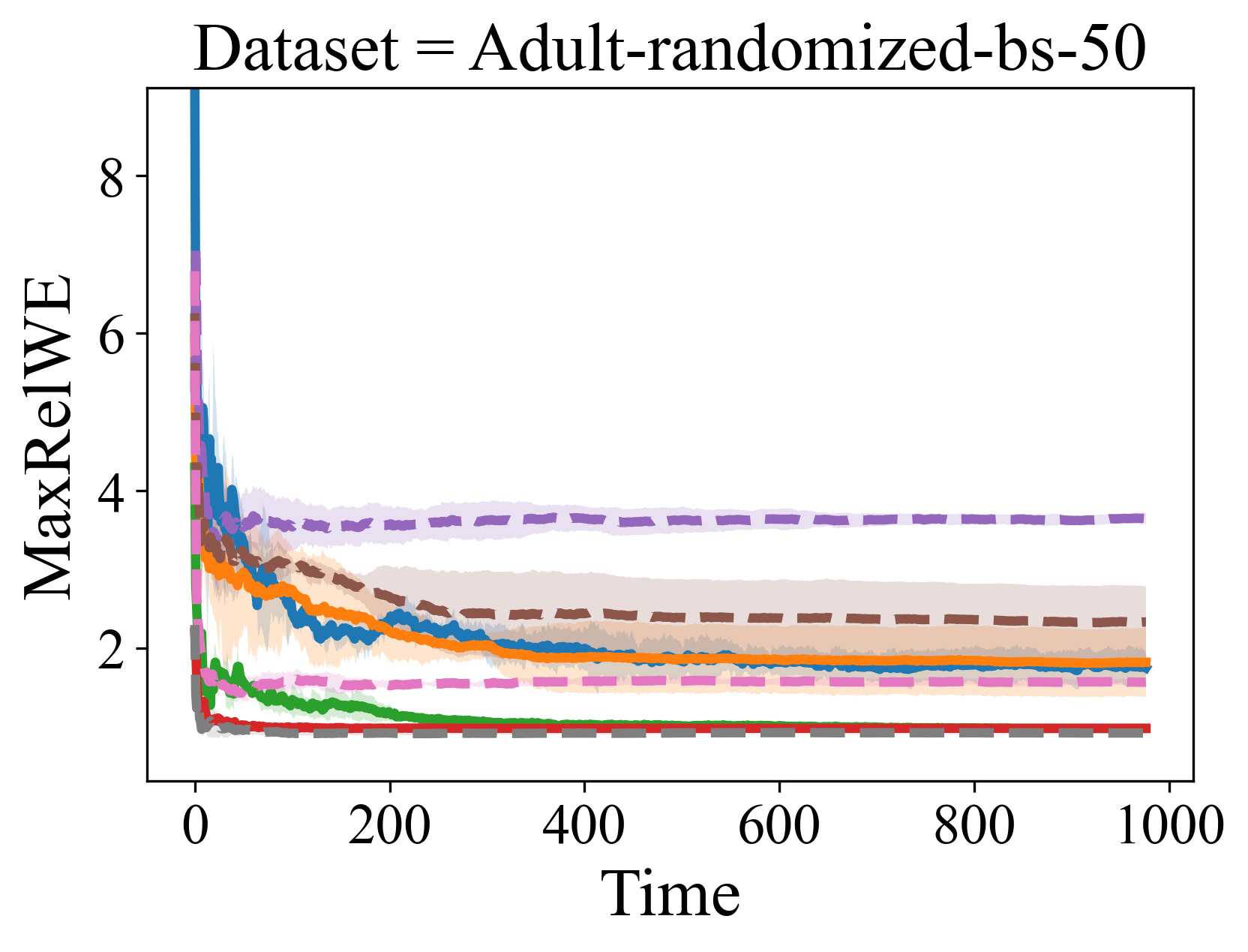}
        \caption{Maximum of relative workload errors}
     \end{subfigure}
    \hfill
     \begin{subfigure}[t]{\linewidth}
         \centering
         \includegraphics[height=0.03\textheight]{figures/legend.png}
     \end{subfigure}
    \caption{Metrics over time to compare the baseline and proposed method for the Adult-randomized-bs-50 dataset.}
    \label{fig:Adult-randomized-bs-50}
\end{figure*}

 \begin{figure*}[ht]
     \centering
     \begin{subfigure}[t]{0.23\linewidth}
         \centering
         \includegraphics[height=0.15\textheight]{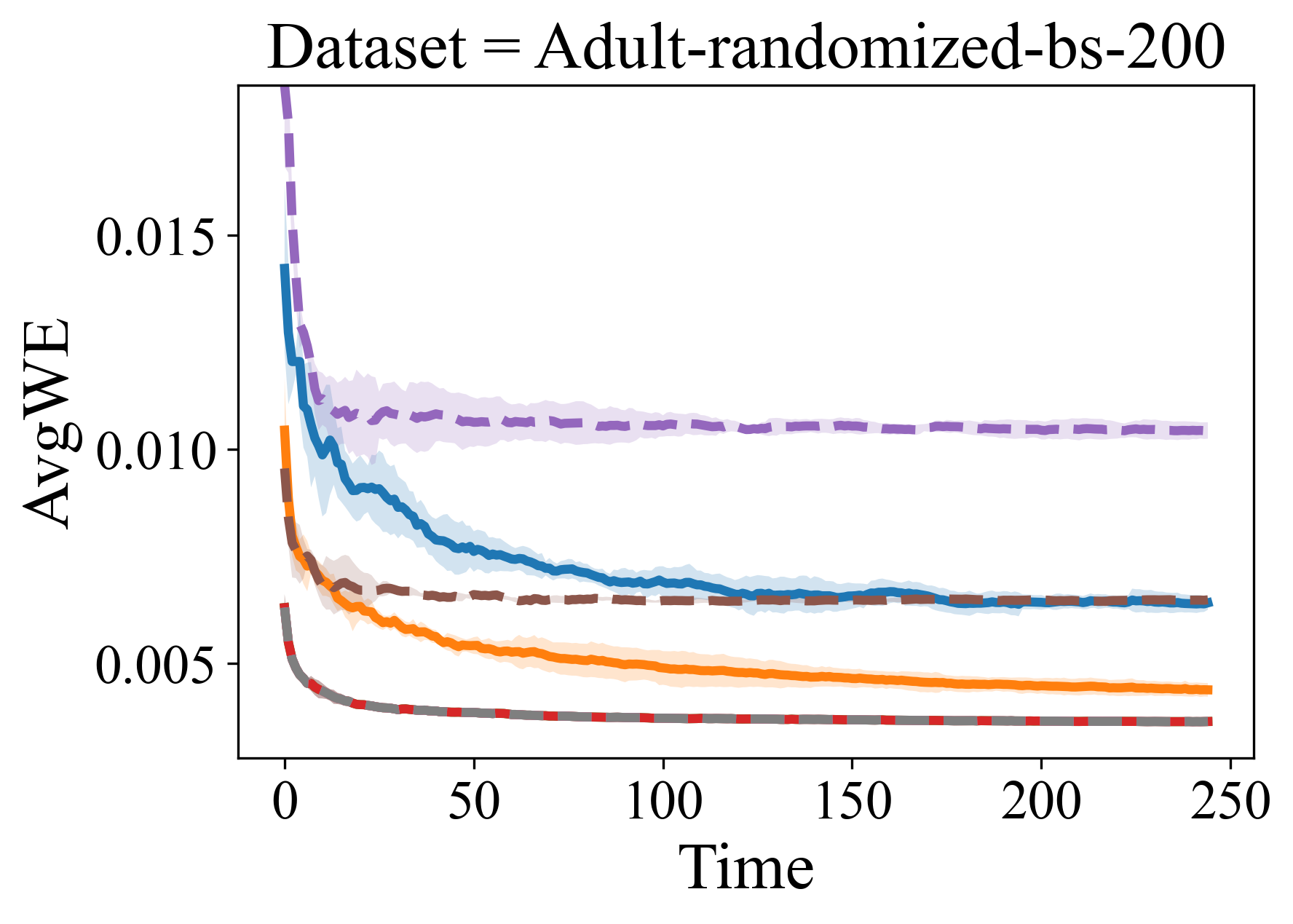}
        \caption{Average of workload errors}
     \end{subfigure}\hspace*{\fill}
     \begin{subfigure}[t]{0.23\linewidth}
         \centering
         \includegraphics[height=0.15\textheight]{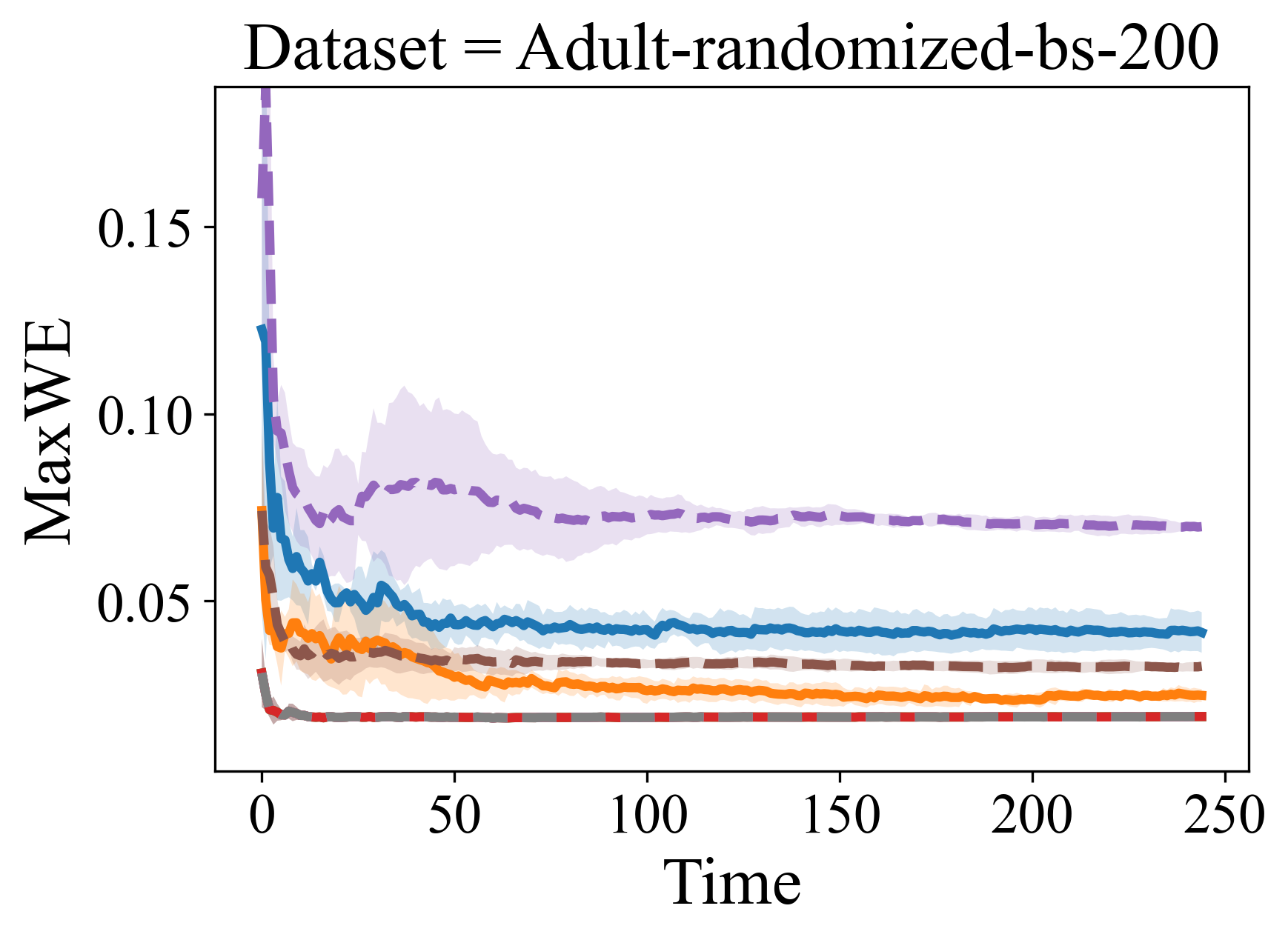}
        \caption{Maximum of workload errors}
     \end{subfigure}\hspace*{\fill}
     \begin{subfigure}[t]{0.23\linewidth}
         \centering
         \includegraphics[height=0.15\textheight]{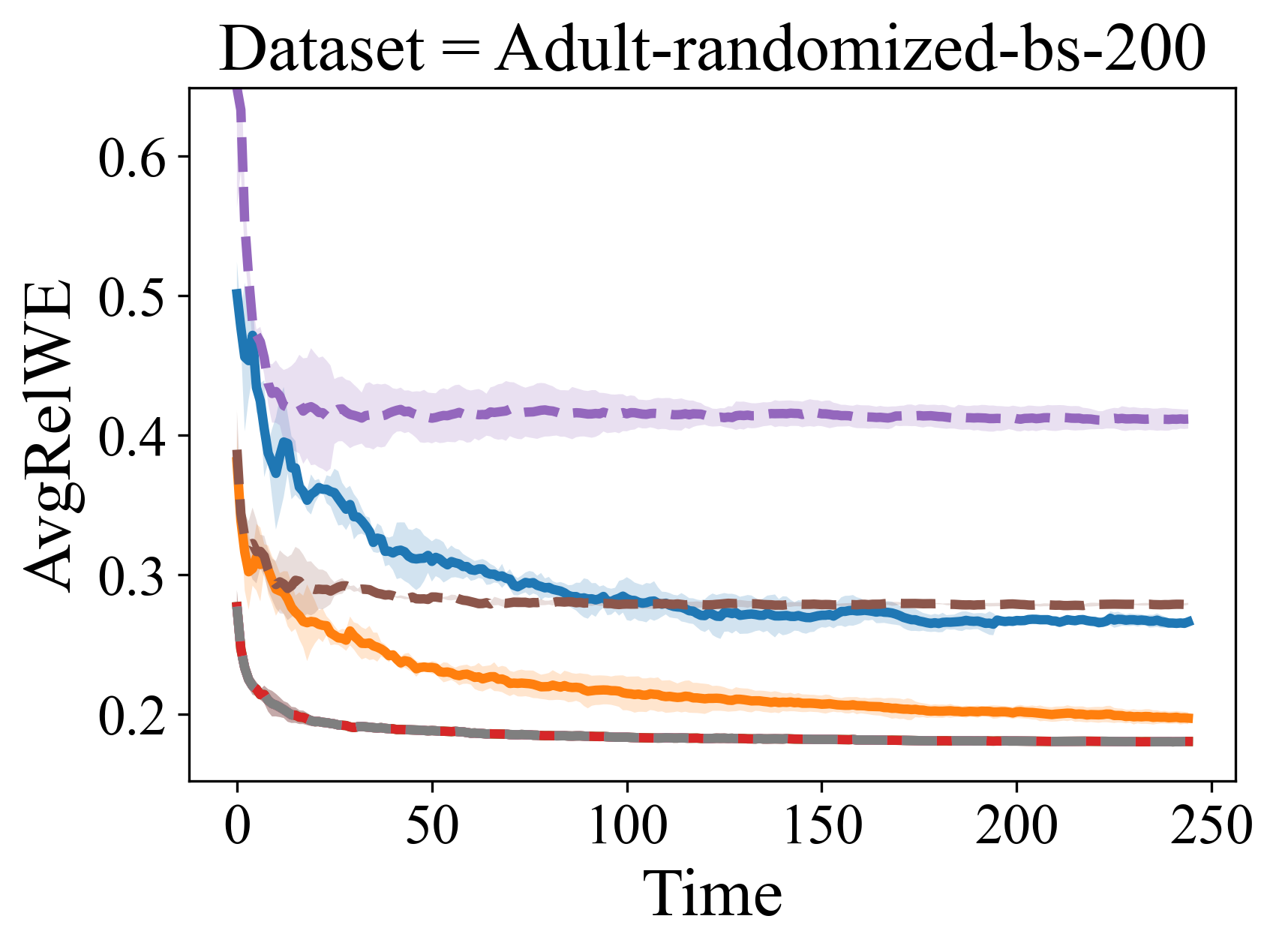}
        \caption{Average of relative workload errors}
     \end{subfigure}\hspace*{\fill}
     \begin{subfigure}[t]{0.23\linewidth}
         \centering
         \includegraphics[height=0.15\textheight]{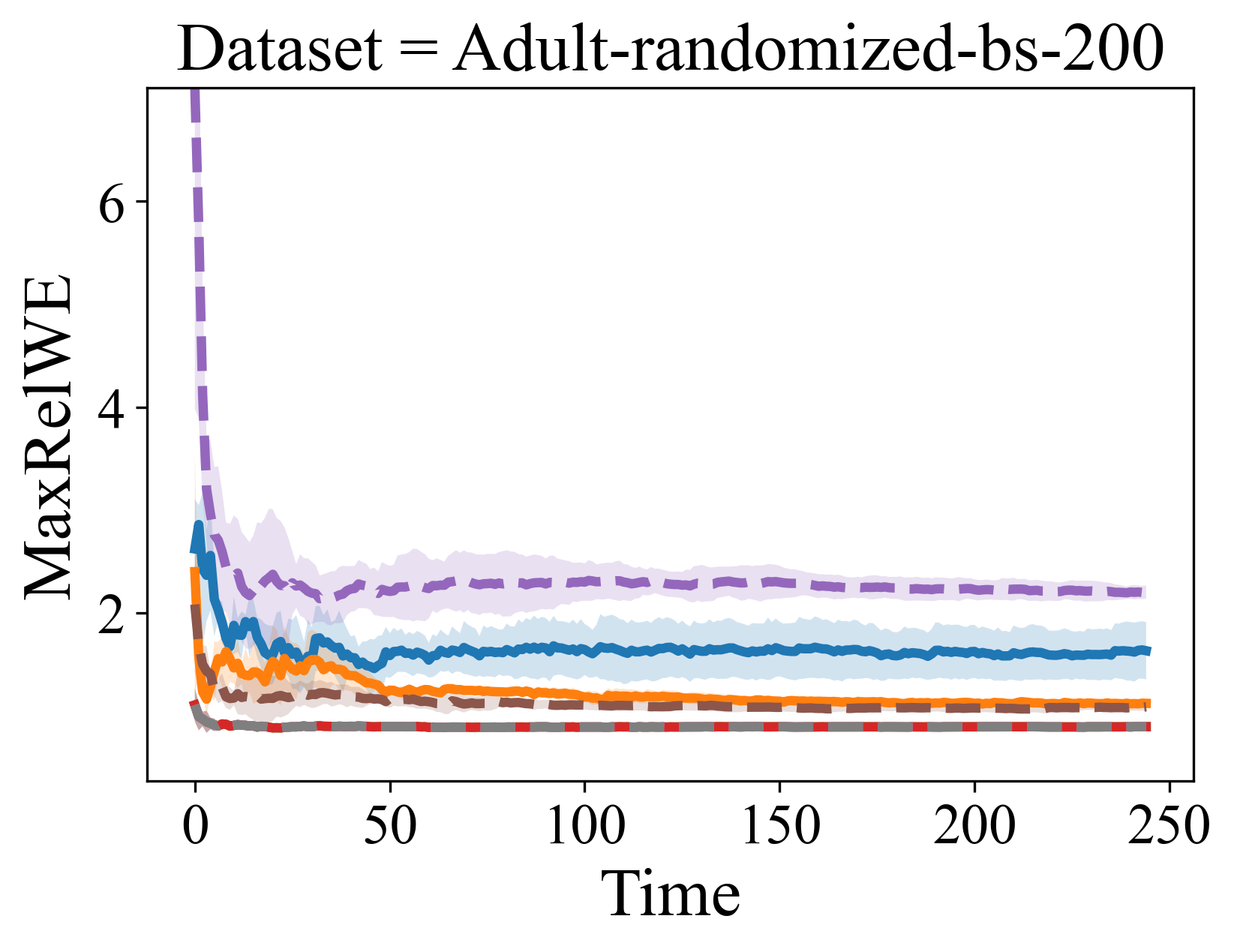}
        \caption{Maximum of relative workload errors}
     \end{subfigure}
    \hfill
     \begin{subfigure}[t]{\linewidth}
         \centering
         \includegraphics[height=0.03\textheight]{figures/legend.png}
     \end{subfigure}
    \caption{Metrics over time to compare the baseline and proposed method for the Adult-randomized-bs-200 dataset.}
    \label{fig:Adult-randomized-bs-200}
\end{figure*}

 \begin{figure*}[ht]
     \centering
     \begin{subfigure}[t]{0.23\linewidth}
         \centering
         \includegraphics[height=0.15\textheight]{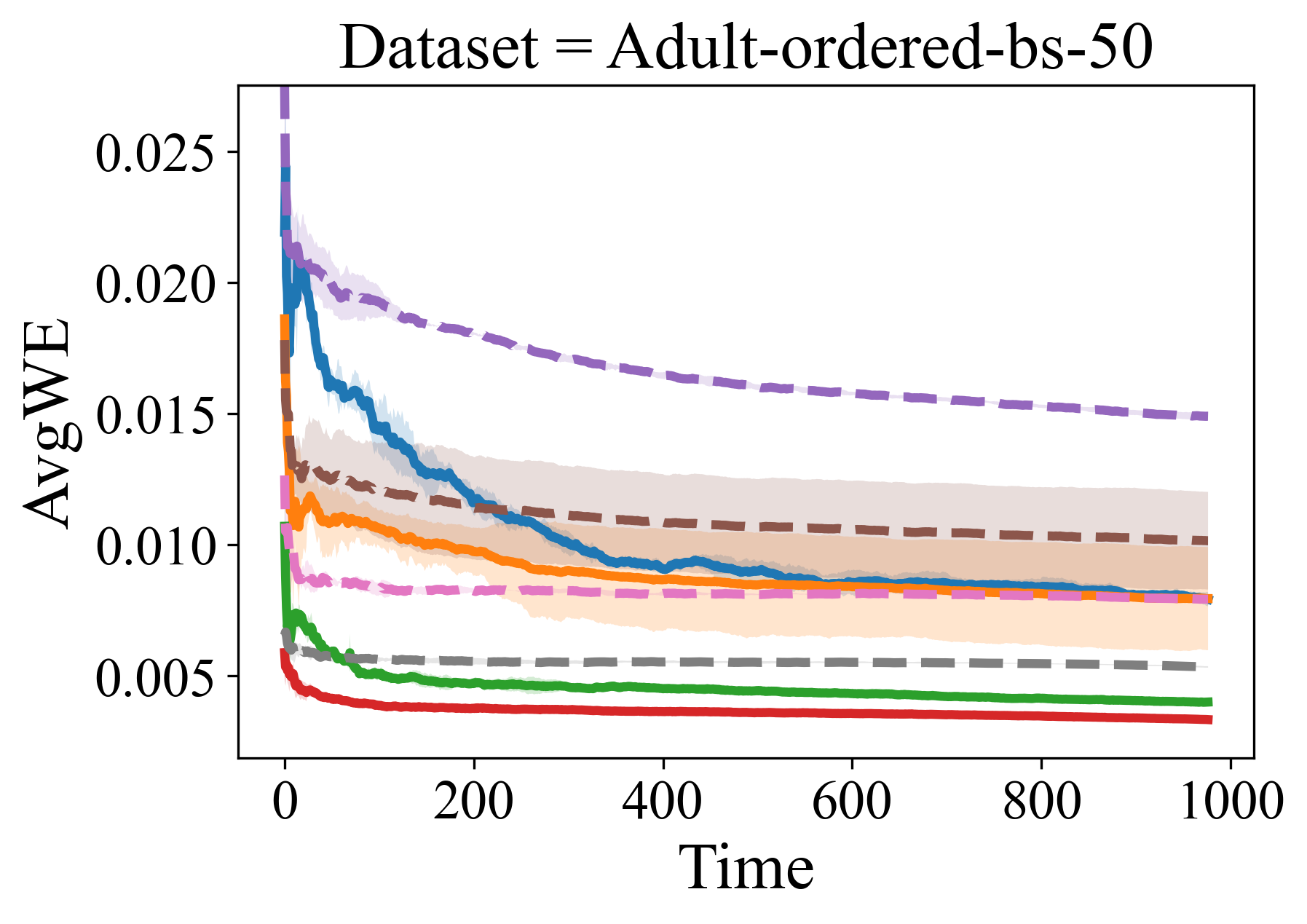}
        \caption{Average of workload errors}
     \end{subfigure}\hspace*{\fill}
     \begin{subfigure}[t]{0.23\linewidth}
         \centering
         \includegraphics[height=0.15\textheight]{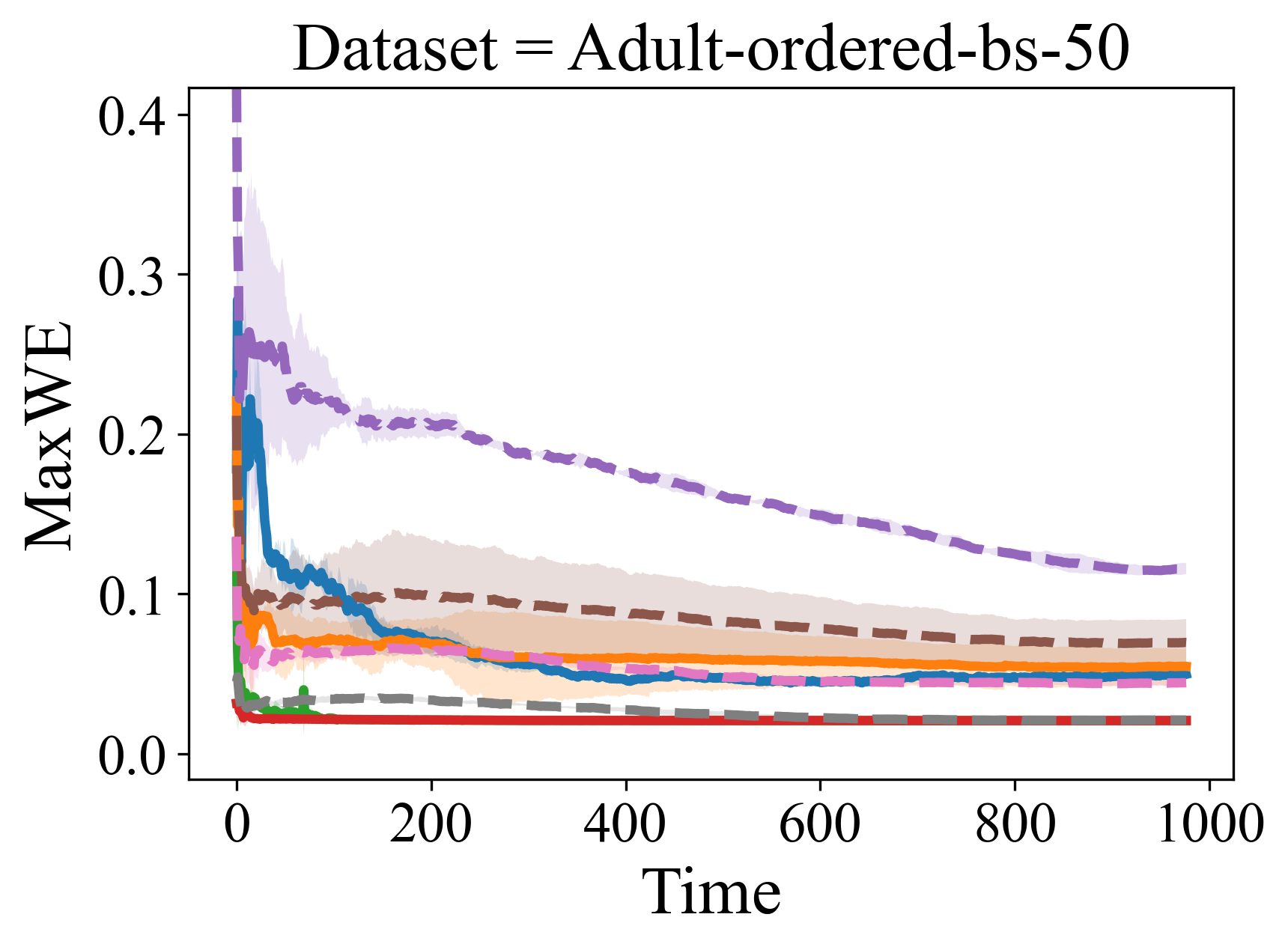}
        \caption{Maximum of workload errors}
     \end{subfigure}\hspace*{\fill}
     \begin{subfigure}[t]{0.23\linewidth}
         \centering
         \includegraphics[height=0.15\textheight]{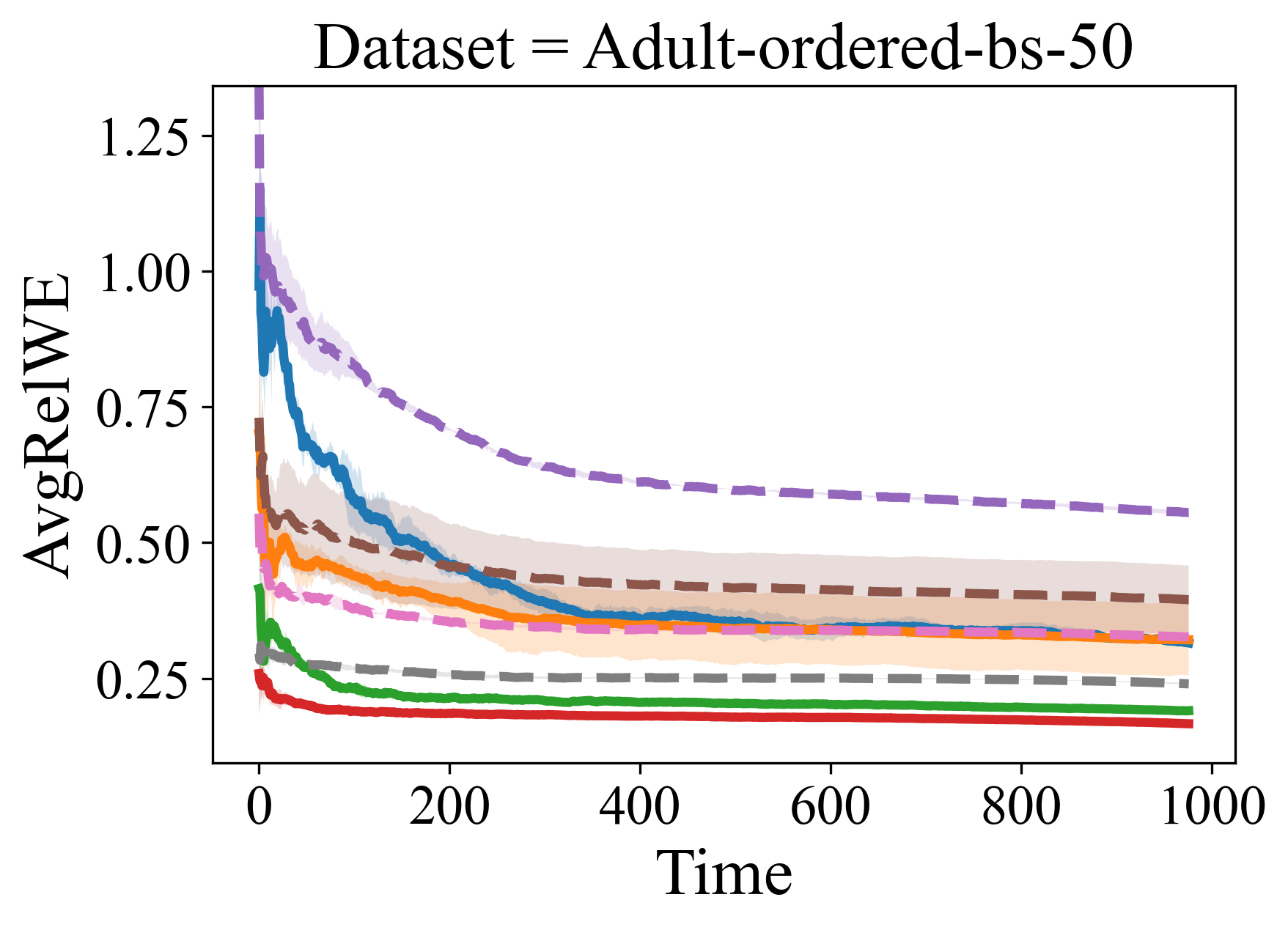}
        \caption{Average of relative workload errors}
     \end{subfigure}\hspace*{\fill}
     \begin{subfigure}[t]{0.23\linewidth}
         \centering
         \includegraphics[height=0.15\textheight]{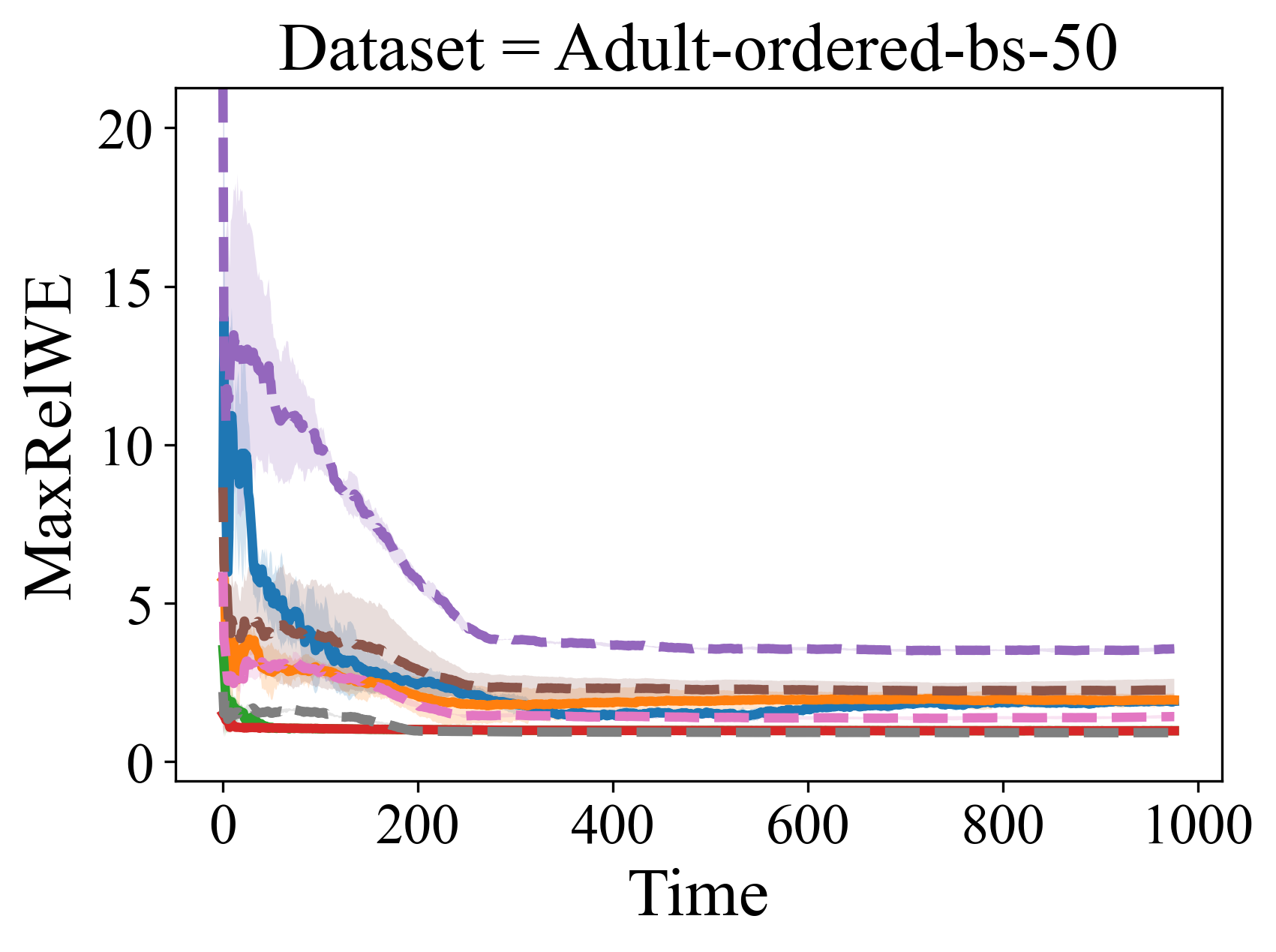}
        \caption{Maximum of relative workload errors}
     \end{subfigure}
    \hfill
     \begin{subfigure}[t]{\linewidth}
         \centering
         \includegraphics[height=0.03\textheight]{figures/legend.png}
     \end{subfigure}
    \caption{Metrics over time to compare the baseline and proposed method for the Adult-ordered-bs-50 dataset.}
    \label{fig:Adult-ordered-bs-50}
\end{figure*}

 \begin{figure*}[ht]
     \centering
     \begin{subfigure}[t]{0.23\linewidth}
         \centering
         \includegraphics[height=0.15\textheight]{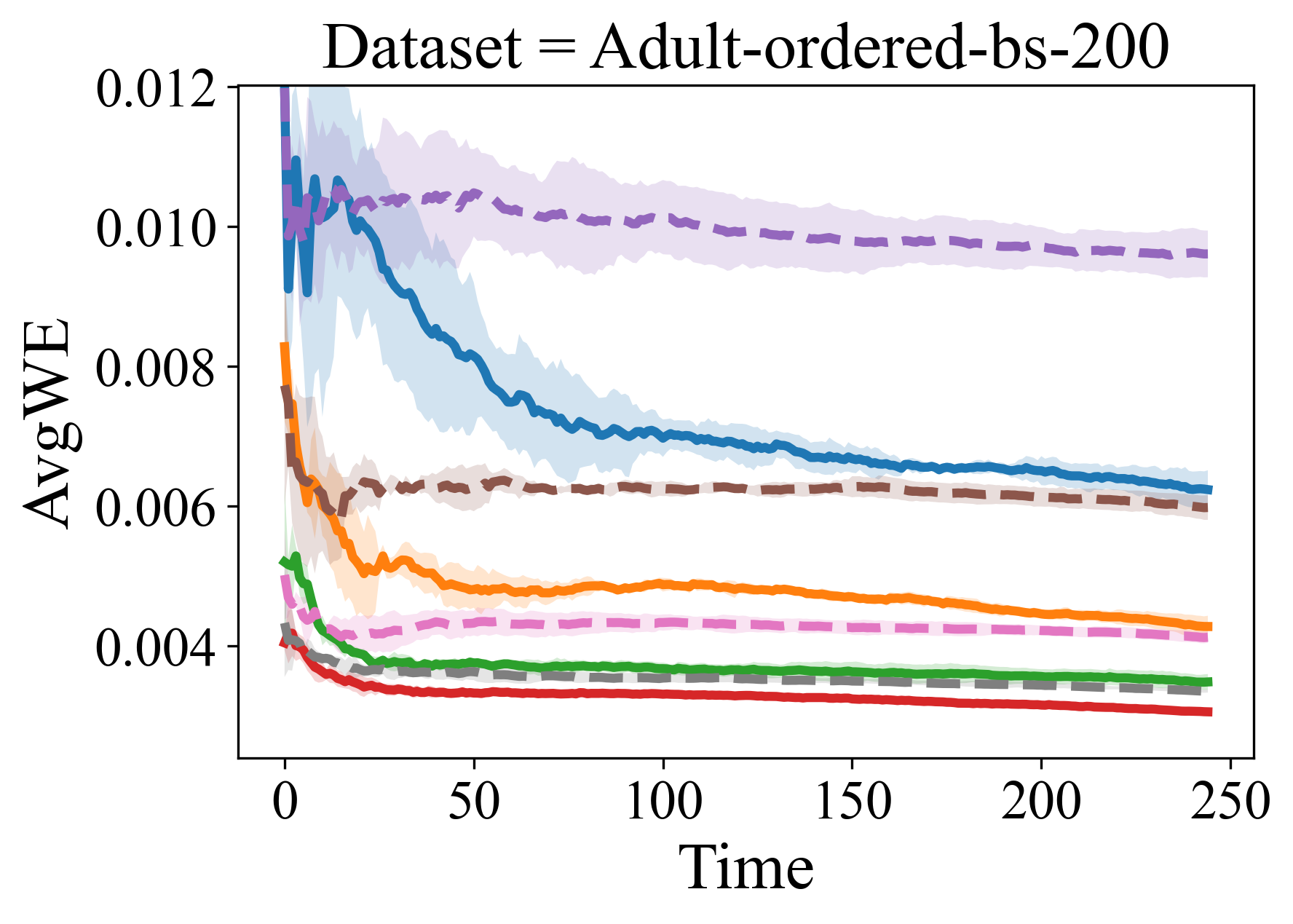}
        \caption{Average of workload errors}
     \end{subfigure}\hspace*{\fill}
     \begin{subfigure}[t]{0.23\linewidth}
         \centering
         \includegraphics[height=0.15\textheight]{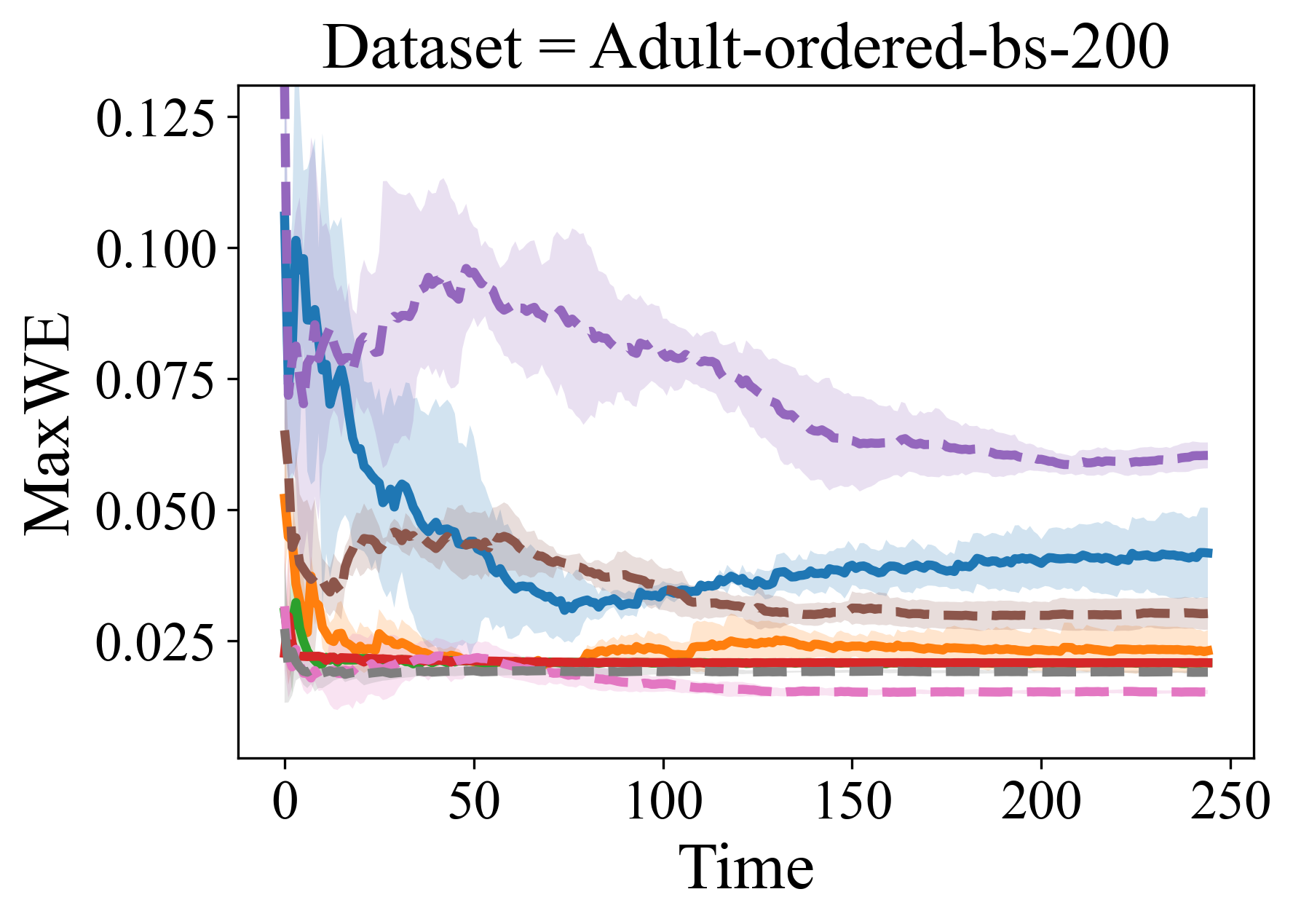}
        \caption{Maximum of workload errors}
     \end{subfigure}\hspace*{\fill}
     \begin{subfigure}[t]{0.23\linewidth}
         \centering
         \includegraphics[height=0.15\textheight]{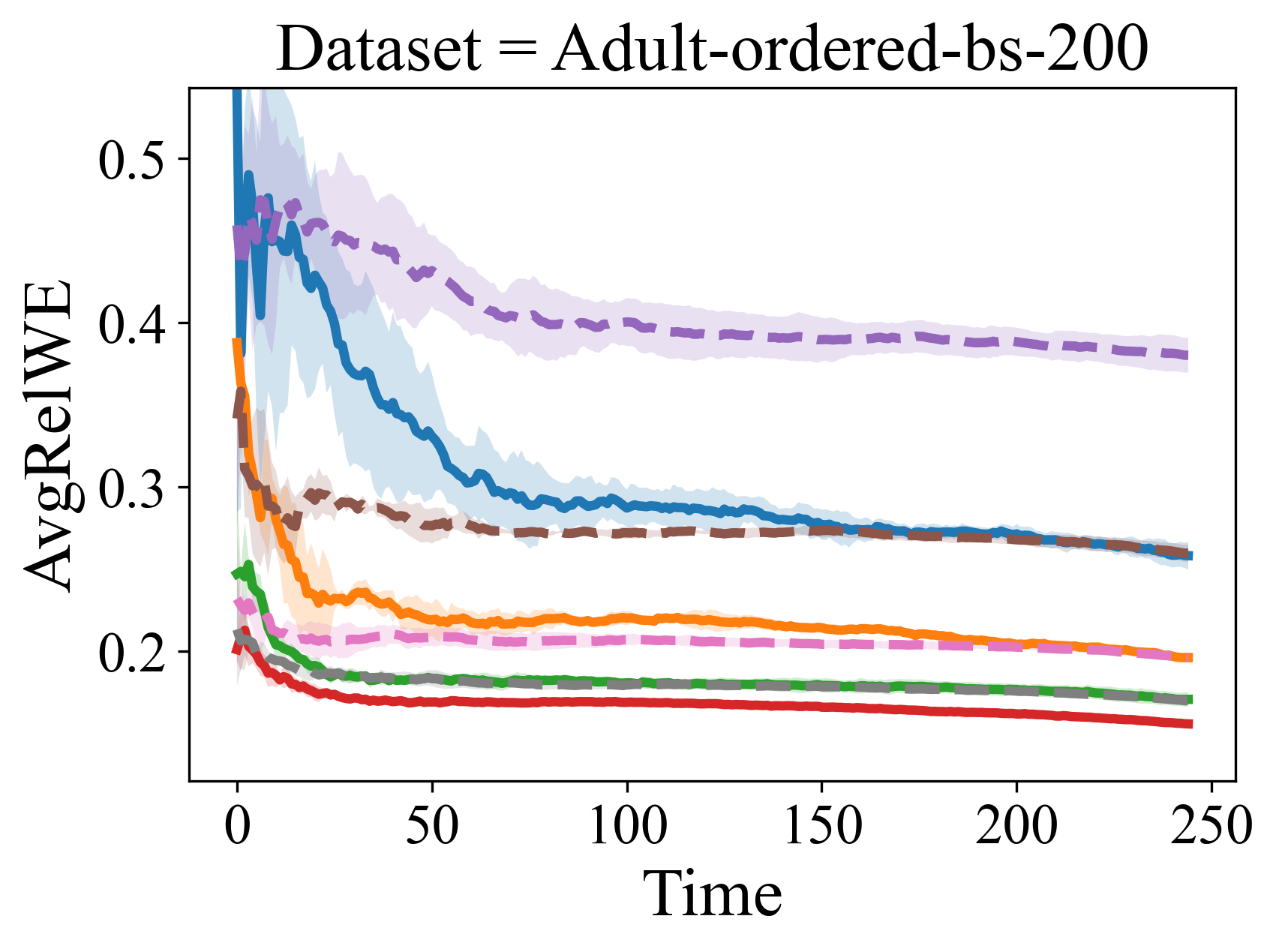}
        \caption{Average of relative workload errors}
     \end{subfigure}\hspace*{\fill}
     \begin{subfigure}[t]{0.23\linewidth}
         \centering
         \includegraphics[height=0.15\textheight]{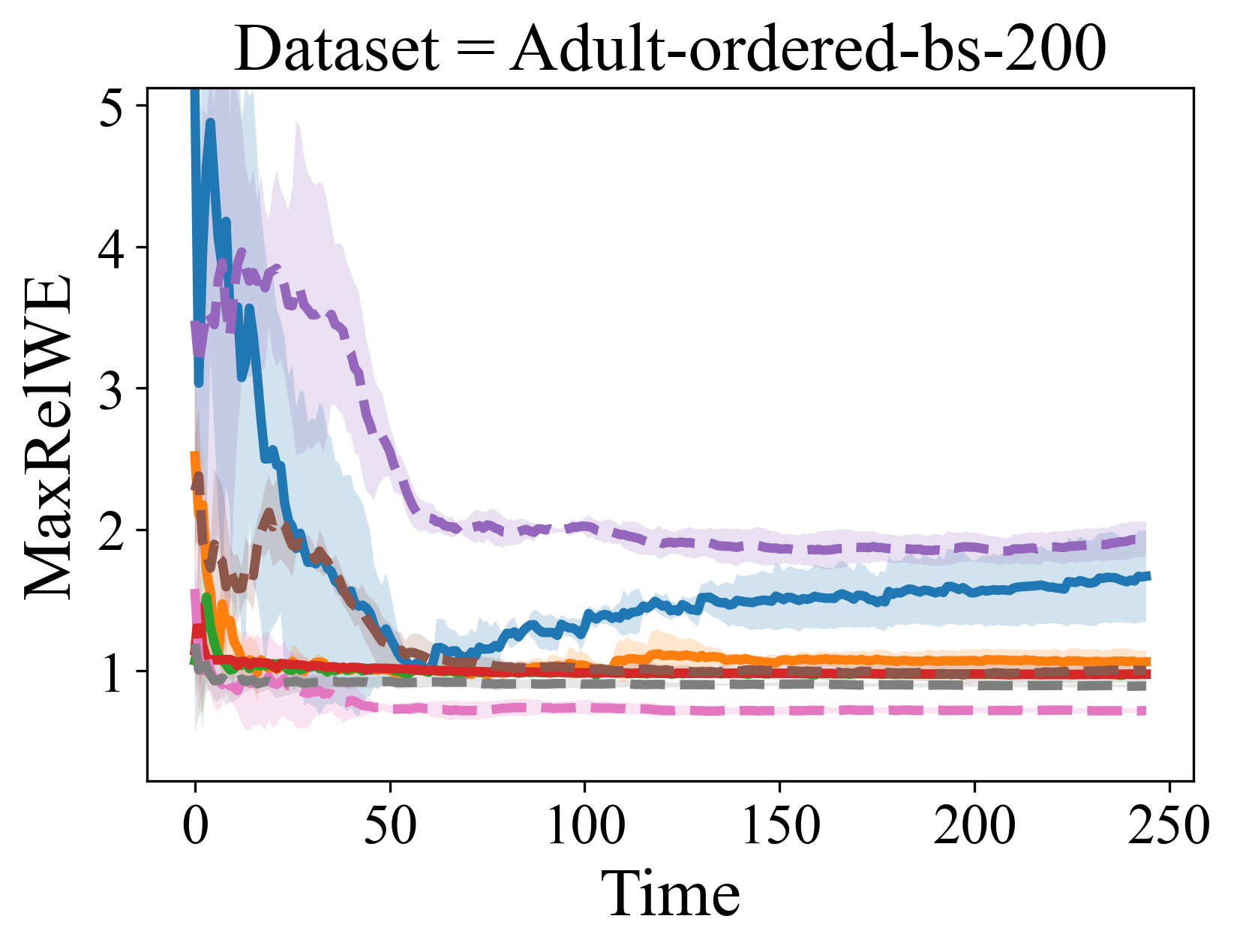}
        \caption{Maximum of relative workload errors}
     \end{subfigure}
    \hfill
     \begin{subfigure}[t]{\linewidth}
         \centering
         \includegraphics[height=0.03\textheight]{figures/legend.png}
     \end{subfigure}
    \caption{Metrics over time to compare the baseline and proposed method for the Adult-ordered-bs-200 dataset.}
    \label{fig:Adult-ordered-bs-200}
\end{figure*}

\begin{table}
    \caption{Query error metrics for the baseline and proposed methods for the {\em Eviction-weekly}, {\em Eviction-bi-weekly}, and {\em Adult-randomized-bs-50} datasets.}
    \label{tab:metrics_tab_1}
    \begin{tabular}{|p{6em}|p{1.5em}|p{4.5em}|p{3.5em}|p{4em}|}
        \hline
        \textbf{Dataset} & $\bm{\e}$ & \textbf{Metric} & \textbf{Baseline} & \textbf{Proposed} \\
        \hline
        \multirow{16}{6em}{Eviction-weekly} & \multirow{4}{1.5em}{0.5} & AvgRelWE & 1.4523 & \textbf{0.6024}
        \\
        &  & AvgWE & 0.0478 & \textbf{0.0191}
        \\
        &  & MaxRelWE & 4.0261 & \textbf{1.7817}
        \\
        &  & MaxWE & 0.218 & \textbf{0.0467}
        \\
        \cline{2-5}
        & \multirow{4}{1.5em}{1.0} & AvgRelWE & 0.6906 & \textbf{0.4409}
        \\
        &  & AvgWE & 0.0219 & \textbf{0.014}
        \\
        &  & MaxRelWE & 1.7988 & \textbf{1.3978}
        \\
        &  & MaxWE & \textbf{0.0378} & 0.0436
        \\
        \cline{2-5}
        & \multirow{4}{1.5em}{2.0} & AvgRelWE & 0.3309 & \textbf{0.2746}
        \\
        &  & AvgWE & 0.0096 & \textbf{0.0079}
        \\
        &  & MaxRelWE & 1.1009 & \textbf{0.7573}
        \\
        &  & MaxWE & 0.032 & \textbf{0.0154}
        \\
        \cline{2-5}
        & \multirow{4}{1.5em}{4.0} & AvgRelWE & \textbf{0.2066} & 0.2632
        \\
        &  & AvgWE & \textbf{0.0054} & 0.0075
        \\
        &  & MaxRelWE & 0.923 & \textbf{0.71}
        \\
        &  & MaxWE & 0.0303 & \textbf{0.0162}
        \\
        \hline
        % ====== Eviction-bi-weekly starts
        \multirow{16}{6em}{Eviction-bi-weekly} & \multirow{4}{1.5em}{0.5} & AvgRelWE & 0.7719 & \textbf{0.4009}
        \\
        &  & AvgWE & 0.0248 & \textbf{0.0121}
        \\
        &  & MaxRelWE & 2.0291 & \textbf{1.6396}
        \\
        &  & MaxWE & \textbf{0.043} & 0.0499
        \\
        \cline{2-5}
        & \multirow{4}{1.5em}{1.0} & AvgRelWE & 0.3487 & \textbf{0.2864}
        \\
        &  & AvgWE & 0.01 & \textbf{0.0082}
        \\
        &  & MaxRelWE & \textbf{1.1644} & 1.1651
        \\
        &  & MaxWE & \textbf{0.0331} & 0.0411
        \\
        \cline{2-5}
        & \multirow{4}{1.5em}{2.0} & AvgRelWE & \textbf{0.2046} & 0.2166
        \\
        &  & AvgWE & \textbf{0.0053} & 0.006
        \\
        &  & MaxRelWE & 1.0008 & \textbf{0.71}
        \\
        &  & MaxWE & 0.0333 & \textbf{0.0141}
        \\
        \cline{2-5}
        & \multirow{4}{1.5em}{4.0} & AvgRelWE & \textbf{0.1571} & 0.2207
        \\
        &  & AvgWE & \textbf{0.0039} & 0.0064
        \\
        &  & MaxRelWE & \textbf{0.6945} & 0.71
        \\
        &  & MaxWE & 0.0238 & \textbf{0.0167}
        \\
        \hline
        % ====== Eviction-bi-weekly ends
        \multirow{16}{6em}{Adult-randomized-bs-50} & \multirow{4}{1.5em}{0.5} & AvgRelWE & 0.5624 & \textbf{0.3035}
        \\
        &  & AvgWE & 0.0151 & \textbf{0.0075}
        \\
        &  & MaxRelWE & 3.6469 & \textbf{1.7575}
        \\
        &  & MaxWE & 0.1182 & \textbf{0.0504}
        \\
        \cline{2-5}
        & \multirow{4}{1.5em}{1.0} & AvgRelWE & 0.4005 & \textbf{0.3191}
        \\
        &  & AvgWE & 0.0103 & \textbf{0.0079}
        \\
        &  & MaxRelWE & 2.3292 & \textbf{1.8162}
        \\
        &  & MaxWE & 0.0712 & \textbf{0.0514}
        \\
        \cline{2-5}
        & \multirow{4}{1.5em}{2.0} & AvgRelWE & 0.351 & \textbf{0.1863}
        \\
        &  & AvgWE & 0.0085 & \textbf{0.0039}
        \\
        &  & MaxRelWE & 1.5646 & \textbf{0.977}
        \\
        &  & MaxWE & 0.0488 & \textbf{0.0208}
        \\
        \cline{2-5}
        & \multirow{4}{1.5em}{4.0} & AvgRelWE & 0.2606 & \textbf{0.1673}
        \\
        &  & AvgWE & 0.0058 & \textbf{0.0033}
        \\
        &  & MaxRelWE & \textbf{0.9187} & 0.9778
        \\
        &  & MaxWE & 0.0243 & \textbf{0.0208}
        \\
        \hline
    \end{tabular}
\end{table}

\begin{table}
    \caption{Query error metrics for the baseline and proposed methods for the {\em Adult-randomized-bs-200}, {\em Adult-ordered-bs-50}, and {\em Adult-ordered-bs-200} datasets.}
    \label{tab:metrics_tab_2}
    \begin{tabular}{|p{6em}|p{1.5em}|p{4.5em}|p{3.5em}|p{4em}|}
        \hline
        \textbf{Dataset} & $\bm{\e}$ & \textbf{Metric} & \textbf{Baseline} & \textbf{Proposed} \\
        \hline
        \multirow{16}{6em}{Adult-randomized-bs-200} & \multirow{4}{1.5em}{0.5} & AvgRelWE & 0.4113 & \textbf{0.2658}
        \\
        &  & AvgWE & 0.0104 & \textbf{0.0064}
        \\
        &  & MaxRelWE & 2.1987 & \textbf{1.6246}
        \\
        &  & MaxWE & 0.0698 & \textbf{0.0419}
        \\
        \cline{2-5}
        & \multirow{4}{1.5em}{1.0} & AvgRelWE & 0.2786 & \textbf{0.1975}
        \\
        &  & AvgWE & 0.0065 & \textbf{0.0044}
        \\
        &  & MaxRelWE & \textbf{1.0789} & 1.1166
        \\
        &  & MaxWE & 0.0323 & \textbf{0.0249}
        \\
        \cline{2-5}
        & \multirow{4}{1.5em}{2.0} & AvgRelWE & 0.1802 & \textbf{0.1802}
        \\
        &  & AvgWE & 0.0036 & \textbf{0.0036}
        \\
        &  & MaxRelWE & 0.8907 & \textbf{0.8907}
        \\
        &  & MaxWE & 0.0191 & \textbf{0.0191}
        \\
        \cline{2-5}
        & \multirow{4}{1.5em}{4.0} & AvgRelWE & 0.1802 & \textbf{0.1802}
        \\
        &  & AvgWE & 0.0036 & \textbf{0.0036}
        \\
        &  & MaxRelWE & 0.8907 & \textbf{0.8907}
        \\
        &  & MaxWE & 0.0191 & \textbf{0.0191}
        \\
        \hline
        \multirow{16}{6em}{Adult-ordered-bs-50} & \multirow{4}{1.5em}{0.5} & AvgRelWE & 0.5559 & \textbf{0.3165}
        \\
        &  & AvgWE & 0.0149 & \textbf{0.008}
        \\
        &  & MaxRelWE & 3.5589 & \textbf{1.9028}
        \\
        &  & MaxWE & 0.1157 & \textbf{0.0497}
        \\
        \cline{2-5}
        & \multirow{4}{1.5em}{1.0} & AvgRelWE & 0.3953 & \textbf{0.3214}
        \\
        &  & AvgWE & 0.0102 & \textbf{0.008}
        \\
        &  & MaxRelWE & 2.2663 & \textbf{1.9528}
        \\
        &  & MaxWE & 0.0695 & \textbf{0.0547}
        \\
        \cline{2-5}
        & \multirow{4}{1.5em}{2.0} & AvgRelWE & 0.3264 & \textbf{0.1904}
        \\
        &  & AvgWE & 0.0079 & \textbf{0.004}
        \\
        &  & MaxRelWE & 1.4234 & \textbf{0.9798}
        \\
        &  & MaxWE & 0.0444 & \textbf{0.0208}
        \\
        \cline{2-5}
        & \multirow{4}{1.5em}{4.0} & AvgRelWE & 0.2402 & \textbf{0.1667}
        \\
        &  & AvgWE & 0.0053 & \textbf{0.0033}
        \\
        &  & MaxRelWE & \textbf{0.9192} & 0.978
        \\
        &  & MaxWE & 0.0211 & \textbf{0.0208}
        \\
        \hline
        \multirow{16}{6em}{Adult-ordered-bs-200} & \multirow{4}{1.5em}{0.5} & AvgRelWE & 0.381 & \textbf{0.2593}
        \\
        &  & AvgWE & 0.0096 & \textbf{0.0063}
        \\
        &  & MaxRelWE & 1.9203 & \textbf{1.6522}
        \\
        &  & MaxWE & 0.0601 & \textbf{0.0413}
        \\
        \cline{2-5}
        & \multirow{4}{1.5em}{1.0} & AvgRelWE & 0.2608 & \textbf{0.1972}
        \\
        &  & AvgWE & 0.006 & \textbf{0.0043}
        \\
        &  & MaxRelWE & \textbf{1.006} & 1.067
        \\
        &  & MaxWE & 0.0302 & \textbf{0.0232}
        \\
        \cline{2-5}
        & \multirow{4}{1.5em}{2.0} & AvgRelWE & 0.1969 & \textbf{0.1711}
        \\
        &  & AvgWE & 0.0041 & \textbf{0.0035}
        \\
        &  & MaxRelWE & \textbf{0.7189} & 0.9752
        \\
        &  & MaxWE & \textbf{0.0152} & 0.0208
        \\
        \cline{2-5}
        & \multirow{4}{1.5em}{4.0} & AvgRelWE & 0.1705 & \textbf{0.1563}
        \\
        &  & AvgWE & 0.0034 & \textbf{0.0031}
        \\
        &  & MaxRelWE & \textbf{0.8939} & 0.9776
        \\
        &  & MaxWE & \textbf{0.019} & 0.0208
        \\
        \hline
    \end{tabular}
\end{table}

%=================================
%=================================
%========= End of plots 
%=================================
%=================================

We explore the performance of our algorithm on real-world datasets. We use the Probabilistic Graphical Model (PGM) \cite{tabular_pgm_mckenna19a} as the subroutine $\AA_{Dataset}$ for both the baseline and proposed algorithm. We briefly introduced PGM in Section~\ref{s:offline_pgm}. Similar to \cite{kumar_algorithm_2024}, we found that Simple counter works the best for our use case as we do not have a very large time horizon for the stream. So, for the experiments in this section, we use the Simple counter as the subroutine $\AA_{Counter}$. In the subsequent subsections, we discuss the details of over experiments.

\subsection{Datasets}
\subsubsection{Eviction}
The Eviction Dataset \cite{sfgovEvictionNotices} contains eviction notices filed with the San Francisco Rent Board from January 1, 1997. The dataset has an attribute ``File Date'' which represents the date on which the eviction notice was filed with the Rent Board of Arbitration. We use the value of this attribute to construct our time index for the stream. We fix a synthetic data release frequency as weekly or bi-weekly, and based on the attribute File Date and this frequency, create the time index for our data. In the results, we refer to the corresponding streams as \textit{Eviction-weekly} and \textit{Eviction-bi-weekly} respectively. We limit the dataset to $3$ location-based categorical attributes - ``Eviction Notice Source Zipcode'', ``Supervisor District'', and ``Neighborhoods'', and all binary attributes such as - ``Non Payment'', ``Breach'', and ``Illegal Use''. The data space of the resulting dataset was $22$ dimensional with $19$ binary and $3$ categorical attributes.

\subsubsection{Adult}
The {\em Adult} dataset \cite{adultDataset} has been used extensively in previous research in this area and so we also use this dataset. We use a processed version of the dataset released in the source code provided by \cite{Liu2021IterativeMF}. Note that there is no notion of time in this dataset. We artificially create time in two ways which results in the following two streams: (1) \textit{Adult-randomized}: we fix a constant batch size say $B$, that is the number of points that are added at each time, then at any time $t$, we simply add $B$ points to our stream that are sampled uniformly at random from the Adult dataset without replacement; (2) \textit{Adult-ordered}: the stream is also created by adding a fixed batch of $B$ points, except the points are selected deterministically, where we first sort the entire dataset in increasing order and then add the next $B$ points based on the indices at any time. We explore a small and large value of batch size $B$ as $50$ and $200$ respectively. Note that the data stream Adult-ordered is interesting in the sense that the query values may change very drastically over time.

\subsection{Workload of queries}
In the experiments, we aim to preserve all $2$-way marginals on the space $\XX$. However, instead of using the set of all $2$-way marginals queries as $Q$, we use the set of all $2$-way {\em workloads}. 
\begin{definition}[Workload]\label{def:workload}
    A $k$-way workload $W$ is defined by a tuple of $k$ column indices, $\bp{c_1, c_2, \ldots, c_k}$ such that $c_i \in [p]$ for all $i \in [k]$ and $c_1 <c_2<\ldots<c_k$. Let $columns(W) = \bp{c_1, c_2, \ldots, c_k}$. Then, $W$ is an ordered collection of all $k$-way marginal queries on $columns(W)$, where the order is taken as the lexicographic order of the values corresponding to queries. Furthermore, we denote the number of marginal queries in $W$ with $\abs{W}$. The value of a workload over a dataset $f:\XX\to\N_0$ is defined as the tuple $W(f)=\bp{q(f)}_{q\in W}$.
\end{definition}

Using workloads instead of marginal queries is a common practice in the literature and is sometimes referred to as the {\em marginal trick} \cite{Liu2021IterativeMF}. To see the advantage, note that the collections of queries in a workload create a disjoint space in the sense that a user can contribute data in at most one of them. Thus we can use the Parallel Composition of differential privacy and do not have to divide the privacy budget among the queries in a workload when estimating them. In other words, estimating any workload is a histogram query with sensitivity $1$ for any pair of neighboring datasets. This is true even though there are (likely) multiple queries in the workload.

Note that since we are now using workloads instead of marginal queries, the following changes are needed in Algorithm~\ref{alg:main} to ensure compatibility,
\begin{enumerate}
    \item $Q$ is an ordered set of workloads such that for any $i\in\abs{Q}$, $\abs{q_i}$ denotes the number of marginal queries in $q_i$;
    \item each counter instance $C_i$, corresponding to the workload $q_i$, is a multi-dimensional counter of dimension $\abs{q_i}$, as defined in \cite{kumar_algorithm_2024};
    \item at an iteration $l$ of time $t$, we define $e_{t, l}$ as
    \begin{equation*}
        \begin{split}
            e_{t, l} = \bp{ 
            \begin{aligned}
                & \frac{1}{\abs{q_i}} \norm{ q_i(\nabla f_t + g_{t-1}) - q_i(h_{t, l-1}) }_{\ell_1} 
                \\ & \quad - \abs{\XX_{columns(q_i)}}
            \end{aligned}
            }_{i \in J_{t, l}},
        \end{split}
    \end{equation*}
    where $\abs{\XX_{columns(q_i)}}$ is a bias correction term accounting for the number of queries in a workload;
    \item for $2$-way workloads, the resulting sensitivity of the exponential mechanism is $\frac{1}{4}$.
\end{enumerate}

\subsection{Evaluation metrics}
We measure the performance of our algorithm using the error introduced by the generated synthetic data in answering queries. Since we use workloads in our algorithm, we measure the error in queries grouped by the workloads. This results in two levels of aggregation, one for queries within a workload and the second over different workloads.

Let us assume that we are looking for error at time $t$ for the synthetic stream $g$ given the input stream $f$ and the set of workloads $Q$. Then, we define the workload errors as follows:

\begin{enumerate}
    \item \textit{Workload error (WE)}: For any workload $W$, we define workload error at time $t$ as the average error in queries within $W$, that is
    $$
        WE(W, f, g, t) \coloneqq \frac{1}{\abs{W}} \sum_{q\in W} \abs{q(f_t) - q(g_t)}.
    $$
    \item \textit{Relative workload error (RelWE)}: Similar to workload error, for any workload $W$, we define relative workload error at time $t$ as the average relative error in queries within $W$, that is
    $$
        RelWE(W, f, g, t) \coloneqq \frac{1}{\abs{W}} \sum_{q\in W} \abs{ \frac{q(f_t) - q(g_t)}{q(f_t)} }.
    $$
\end{enumerate}

Our final metric is the aggregated (average and maximum) error over all workloads. Given a set $Q$ containing workloads, an input stream $f$, and a synthetic stream $g$, at any time $t$ we define:

\begin{enumerate}
    \item \textit{Average over workload errors (AvgWE)}: as the average workload error over workloads in $Q$, that is
    $$
        AvgWE(Q, f, g, t) \coloneqq \frac{1}{\abs{Q}} \sum_{W\in Q} WE(W, f, g, t).
    $$
    \item \textit{Maximum over workload errors (MaxWE)}: as the maximum workload error over workloads in $Q$, that is
    $$
        MaxWE(Q, f, g, t) \coloneqq \max_{W\in Q} \bp{ WE(W, f, g, t) }.
    $$
    \item \textit{Average over relative workload errors (AvgRelWE)}: as the average relative workload error over workloads in $Q$, that is
    $$
        AvgRelWE(Q, f, g, t) \coloneqq \frac{1}{\abs{Q}} \sum_{W\in Q} RelWE(W, f, g, t).
    $$
    \item \textit{Maximum over relative workload errors (MaxRelWE)}: as the maximum relative workload error over workloads in $Q$, that is
    $$
        MaxRelWE(Q, f, g, t) \coloneqq \max_{W\in Q} \bp{ RelWE(W, f, g, t) }.
    $$
\end{enumerate}

\subsection{Results}\label{s:results}

Figures~\ref{fig:Eviction-weekly}, \ref{fig:Eviction-bi-weekly}, \ref{fig:Adult-randomized-bs-50}, \ref{fig:Adult-randomized-bs-200}, \ref{fig:Adult-ordered-bs-50}, and~\ref{fig:Adult-ordered-bs-200} provide our results for the Eviction-weekly, Eviction-bi-weekly, Adult-randomized-bs-50, Adult-randomized-bs-200, Adult-ordered-bs-50, and Adult-ordered-bs-200 datasets respectively. The horizontal axis in all of these figures represents time and the vertical axis is the metric mentioned in the y-axis label of the corresponding subfigure. At the beginning of time, we see a large variance in the metrics, and the proposed method has a larger error in some experiments. However, as the time index increases, in most cases, our method outperforms the baseline across various datasets, metrics, and privacy budgets. Moreover, we observe that among the various metrics, the most variation occurs in metrics that measure the worst-case errors: MaxWE and MaxRelWE. We also provide these metrics in a tabular view in Tables~\ref{tab:metrics_tab_1} and \ref{tab:metrics_tab_2} to facilitate the comparison. These tables contain the value of each metric averaged over the last $10$ time steps for various datasets and privacy budgets $\e$.

%% file: 6_unbounded_block_counter.tex
%==============================================
%========= Unbounded block Counter ==========
%==============================================

\section{A new (unbounded) Block counter} \label{s:block_counter}

\begin{table}
    \caption{Query error metrics comparing the simple and block counters in the proposed method for the {\em Eviction-weekly}, {\em Adult-randomized-bs-50}, and {\em Adult-ordered-bs-50} datasets.}
    \label{tab:metrics_block_v_simple}
    \begin{tabular}{|p{6em}|p{1.5em}|p{4.5em}|p{3.5em}|p{4em}|}
        \hline
        \textbf{Dataset} & $\bm{\e}$ & \textbf{Metric} & \textbf{Simple} & \textbf{Block} \\
        \hline
        \multirow{16}{6em}{Eviction-weekly} & \multirow{4}{1.5em}{0.5} & AvgRelWE & 0.6024 & \textbf{0.2545}
        \\
         &  & AvgWE & 0.0191 & \textbf{0.0092}
        \\
         &  & MaxRelWE & 1.7817 & \textbf{1.4485}
        \\
         &  & MaxWE & \textbf{0.0467} & 0.0507
        \\
        \cline{2-5}
         & \multirow{4}{1.5em}{1.0} & AvgRelWE & 0.4409 & \textbf{0.2553}
        \\
         &  & AvgWE & 0.014 & \textbf{0.0136}
        \\
         &  & MaxRelWE & 1.3978 & \textbf{1.2361}
        \\
         &  & MaxWE & \textbf{0.0436} & 0.057
        \\
        \cline{2-5}
         & \multirow{4}{1.5em}{2.0} & AvgRelWE & 0.2746 & \textbf{0.2376}
        \\
         &  & AvgWE & \textbf{0.0079} & 0.0126
        \\
         &  & MaxRelWE & \textbf{0.7573} & 1.121
        \\
         &  & MaxWE & \textbf{0.0154} & 0.044
        \\
        \cline{2-5}
         & \multirow{4}{1.5em}{4.0} & AvgRelWE & 0.2632 & \textbf{0.2284}
        \\
         &  & AvgWE & \textbf{0.0075} & 0.012
        \\
         &  & MaxRelWE & 0.71 & \textbf{0.71}
        \\
         &  & MaxWE & \textbf{0.0162} & 0.0411
        \\
        \hline
        \multirow{16}{6em}{Adult-randomized-bs-50} & \multirow{4}{1.5em}{0.5} & AvgRelWE & \textbf{0.3035} & 0.3718
        \\
         &  & AvgWE & \textbf{0.0075} & 0.0092
        \\
         &  & MaxRelWE & \textbf{1.7575} & 1.9395
        \\
         &  & MaxWE & \textbf{0.0504} & 0.0525
        \\
        \cline{2-5}
         & \multirow{4}{1.5em}{1.0} & AvgRelWE & 0.3191 & \textbf{0.2534}
        \\
         &  & AvgWE & 0.0079 & \textbf{0.0058}
        \\
         &  & MaxRelWE & 1.8162 & \textbf{1.5294}
        \\
         &  & MaxWE & 0.0514 & \textbf{0.0383}
        \\
        \cline{2-5}
         & \multirow{4}{1.5em}{2.0} & AvgRelWE & 0.1863 & \textbf{0.1862}
        \\
         &  & AvgWE & 0.0039 & \textbf{0.0039}
        \\
         &  & MaxRelWE & \textbf{0.977} & 0.9804
        \\
         &  & MaxWE & 0.0208 & \textbf{0.0208}
        \\
        \cline{2-5}
         & \multirow{4}{1.5em}{4.0} & AvgRelWE & 0.1673 & \textbf{0.1629}
        \\
         &  & AvgWE & 0.0033 & \textbf{0.0032}
        \\
         &  & MaxRelWE & \textbf{0.9778} & 0.9779
        \\
         &  & MaxWE & 0.0208 & \textbf{0.0208}
        \\
        \hline
        \multirow{16}{6em}{Adult-ordered-bs-50} & \multirow{4}{1.5em}{0.5} & AvgRelWE & \textbf{0.3165} & 0.3698
        \\
         &  & AvgWE & \textbf{0.008} & 0.0093
        \\
         &  & MaxRelWE & \textbf{1.9028} & 2.2472
        \\
         &  & MaxWE & \textbf{0.0497} & 0.0572
        \\
        \cline{2-5}
         & \multirow{4}{1.5em}{1.0} & AvgRelWE & 0.3214 & \textbf{0.2552}
        \\
         &  & AvgWE & 0.008 & \textbf{0.0059}
        \\
         &  & MaxRelWE & 1.9528 & \textbf{1.3285}
        \\
         &  & MaxWE & 0.0547 & \textbf{0.0338}
        \\
        \cline{2-5}
         & \multirow{4}{1.5em}{2.0} & AvgRelWE & 0.1904 & \textbf{0.1864}
        \\
         &  & AvgWE & 0.004 & \textbf{0.0039}
        \\
         &  & MaxRelWE & \textbf{0.9798} & 1.0114
        \\
         &  & MaxWE & 0.0208 & \textbf{0.0208}
        \\
        \cline{2-5}
         & \multirow{4}{1.5em}{4.0} & AvgRelWE & 0.1667 & \textbf{0.1632}
        \\
         &  & AvgWE & 0.0033 & \textbf{0.0032}
        \\
         &  & MaxRelWE & 0.978 & \textbf{0.978}
        \\
         &  & MaxWE & 0.0208 & \textbf{0.0208}
        \\
        \hline
    \end{tabular}
\end{table}

 \begin{figure*}[ht]
     \centering
     \begin{subfigure}[t]{0.23\linewidth}
         \centering
         \includegraphics[height=0.14\textheight]{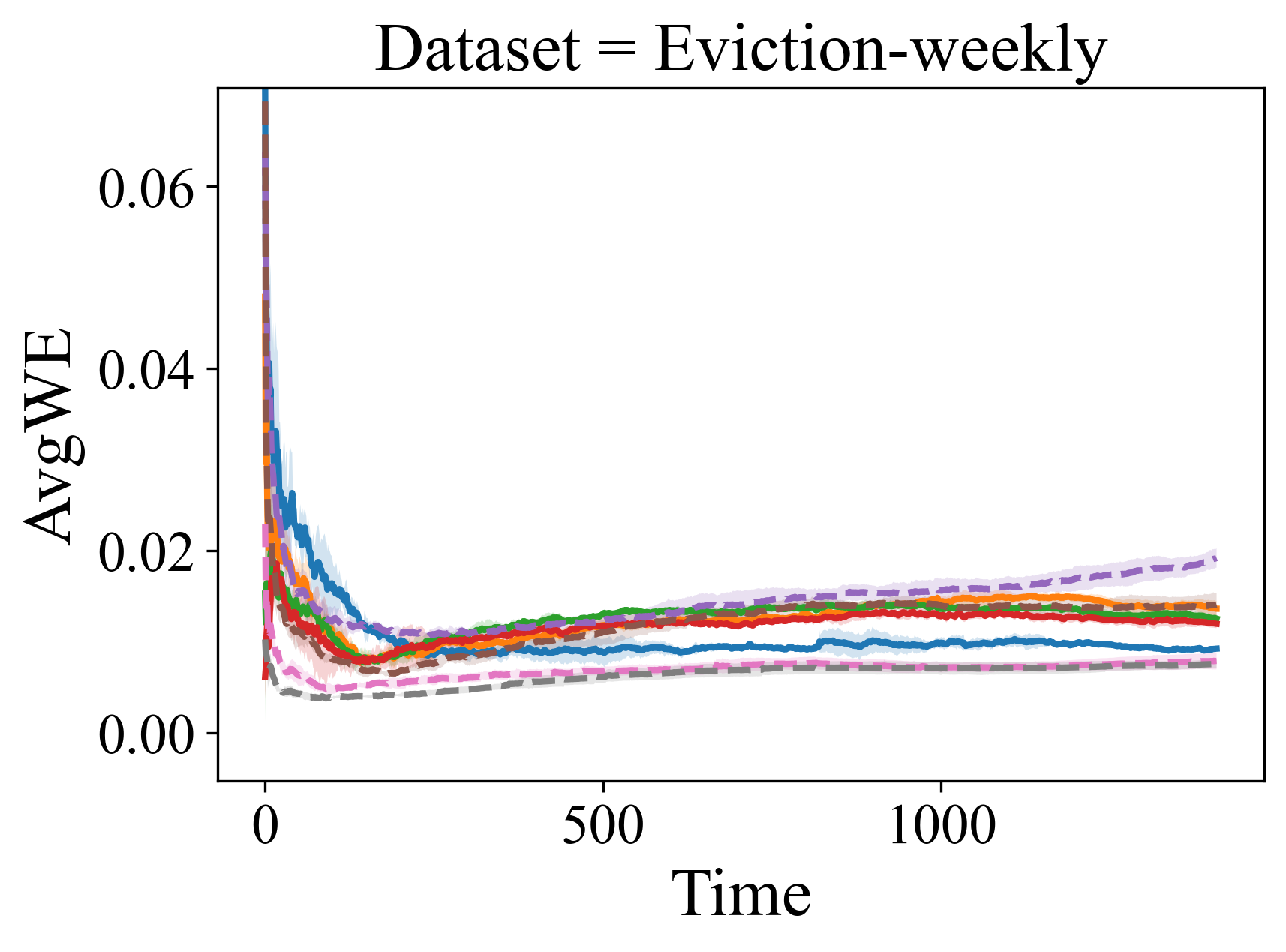}
        \caption{Average of workload errors}
     \end{subfigure}\hspace*{\fill}
     \begin{subfigure}[t]{0.23\linewidth}
         \centering
         \includegraphics[height=0.14\textheight]{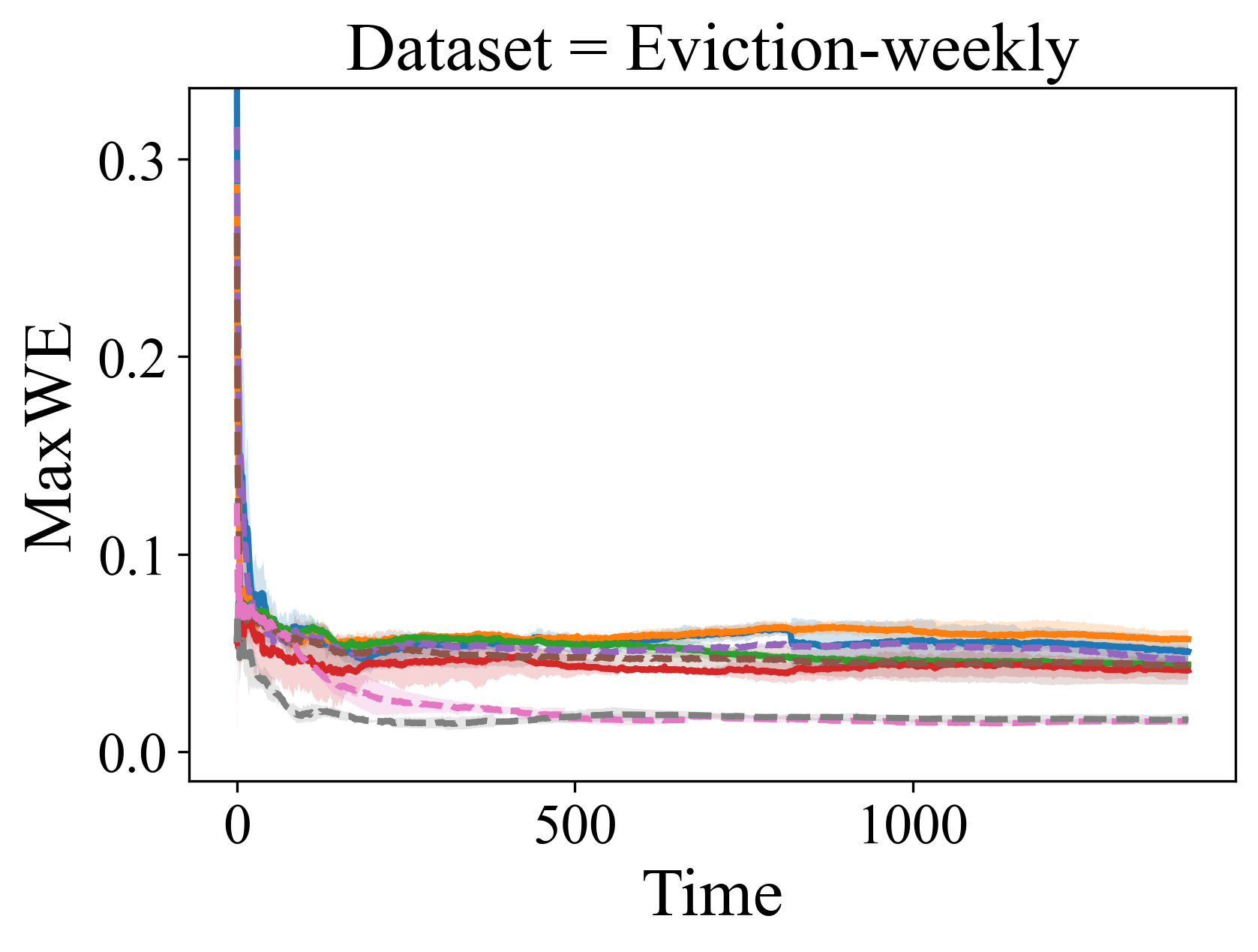}
        \caption{Maximum of workload errors}
     \end{subfigure}\hspace*{\fill}
     \begin{subfigure}[t]{0.23\linewidth}
         \centering
         \includegraphics[height=0.14\textheight]{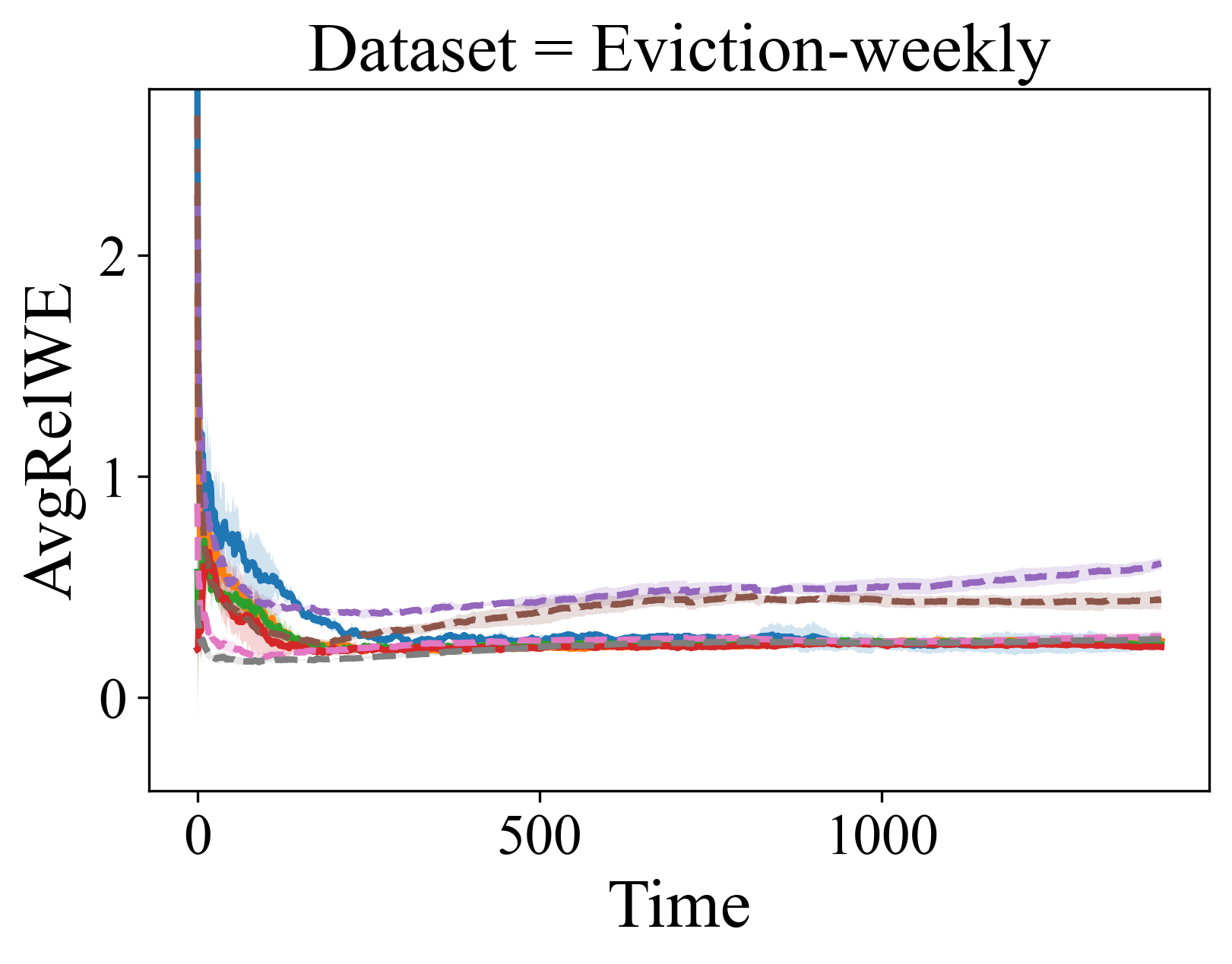}
        \caption{Average of relative workload errors}
     \end{subfigure}\hspace*{\fill}
     \begin{subfigure}[t]{0.23\linewidth}
         \centering
         \includegraphics[height=0.14\textheight]{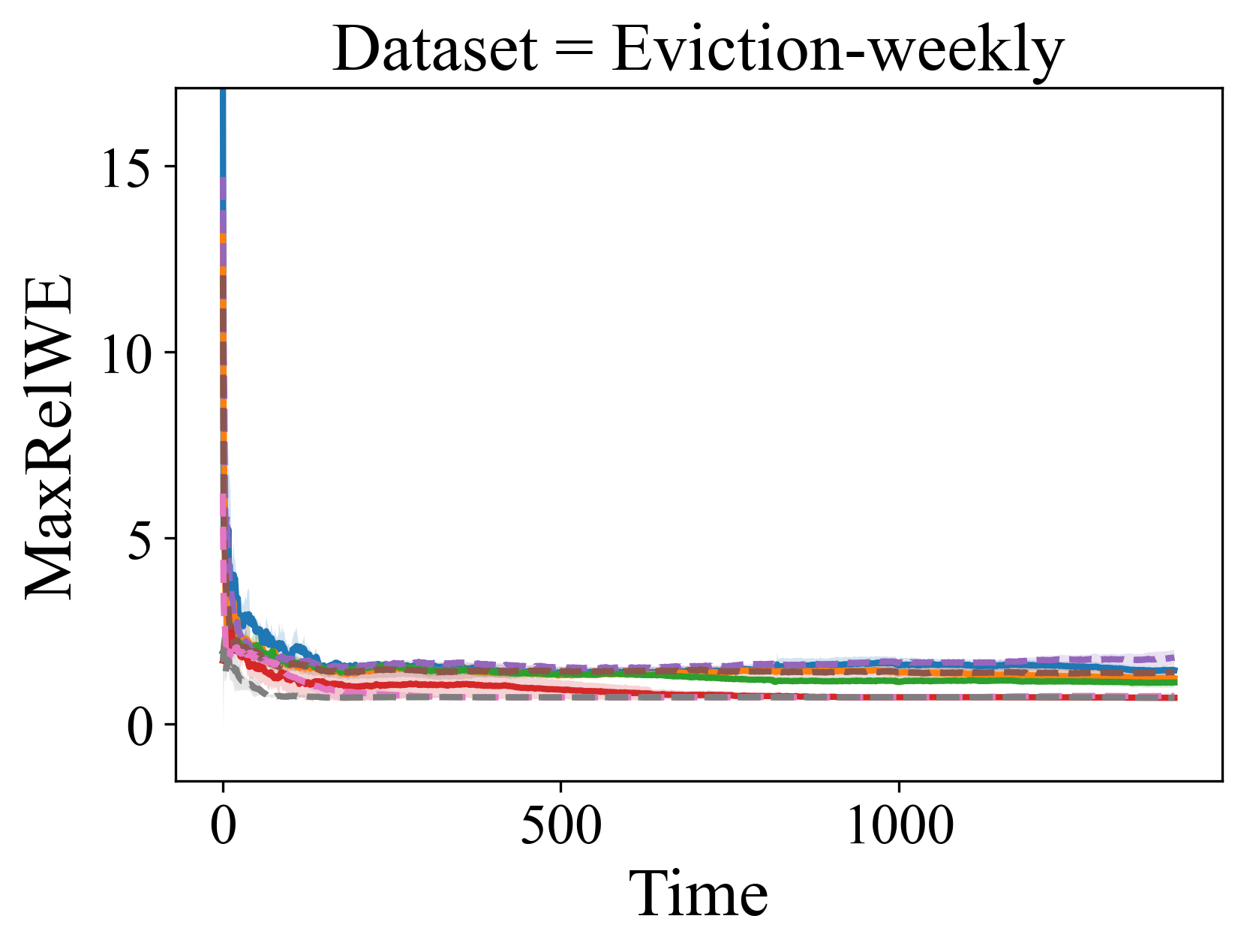}
        \caption{Maximum of relative workload errors}
     \end{subfigure}
    \hfill
     \begin{subfigure}[t]{\linewidth}
         \centering
         \includegraphics[height=0.03\textheight]{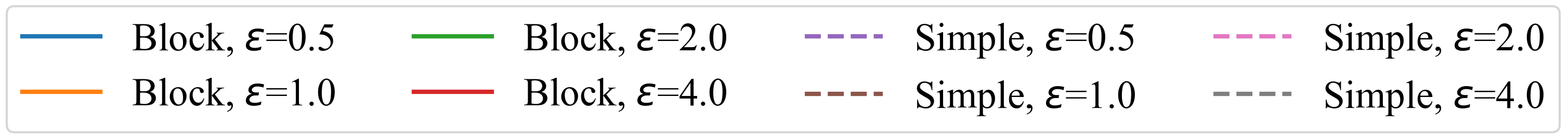}
     \end{subfigure}
    \caption{Metrics over time to compare the performance of simple and block counters for the Eviction-weekly dataset.}
    \label{fig:Eviction-weekly_simple_vs_block}
\end{figure*}

 \begin{figure*}[ht]
     \centering
     \begin{subfigure}[t]{0.23\linewidth}
         \centering
         \includegraphics[height=0.14\textheight]{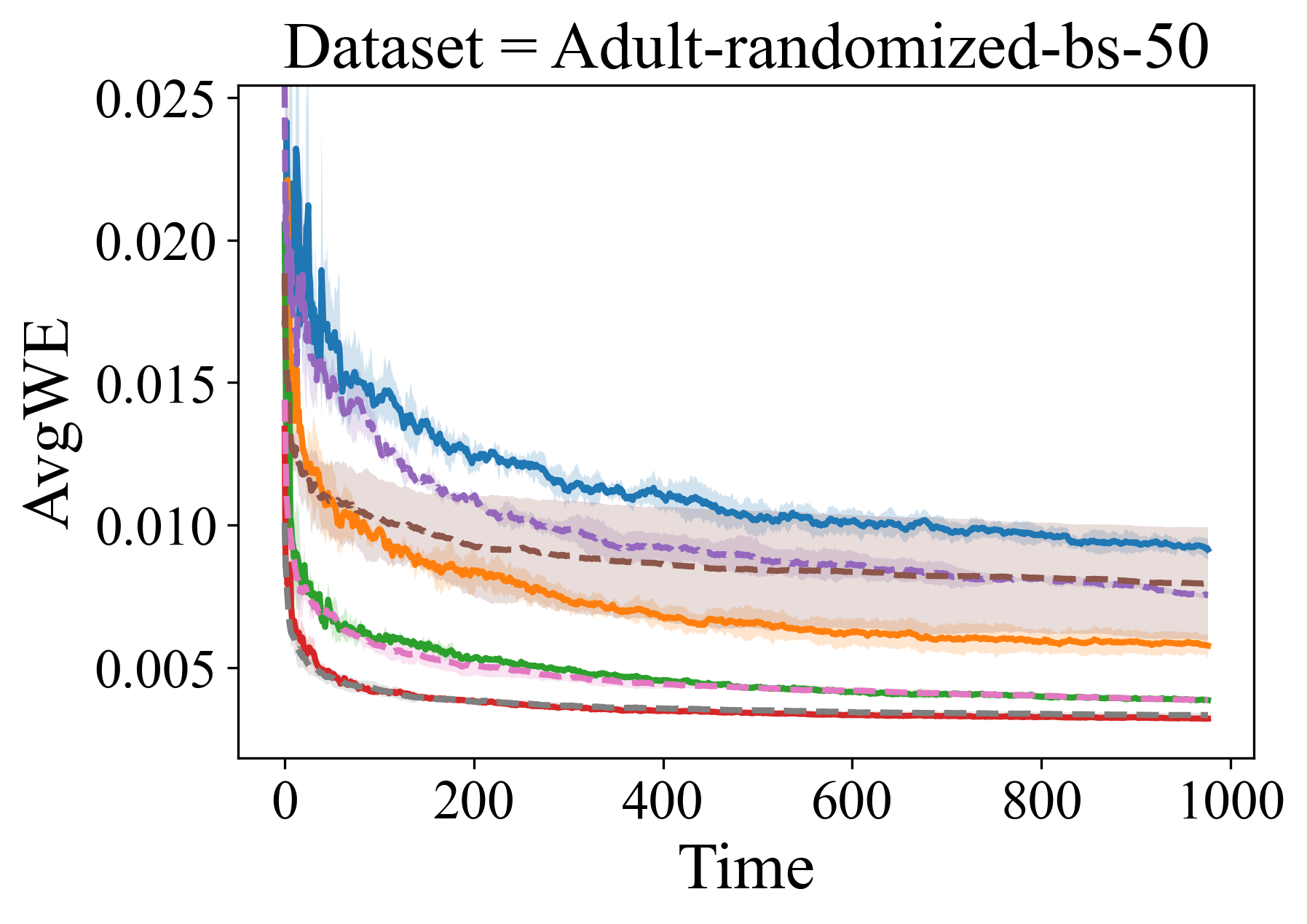}
        \caption{Average of workload errors}
     \end{subfigure}\hspace*{\fill}
     \begin{subfigure}[t]{0.23\linewidth}
         \centering
         \includegraphics[height=0.14\textheight]{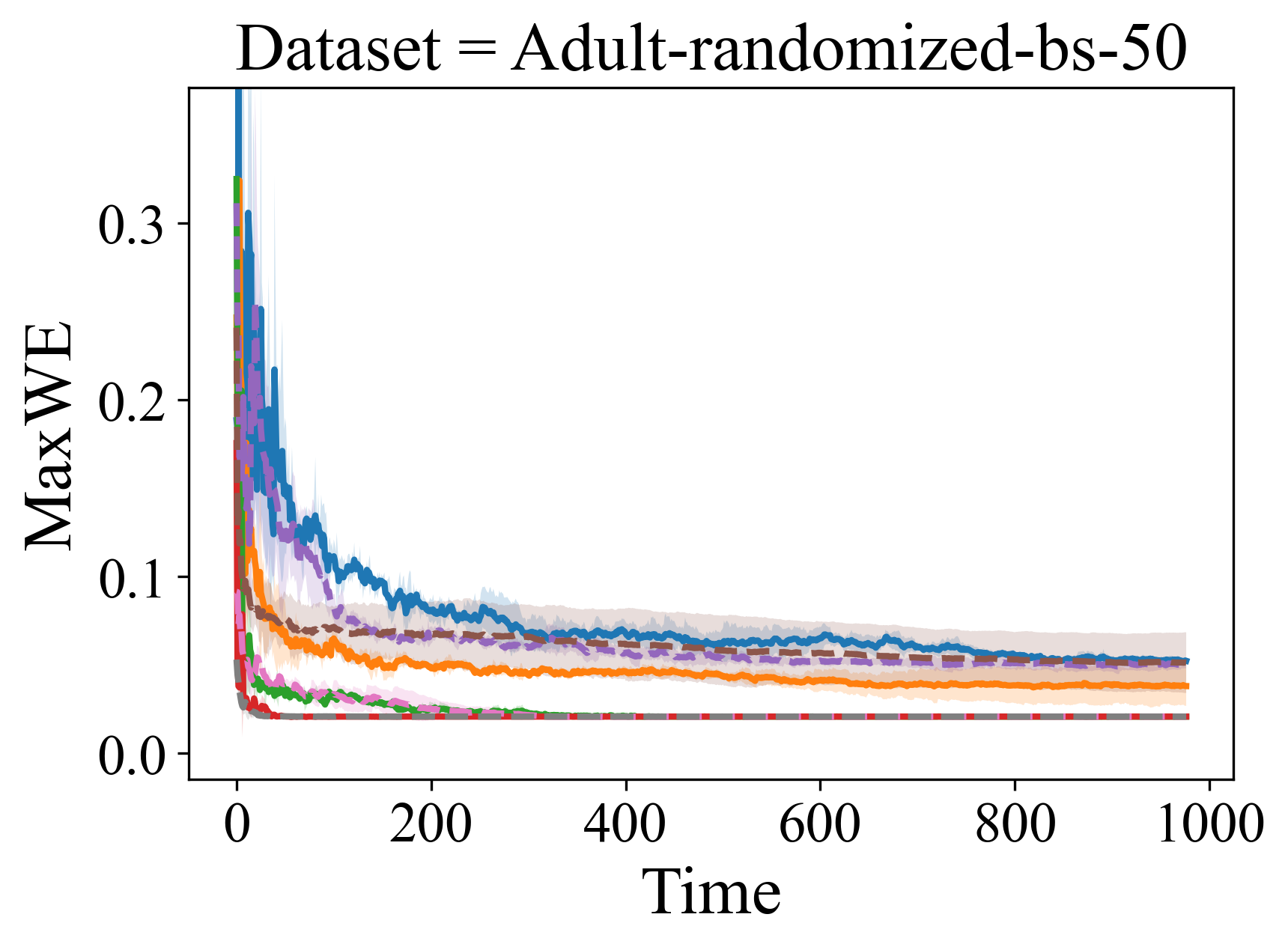}
        \caption{Maximum of workload errors}
     \end{subfigure}\hspace*{\fill}
     \begin{subfigure}[t]{0.23\linewidth}
         \centering
         \includegraphics[height=0.14\textheight]{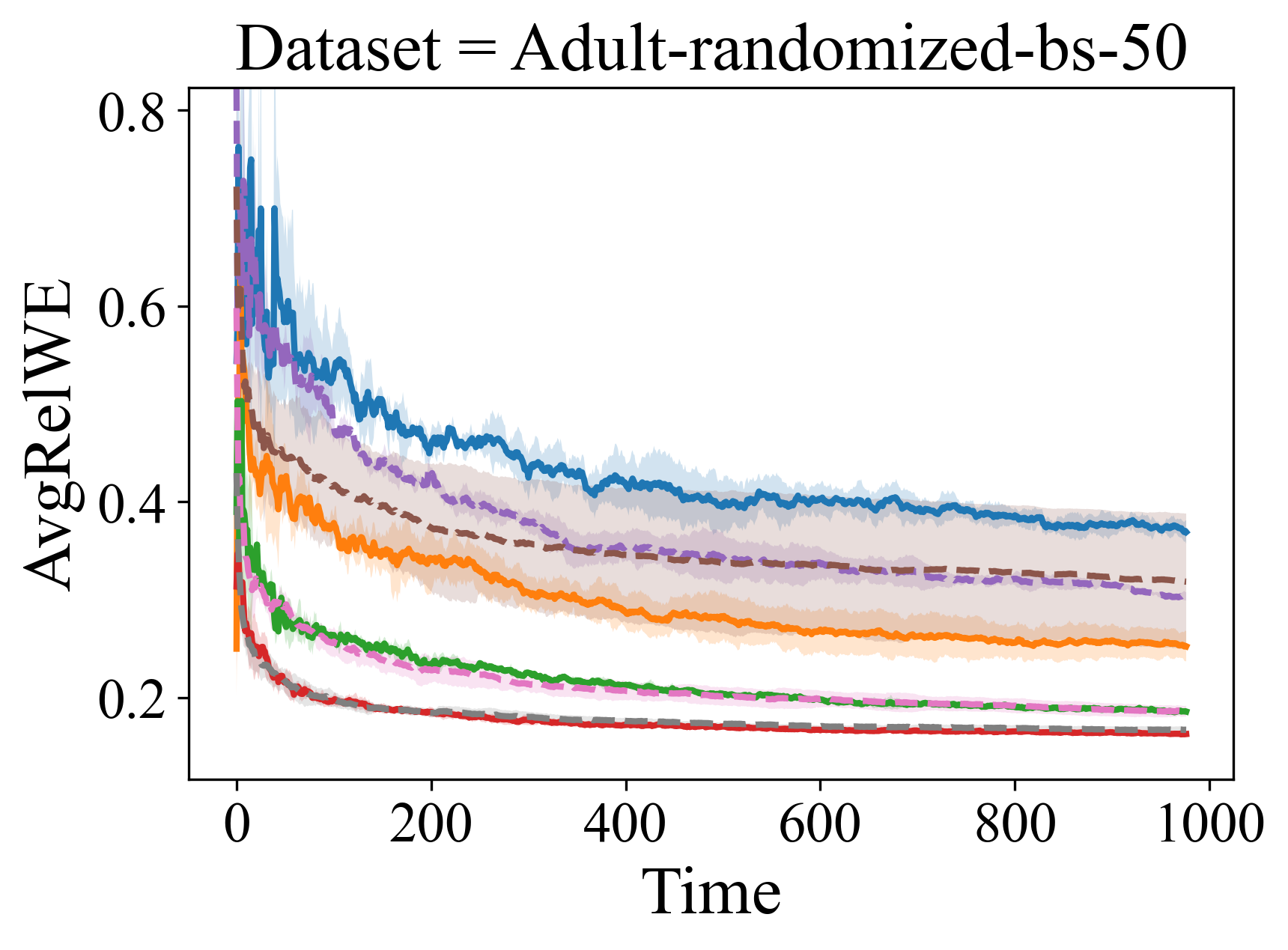}
        \caption{Average of relative workload errors}
     \end{subfigure}\hspace*{\fill}
     \begin{subfigure}[t]{0.23\linewidth}
         \centering
         \includegraphics[height=0.14\textheight]{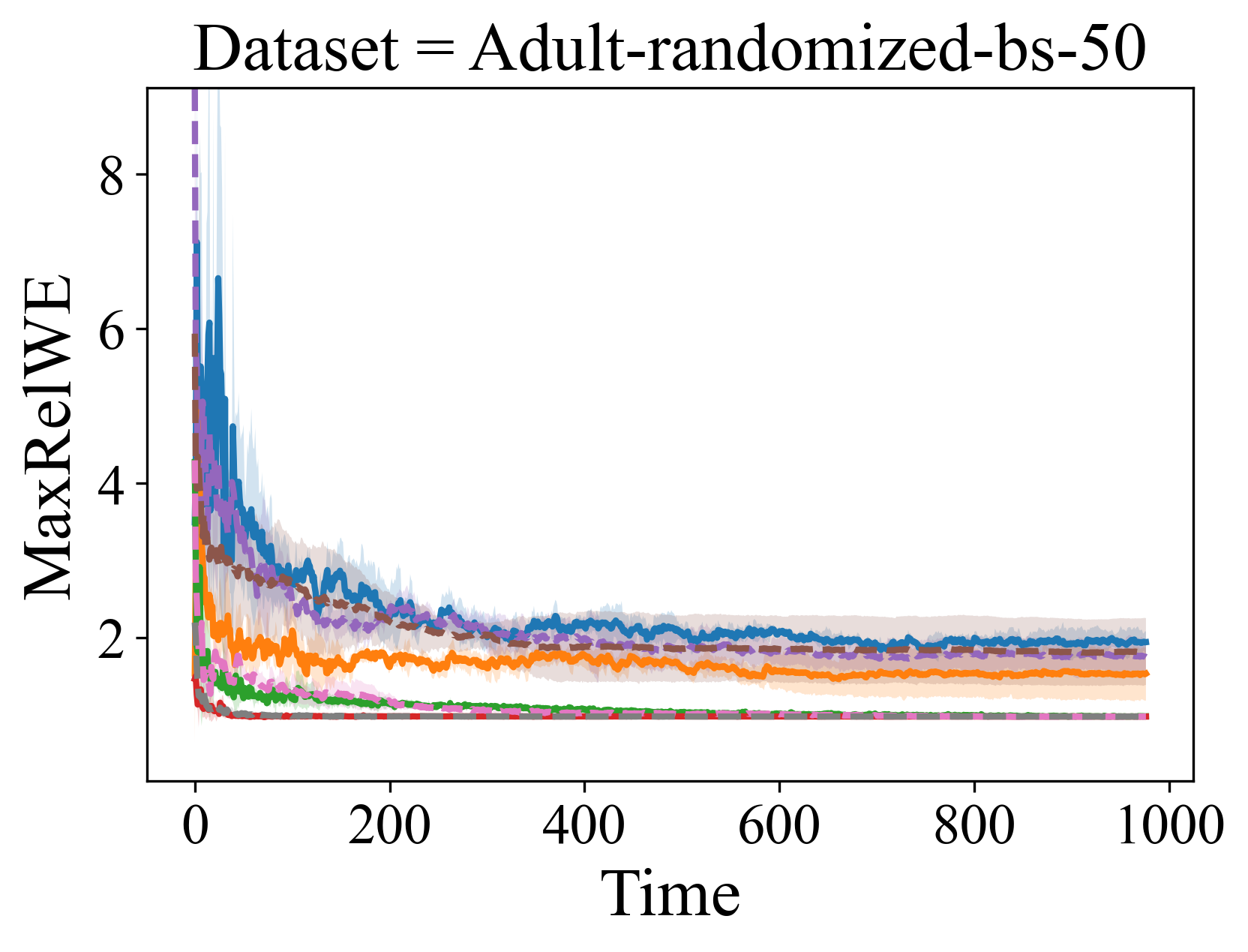}
        \caption{Maximum of relative workload errors}
     \end{subfigure}
    \hfill
     \begin{subfigure}[t]{\linewidth}
         \centering
         \includegraphics[height=0.03\textheight]{block_figures/legend.png}
     \end{subfigure}
    \caption{Metrics over time to compare the performance of simple and block counters for the Adult-randomized-bs-50 dataset.}
    \label{fig:Adult-randomized-bs-50_simple_vs_block}
\end{figure*}

 \begin{figure*}[ht]
     \centering
     \begin{subfigure}[t]{0.23\linewidth}
         \centering
         \includegraphics[height=0.14\textheight]{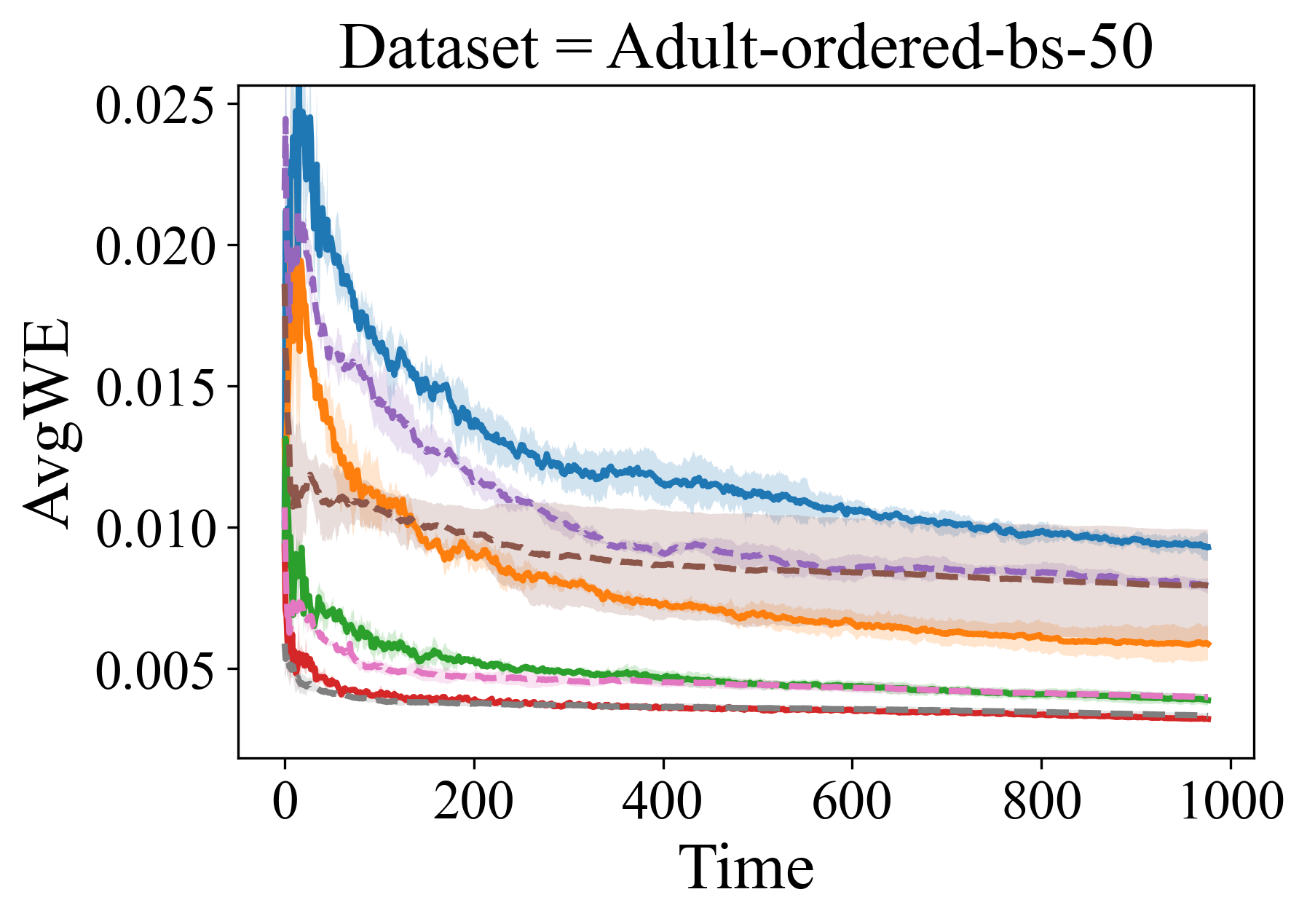}
        \caption{Average of workload errors}
     \end{subfigure}\hspace*{\fill}
     \begin{subfigure}[t]{0.23\linewidth}
         \centering
         \includegraphics[height=0.14\textheight]{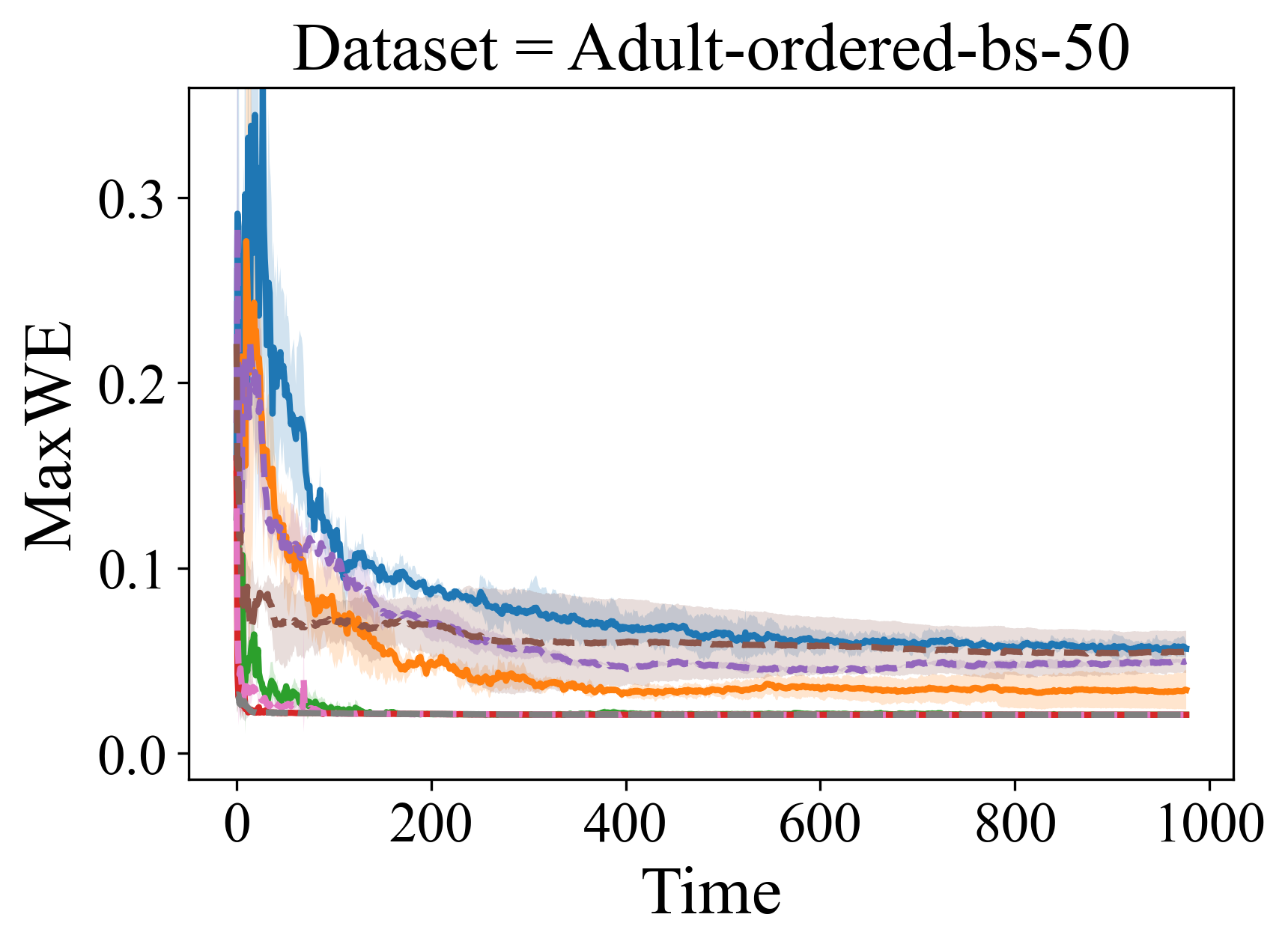}
        \caption{Maximum of workload errors}
     \end{subfigure}\hspace*{\fill}
     \begin{subfigure}[t]{0.23\linewidth}
         \centering
         \includegraphics[height=0.14\textheight]{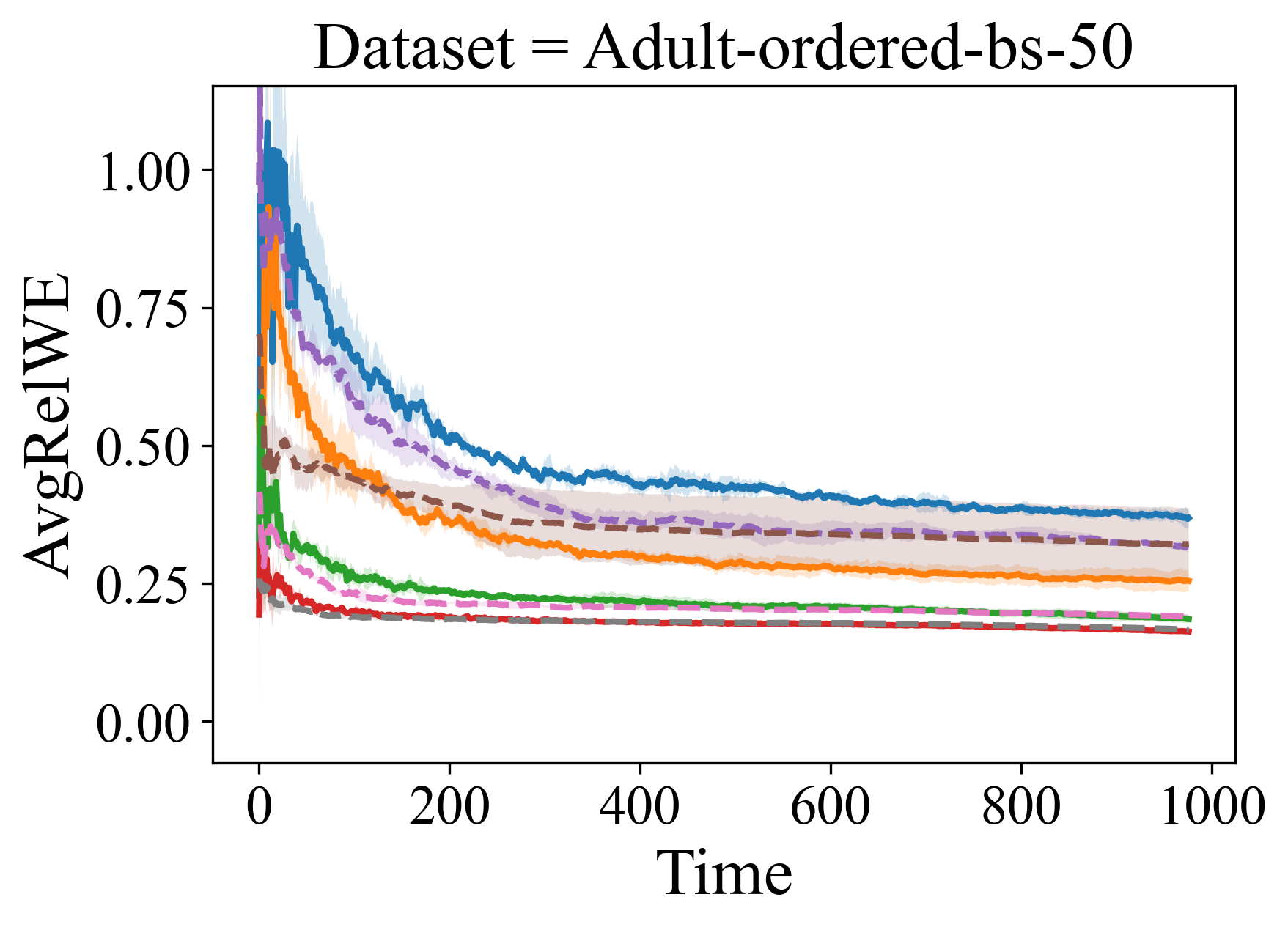}
        \caption{Average of relative workload errors}
     \end{subfigure}\hspace*{\fill}
     \begin{subfigure}[t]{0.23\linewidth}
         \centering
         \includegraphics[height=0.14\textheight]{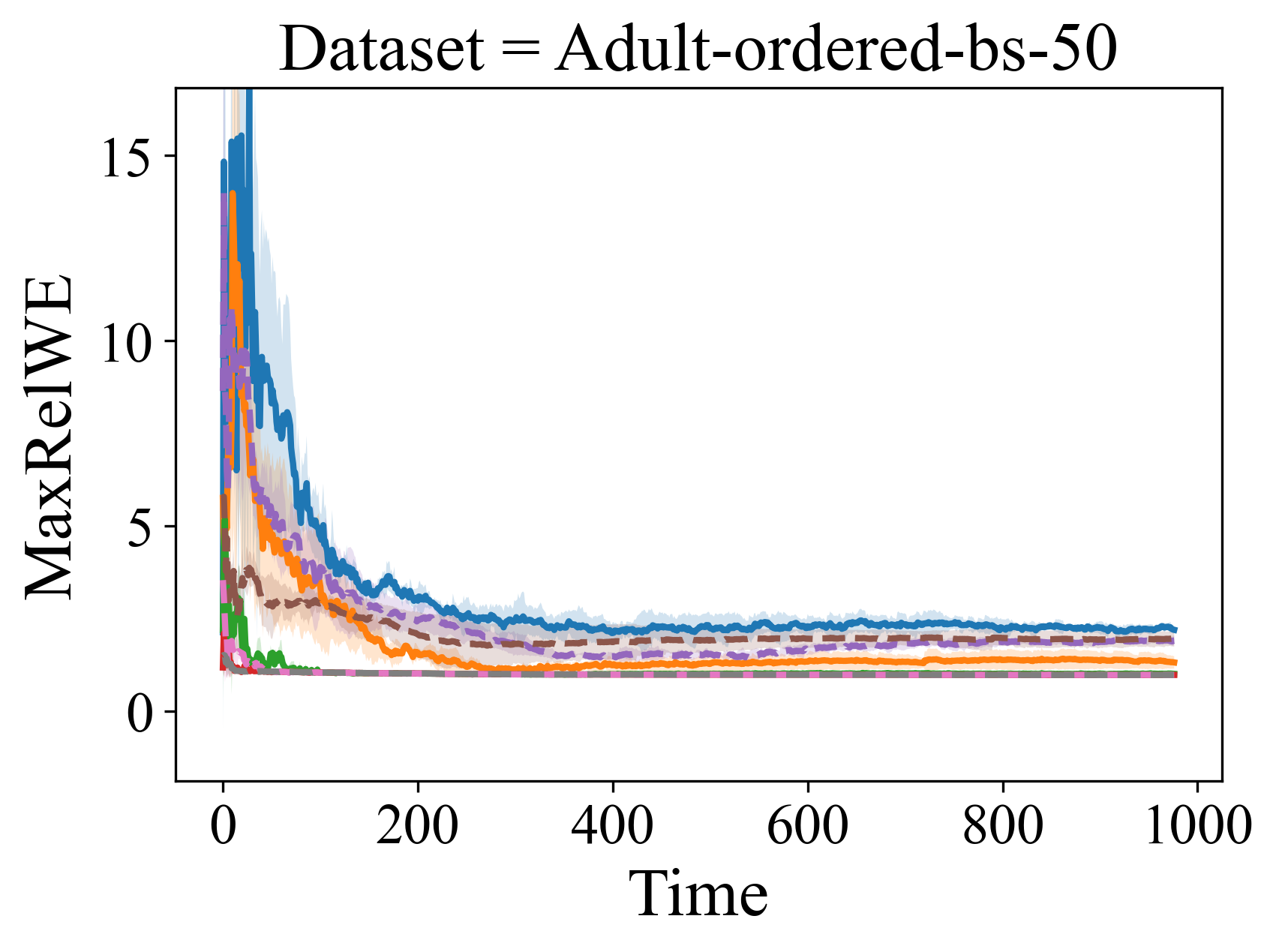}
        \caption{Maximum of relative workload errors}
     \end{subfigure}
    \hfill
     \begin{subfigure}[t]{\linewidth}
         \centering
         \includegraphics[height=0.03\textheight]{block_figures/legend.png}
     \end{subfigure}
    \caption{Metrics over time to compare the performance of simple and block counters for the Adult-ordered-bs-50 dataset.}
    \label{fig:Adult-ordered-bs_50_simple_vs_block}
\end{figure*}

We extend the Two-Level counter mechanism (also referred to as Block counter) due to \cite{Chan2010ContinualPrivateStats}  to unbounded streams. We present it formally in Algorithm~\ref{alg:unbounded_block_counter}. The idea is similar to how the bounded Binary Mechanism is extended to the unbounded Hybrid Mechanism in \cite{Chan2010ContinualPrivateStats}. As shown in \cite{Chan2010ContinualPrivateStats}, an optimal block size of the Bounded Block Counter for a stream of size $T$ is $\sqrt{T}$. The key idea is to partition the time dimension of the stream $f:\N\to\R$ into intervals of size $4, 9, 16, \ldots$ (that is perfect squares), and within each of the corresponding intervals, we use a bounded block counter of block size $2, 3, 4, \ldots$ respectively.

\begin{algorithm}[ht]
    \begin{algorithmic}[1]
        \State {\bf Input:} An input data stream $f:\N \to \R$, the privacy budget $\e$.
        \State {\bf Output:} A synthetic stream $g: \N \to \R$.
        \State Initialize partition size $T \leftarrow 4$.
        \State Initialize block size $B \leftarrow 2$.
        \State Last block value $\a_{lastBlock} \leftarrow 0$.
        \State True value within block $\a_{trueInBlock} \leftarrow 0$.
        \State Synthetic value within block $\a_{synthInBlock} \leftarrow 0$.
        \State Time when the last partition changed $t_{atPartition} \leftarrow 0$.
        \State Set $g(0)=0$.
        \For{$t=1, 2, \ldots$}
            \State Set $\d \leftarrow t-t_{atPartition}$.
            \State Update $\a_{trueInBlock} \leftarrow \a_{trueInBlock} + f(t)$.
            \If{$\d = kB$ for some $k \in Z$}
                \State Update $\a_{lastBlock} \leftarrow \a_{lastBlock} + \a_{tueInBlock} + \Lap \bp{ \frac{2}{\e} }$.
                \State Update $\a_{trueInBlock} \leftarrow 0$ and  $\a_{synthInBlock} \leftarrow 0$.
                \State Set $g(t) \leftarrow \a_{lastBlock}$.
                \If{$\d = T$}
                    \State Update $t_{atPartition} \leftarrow t$.
                    \State Update $B \leftarrow B+1$ and $T \leftarrow B^2$.
                \EndIf
            \Else
                \State Update $\a_{synthInBlock} \leftarrow \a_{synthInBlock} + f(t) + \Lap \bp{ \frac{2}{\e} }$.
                \State Set $g(t) \leftarrow \a_{lastBlock} + \a_{synthInBlock}$.
            \EndIf
            \State Release $g(t)$.
        \EndFor
    \end{algorithmic}
    \caption{Unbounded Block Counter}
    \label{alg:unbounded_block_counter}
\end{algorithm}

\begin{theorem}[Privacy of unbounded block counter]
    The unbounded block counter, as presented in Algorithm~\ref{alg:unbounded_block_counter}, satisfies $\e$-differential privacy.
\end{theorem}
\begin{proof}
    Note, Algorithm~\ref{alg:unbounded_block_counter} is exactly the block counter algorithm, except the size of the block changes over time. However, the change in the block size is independent of the input data stream. Hence, similar to the Block counter, Algorithm~\ref{alg:unbounded_block_counter} is $\e$-differentially private. 
\end{proof}

\subsection{Results}

In this section, we present our results for empirical analysis of the proposed unbounded block counter (Algorithm~\ref{alg:unbounded_block_counter}) as compared to the simple counter when used as the subroutine $\AA_{Dataset}$ in Algorithm~\ref{alg:main}. Based on the evidence in \cite{kumar_algorithm_2024}, we know that the block counter performs better than the simple counter only after sufficiently large time $t$, hence we only use the datasets Eviction-weekly, Adult-ordered-bs-50, and Adult-randomized-bs-50 in our experiments. The length of time horizons for these datasets are $1409$, $977$, and $977$ respectively.

We present the findings of our experiments in Figures~\ref{fig:Eviction-weekly_simple_vs_block}, \ref{fig:Adult-randomized-bs-50_simple_vs_block}, and~\ref{fig:Adult-ordered-bs_50_simple_vs_block} which show various error metrics over time, analogous to Section~\ref{s:results}. We also provide a tabular view of these metrics in Table~\ref{tab:metrics_block_v_simple}. Let us first focus on the Adult dataset and privacy budget $\e\geq 1$. Using the block counter in Algorithm~\ref{alg:main}  is typically better than using the simple counter. However, for $\e=0.5$, we see that the simple counter performs better. The results of the experiments over the Eviction dataset do not yield a clear conclusion whether the block counter is better than the simple counter for any particular $\e$. We believe that the high variance of the block counter at the beginning of time, together with the selection error due to the Exponential mechanism, leads to such behavior.

%% file: appendices/app_acc_baseline.tex
\section{Accuracy of baseline algorithm}\label{s:appendix_baseline_accuracy}
\begin{proof}[Proof of Theorem~\ref{thm:streaming_mwem_acc}]
    The below analysis is similar to the analysis of (offline) MWEM algorithm due to \cite{hardt2012mwem}. Let us focus the analysis on iteration $l$ of time $t$.
    
    \textit{Selection error:} First, we will analyze the error in query selection at Step~\ref{lst:alg_smwem_selection}. Let $\maxerr_{t,l}$ denote the maximum possible absolute difference between values of any query in $Q$ as measured on $h_{t,l-1}$ and $\nabla f_t$, that is,
    \begin{equation}
        \maxerr_{t,l} = \max_{q \in Q} \abs{q(h_{t, l-1}) - q(f_t)}.
    \end{equation}
    
    % Similarly, let us define
    % \begin{equation}
    %     \maxerr_{t} = \max_{q \in Q} \abs{q(g_t) - q(f_t)}.
    % \end{equation}
    
    At the $l^{th}$ iteration at time $t$, we select the query with index $j$, where $j$ is a shorthand for $e_{t,l}$. By the utility of exponential mechanism (Theorem~\ref{thm:exp_mech_accuracy}) invocated with a privacy budget $\e/2k$ and sensitivity $1$, for any $\b>0$, we have,
    
    \begin{equation}\label{eq:smwem_selection_err_t_l}
        \prob{ \abs{q_j(h_{t,l-1}) - q_j(\nabla f_t)} \leq \maxerr_{t,l} - \frac{4k}{\e} \ln{\frac{|Q|}{\b}} } \leq \b.
    \end{equation}

    \textit{Additive error:} Let us now analyze the error due to the Laplace Mechanism at Step~\ref{lst:alg_smwem_measure}. Let $\adderr_{t,l}$ denote the additive error when measuring the query $q_{e_{t,l}}$, that is,
    \begin{equation}
        \adderr_{t,l} = \abs{ m_{t,l} - q_{e_{t,l}} (\nabla f_t) }.
    \end{equation}
    Again using $j$ as shorthand for $e_{t,l}$. By concentration of the Laplace random variable we have, that for a noise of scale $\frac{2k}{\e}$,
    \begin{equation}\label{eq:smwem_add_err_t_l}
        \prob{\abs{ m_{t,l} - q_{e_{t,l}} (\nabla f_t) } > \frac{2k}{\e}\log{\frac{1}{\b}}}=\b.
    \end{equation}
    
    \textit{Relative entropy:} Similar to \cite{hardt2012mwem} we rely on relative entropy to show improvement in each iteration by using the multiplicative weights algorithm. Let the relative entropy at the end of iteration $l$ at time $t$ be given as
    \begin{equation}
        \Psi_{t,l} = \frac{1}{|\nabla f_t|} \sum_{x \in \XX} \nabla f_t(x) \ln\bp{\frac{\nabla f_t(x)}{h_{t,l}(x)}}.
    \end{equation}
    Then we have the following relations,
    \begin{align}
        \Psi_{t,l} &\geq 0, \\
        \Psi_{0,0} &\leq ln\abs{\XX}, \\
         \Psi_{t,l-1}-\Psi_{t,l} &\geq \bp{\frac{q_{e_{t,l}}(h_{t,l}) - q_{e_{t,l}}(\nabla f_t)}{2|\nabla f_t|}}^2 - \bp{\frac{m_{t,l} - q_{e_{t,l}}(\nabla f_t)}{2|\nabla f_t|}}^2. \label{eq:relative_entropy_diff} 
    \end{align}
    Equation~\eqref{eq:relative_entropy_diff} can be derived as follows,
    \begin{align*}
        \Psi_{t,l-1} - \Psi_{t,l} &= \frac{1}{|\nabla f_t|} \sum_{x \in \XX} \nabla f_t(x) \ln\bp{\frac{h_{t,l}(x)}{h_{t,l-1}(x)}}
        = \frac{1}{|\nabla f_t|} \sum_{x \in \XX} \nabla f_t(x) \ln\bp{ \frac{ h_{t,l-1}(x) \cdot \exp\bp{q_{e_{t,l}}(x) \cdot \bp{ \frac{m_{t,l}-q_{e_{t,l}}(h_{t,l-1})}{2|\nabla f_t|} }} } { h_{t,l-1}(x) Z_{t,l} } },
    \end{align*}
    where $Z_{t,l} = \frac{1}{|\nabla f_t|} \sum_{x \in \XX} h_{t,l-1}(x) \exp\bp{q_{e_{t,l}}(x) \cdot \bp{ \frac{m_{t,l}-q_{e_{t,l}}(h_{t,l-1})}{2|\nabla f_t|} }}$ is the normalization constant. So,
    \begin{align*}
        \Psi_{t,l-1} - \Psi_{t,l}
        &= \frac{1}{|\nabla f_t|} \sum_{x \in \XX} \nabla f_t(x) \bp{ q_{e_{t,l}}(x) \cdot \bp{ \frac{m_{t,l}-q_{e_{t,l}}(h_{t,l-1})}{2|\nabla f_t|}}  - \ln Z_{t,l} }
        = \bp{ \frac{m_{t,l}-q_{e_{t,l}}(h_{t,l-1})}{2|\nabla f_t|^2}} q_{e_{t,l}}(\nabla f_t)  - \ln Z_{t,l}.
    \end{align*}
    Using $e^x \leq 1+x+x^2$ for all $|x|\leq 1$ and  $\abs{q_{e_{t,l}}(x) \cdot \frac{m_{t,l}-q_{e_{t,l}}(h_{t,l-1})}{2|\nabla f_t|}} \leq 1$, we have,
    \begin{align*}
        Z_{t,l}
        &= \frac{1}{|\nabla f_t|} \sum_{x \in \XX} h_{t,l-1}(x) \exp\bp{q_{e_{t,l}}(x) \cdot \bp{ \frac{m_{t,l}-q_{e_{t,l}}(h_{t,l-1})}{2|\nabla f_t|} }}
        \\
        &\leq \frac{1}{|\nabla f_t|} \sum_{x \in \XX} h_{t,l-1}(x) \bp{1 + q_{e_{t,l}}(x) \cdot \bp{ \frac{m_{t,l}-q_{e_{t,l}}(h_{t,l-1})}{2|\nabla f_t|} } + \bp{ q_{e_{t,l}}(x) \cdot \bp{ \frac{m_{t,l}-q_{e_{t,l}}(h_{t,l-1})}{2|\nabla f_t|} } }^2 }
        \\
        &\leq 1 + \bp{ \frac{m_{t,l}-q_{e_{t,l}}(h_{t,l-1})}{2|\nabla f_t|} }^2 + q_{e_{t,l}}(h_{t, l-1}) \cdot \bp{ \frac{m_{t,l}-q_{e_{t,l}}(h_{t,l-1})}{2|\nabla f_t|^2} }.
        \\
        \implies \ln Z_{t,l}
        &\leq \bp{ \frac{m_{t,l}-q_{e_{t,l}}(h_{t,l-1})}{2|\nabla f_t|} }^2 + q_{e_{t,l}}(h_{t, l-1}) \cdot \bp{ \frac{m_{t,l}-q_{e_{t,l}}(h_{t,l-1})}{2|\nabla f_t|^2} }.
    \end{align*}
    Using this in the entropy difference bound we have,
    \begin{align}
        \Psi_{t,l-1} - \Psi_{t,l}
        &\geq \bp{ \frac{m_{t,l}-q_{e_{t,l}}(h_{t,l-1})}{2|\nabla f_t|^2}} \bp{ q_{e_{t,l}}(\nabla f_t) -q_{e_{t,l}}(h_{t, l-1}) }  - \bp{ \frac{m_{t,l}-q_{e_{t,l}}(h_{t,l-1})}{2|\nabla f_t|} }^2 \nonumber
        \\
        &= \bp{ \frac{q_{e_{t,l}}(h_{t, l-1}) -q_{e_{t,l}}(\nabla f_t) }{2|\nabla f_t|} }^2 - \bp{ \frac{ m_{t,l} - q_{e_{t,l}} (\nabla f_t) }{2|\nabla f_t|}}^2.
    \end{align}

    \textit{Finally:} Suppose we are interested in error at time $T$. Let $\b>0$ be the failure probability at some time $t\in[T]$. Then, using Equations~\eqref{eq:smwem_selection_err_t_l}, \eqref{eq:smwem_add_err_t_l} and \eqref{eq:relative_entropy_diff} and a union bound over $l \in [k]$, with probability at least $1-\b$, for all $l \in [k]$ simultaneously,
    \begin{equation}
        \maxerr_{t,l} \leq \abs{q_{e_{t,l}}(h_{t,l-1}) - q_{e_{t,l}}(\nabla f_t)} + \frac{4k}{\e} \log{ \bp{\frac{2k|Q|}{\b}} },
    \end{equation}
    and,
    \begin{equation}
        \adderr_{t,l} = \abs{ m_{t,l} - q_{e_{t,l}} (\nabla f_t) } \leq \frac{2k}{\e}\log{ \bp{\frac{2k}{\b}} }.
    \end{equation}
    
    Combining the above two equations with Equation~\ref{eq:relative_entropy_diff} we have, that with probability at least $1-\b$,
    \begin{align*}
        \maxerr_{t,l}
        &\leq
        \bp{ 4|f_t|^2\bp{\Psi_{t,l-1} - \Psi_{t,l}} + \adderr_{t,l}^2 }^{1/2} + \frac{4k}{\e} \log{ \bp{\frac{2k|Q|}{\b}} }.
    \end{align*}
    Finally, we can bound the maximum error in approximating the differential dataset as,
    \begin{align*}
        \max_{q \in Q} \abs{q(\nabla g_t) - q(\nabla f_t)}
        &= \max_{q \in Q} \abs{ q\bp{ \avg_{l\in[k]} h_{t,l} } - q(\nabla f_t) }        
        \\
        &\leq \avg_{l\in[k]} \max_{q \in Q} \abs{ q(h_{t,l}) - q(\nabla f_t) }
        = \avg_{l\in[k]} \maxerr_{t,l}
        \\
        &\leq 
        \avg_{l\in[k]} \bp{ 4|\nabla f_t|^2\bp{\Psi_{t,l-1} - \Psi_{t,l}} + \adderr_{t,l}^2 }^{1/2}  + \frac{4k}{\e} \log{ \bp{\frac{2k|Q|}{\b}} }
        \\
        &=
        \bp{ \frac{4|\nabla f_t|^2}{k} \bp{\Psi_{t,0} - \Psi_{t,k}} + \adderr_{t,l}^2 }^{1/2}  + \frac{4k}{\e} \log{ \bp{\frac{2k|Q|}{\b}} }
        \\
        &\leq
        \bp{ \frac{4|\nabla f_t|^2}{k} \ln{|\XX|} + \adderr_{t,l}^2  }^{1/2}  + \frac{4k}{\e} \log{ \bp{\frac{2k|Q|}{\b}} }
        \\
        &\leq
        2|\nabla f_t|\sqrt{\frac{\ln{|\XX|}}{k}} + \frac{2k}{\e}\log{ \bp{\frac{2k}{\b}} } + \frac{4k}{\e} \log{ \bp{\frac{2k|Q|}{\b}} }.
    \end{align*}
Let us suppose we are interested in error at time $T\in\N$. Taking a union bound over time we have, that with probability at least $1-\b$, for all $t\leq T$ simultaneously,
\begin{align*}  
    \max_{q \in Q} \abs{q(g_t) - q(f_t)}
    &\leq \sum_{t=1}^{T} \max_{q \in Q} \abs{q(\nabla g_t) - q(\nabla f_t)}
    \\
    &\leq \sum_{t=1}^{T} \bp{ 
        2|\nabla f_t|\sqrt{\frac{\ln{|\XX|}}{k}} + \frac{2k}{\e}\log{ \bp{\frac{2kt}{\b}} } + \frac{4k}{\e} \log{ \bp{\frac{2kt|Q|}{\b}} }
    }
    \\
    &\leq 2|f_T|\sqrt{\frac{\ln{|\XX|}}{k}} + \frac{2k}{\e}\sum_{t=1}^{T} \bp{ \log{ \bp{\frac{2t|Q|}{\b}} } + 2 \log{ \bp{\frac{2t|Q|^2}{\b}} } }
    \\
    &\leq 2|f_T|\sqrt{\frac{\ln{|\XX|}}{k}} + \frac{6k}{\e}\sum_{t=1}^{T} \bp{ \log{ \bp{\frac{2t|Q|^{5/3}}{\b}} } }
    \\
    &\leq 2|f_T|\sqrt{\frac{\ln{|\XX|}}{k}} + \frac{6kT}{\e} \log{ \bp{\frac{2T|Q|^{5/3}}{\b}} }.
\end{align*}

Let us compare the upper bound to a function of the form $u(k) = \frac{a}{\sqrt{k}}+bk$, then we can optimize for the value of $k$ with $k_*=\bp{\frac{a}{2b}}^{2/3}$. This results in $u(k_*) = \bp{2^{1/3}+2^{-1/3}}a^{2/3}b^{1/3}$. Using this optimal value in our upper bound so far, we have,
\begin{align*}  
    \max_{q \in Q} \abs{q(g_t) - q(f_t)}
    &\leq \bigo{ \bp{ |f_T|\sqrt{\ln{|\XX|}} }^{2/3} \bp{ \frac{T}{\e} \log{ \bp{\frac{T|Q|^{5/3}}{\b}} } }^{1/3} }
    \leq \bigo{ |f_T|^{2/3} \bp{ \frac{ \ln{|\XX|}\ln{|Q|} \bp{T\log{T}} }{\e\b} }^{1/3} }.
\end{align*}
\end{proof}

%% file: appendices/app_acc_proposed_algo.tex
\section{Accuracy analysis for the proposed method}
In this section, we try to find a bound on the accuracy of Algorithm~\ref{alg:main}. The analysis mostly follows what we did in Section~\ref{s:appendix_baseline_accuracy} and we use the notations $\maxerr_{t,l}$, $\adderr_{t,l}$ and $\Psi_{t,l}$ from that proof. Additionally, we use the notation $\maxerr_t$ to denote the maximum error comparing the input and synthetic stream snapshot at time $t$, that is
\begin{equation}
    \maxerr_t \coloneqq \max_{q \in Q} \abs{q ( g_{t} - f_{t}) )},
\end{equation}
note that there is only one index in the subscript here, unlike $\maxerr_{t,l}$.

Furthermore, we use the accuracy guarantees of the binary tree mechanism from \cite{Chan2010ContinualPrivateStats} as stated in Lemma~\ref{alg:unbounded_block_counter}.

\begin{lemma}\label{l:unbounded_binary_counter_utility}[Accuracy of unbounded binary tree counter]
For any $t\in\N$ and $\b>0$, an $\e$-differentially private unbounded binary tree counter is $\bp{ \bigo{ \frac{1}{\e} (\log{t})^{1.5} \log{\frac{1}{\b}} }, \b }$-accurate.
\end{lemma}

\textit{Selection error:} In Algorithm~\ref{alg:main}, we use an approximation of the true data in the exponential mechanism, such that at any time $t\in\N$, we use $\nabla f_t + g_{t-1}$ instead of $f_t$ for true data. This introduces bias which can be analyzed as,
\begin{align*}
    \max_{q \in Q} \abs{q (h_{t,l-1}) - q (\nabla f_{t} - g_{t-1}) )}
    &= \max_{q \in Q} \abs{q (h_{t,l-1} - f_t) - q ( g_{t-1} - f_{t-1}) )}
    \\
    &\geq \max_{q \in Q} \abs{ \abs{q (h_{t,l-1} - f_t)} - \abs{q ( g_{t-1} - f_{t-1}) )} }
    \\
    &\geq \max_{q \in Q} \abs{q (h_{t,l-1} - f_t)} - \max_{q \in Q} \abs{q ( g_{t-1} - f_{t-1}) )}
    \\
    &= \maxerr_{t,l} - \maxerr_{t-1}.
\end{align*}

At the $l^{th}$ iteration at time $t$, we select the query with index $i$, where $i$ is a shorthand for $e_{t,l}$. By the utility of the Exponential mechanism (Theorem~\ref{thm:exp_mech_accuracy}) invocated with a privacy budget $\e/2k$ and sensitivity $1$, we have,

\begin{align*}
    \prob{ \abs{q_i(h_{t,l-1}) - q_i(\nabla f_t + g_{t-1})} \leq \max_{q\in Q} \abs{q(h_{t,l-1}) - q (\nabla f_t + g_{t-1})} - \frac{4k}{\e} \ln{\frac{|Q|}{\b}} }
    &\leq \b
    \\
    \prob{ \abs{q_i(h_{t,l-1}) - q_i(\nabla f_t + g_{t-1})} \leq \maxerr_{t,l} - \maxerr_{t-1} - \frac{4k}{\e} \ln{\frac{|Q|}{\b}} }
    &\leq \b
    \\
    \prob{ \abs{q_i(h_{t,l-1}-f_t) + q_i(f_{t-1} - g_{t-1})} \leq \maxerr_{t,l} - \maxerr_{t-1} - \frac{4k}{\e} \ln{\frac{|Q|}{\b}} }
    &\leq \b
    \\
    \prob{ \abs{q_i(h_{t,l-1}-f_t) } + \max_{q \in Q} \abs{ q(f_{t-1} - g_{t-1})} \leq \maxerr_{t,l} - \maxerr_{t-1} - \frac{4k}{\e} \ln{\frac{|Q|}{\b}} }
    &\leq \b
    \\
    \prob{ \abs{q_i(h_{t,l-1}-f_t) } + \maxerr_{t-1} \leq \maxerr_{t,l} - \maxerr_{t-1} - \frac{4k}{\e} \ln{\frac{|Q|}{\b}} }
    &\leq \b,
\end{align*}
which results in,
\begin{equation} \label{eq:acc_proof_select_error}
    \prob{ \abs{q_i(h_{t,l-1}-f_t) } \leq \maxerr_{t,l} -2\maxerr_{t-1} - \frac{4k}{\e} \ln{\frac{|Q|}{\b}} }
    \leq \b.
\end{equation}

\textit{Additive error:} Using $m_{t,l} = C_i(t)+r_i(t)$, additive error becomes,
\begin{equation}
    \adderr_{t,l} = \abs{ C_i(t)+r_i(t) - q_i (f_t) }.
\end{equation}
Let $N_i(t)\subseteq N$ be the times when query index $i$ is selected by the exponential mechanism at or before time $t$. Let $\bar N_i(t) = [t]\setminus N_i(t)$. Then,
\begin{align*}
    \adderr_{t,l}
    & \leq \abs{ C_i(t) - q_i\bp{ f_{N_i(t)} } } + \abs{ r_i(t) - q_i\bp{ f_{\bar N_i(t)} } }
    \\
    & \leq \abs{ C_i(t) - q_i\bp{ f_{N_i(t)} } } + \abs{ q_i(g_{t-1}) - C_i(t-1) - q_i\bp{ f_{\bar N_i(t)} } }
    \\
    & \leq \abs{ C_i(t) - q_i\bp{ f_{N_i(t)} } } + \abs{ q_i(g_{t-1}) - C_i(t-1) - q_i\bp{ f_{\bar N_i(t)} } }
    \\
    & \leq \abs{ C_i(t) - q_i\bp{ f_{N_i(t)} } } + \abs{ q_i(g_{t-1}) - q_i(f_{t-1}) } + \abs{ C_i(t-1) - q_i \bp{f_{N_i(t-1)}} }
\end{align*}

By Lemma~\ref{l:unbounded_binary_counter_utility}, we have,
\begin{equation} \label{eq:acc_proof_bt_error}
    \prob{ \abs{ C_i(t) - q_i\bp{ f_{N_i(t)} } } \geq \frac{c}{\e} \ln{\bp{\frac{1}{\beta}}} \bp{\ln{\abs{N_i(t)}}}^{3/2} } \leq \beta,
\end{equation}
for some constant $c$.

\textit{Overall:}
Then, using a union bound and Equations~\eqref{eq:acc_proof_select_error} and \eqref{eq:acc_proof_bt_error}, with probability at least $1-\b$, for all $l \in [k]$ and $t \in [T]$ simultaneously,
\begin{equation}
    \maxerr_{t,l} \leq \abs{q_i(h_{t,l-1}-f_t) } + 2\maxerr_{t-1} + \frac{4k}{\e} \ln\bp{\frac{2kT|Q|}{\b}},
\end{equation}
and
\begin{equation}
    \adderr_{t,l} \leq \frac{c}{\e} \ln{\bp{\frac{2kT}{\beta}}} \bp{ \bp{\ln{\abs{N_i(t)}}}^{3/2} + \bp{\ln{\abs{N_i(t-1)}}}^{3/2} } + \maxerr_{t-1}.
\end{equation}

\textit{Relative entropy}
Recall that Equation~\eqref{eq:relative_entropy_diff} from the proof of Theorem~\ref{thm:streaming_mwem_acc} states, 
$$
     \Psi_{t,l-1}-\Psi_{t,l} \geq \bp{\frac{q_{e_{t,l}}(h_{t,l}) - q_{e_{t,l}}(f_t)}{2|f_t|}}^2 - \bp{\frac{m_{t,l} - q_{e_{t,l}}(f_t)}{2|f_t|}}^2.
$$

Combining the equations we have,
\begin{align*}
    \maxerr_{t,l}
    &\leq
    \bp{ 4|f_t|^2\bp{\Psi_{t,l-1} - \Psi_{t,l}} + \adderr_{t,l}^2 }^{1/2} + 2\maxerr_{t-1} + \frac{4k}{\e} \ln\bp{\frac{2kT|Q|}{\b}}.
\end{align*}

Then, similar to the proof of Theorem~\ref{thm:streaming_mwem_acc}, for any $t\in[T]$, we have, 
\begin{align*}
    \max_{q \in Q} \abs{q(g_t) - q(f_t)}
    \leq
    2|f_t|\sqrt{\frac{\ln{|\XX|}}{k}} + \adderr_{t,l}  + 2\maxerr_{t-1} + \frac{4k}{\e} \ln\bp{\frac{2kT|Q|}{\b}}.
\end{align*}

This results in the following recursive relation
\begin{equation} \label{eq:acc_proof_recursive_rel}
    \max_{q \in Q} \abs{q(g_t) - q(f_t)}
    \leq
    2|f_t|\sqrt{\frac{\ln{|\XX|}}{k}} + \frac{2c}{\e} \ln{\bp{\frac{2kT}{\beta}}} \bp{\ln{t}}^{3/2} + \frac{4k}{\e} \ln\bp{\frac{2kT|Q|}{\b}} + 3\maxerr_{t-1}.
\end{equation}

We can simplify the above recursive relation such that at time $T$ we have,
$$
\max_{q \in Q} \abs{q(g_T) - q(f_T)} \leq \sum_{t=1}^T 3^{T-t} \bp{ 2|f_t|\sqrt{\frac{\ln{|\XX|}}{k}} + \frac{2c}{\e} \ln{\bp{\frac{2kT}{\beta}}} \bp{\ln{t}}^{3/2} + \frac{4k}{\e} \ln\bp{\frac{2kT|Q|}{\b}} }.
$$

Note that the above bound has a term exponential in time and thus is not better than what we had for the accuracy of StreamingMWEM in Theorem~\ref{thm:streaming_mwem_acc}. However, the results of our empirical experiments (Section ~\ref{s:results}) suggest that the proposed method outperforms the StreamingMWEM. The reason that we do not have a better bound in the theoretical proof is that our algorithm depends on the output of the previous time step, which results in the worst-case recursive relation as mentioned in Equation~\eqref{eq:acc_proof_recursive_rel}. Moreover, the theorem and proof do not exploit potential cancellations in the added noise and we conjecture that being able to utilize  cancellations should give an improved bound.